\documentclass[12pt]{extarticle}
\usepackage[utf8]{inputenc}
\usepackage[T1]{fontenc}

\usepackage{caption}
\oddsidemargin \evensidemargin
\usepackage{amssymb}
\usepackage{pmboxdraw}
\usepackage{amsmath}
\usepackage{bbm}
\usepackage{amsthm}
\usepackage{amssymb}
\usepackage[dvipsnames]{xcolor}
\usepackage{graphicx}
\usepackage{alphabeta}
\usepackage{tikz,pgfplots}
\usepackage{hyperref}
\usepackage[all]{xy}
\usepackage[greek,english]{babel}
\usepackage[export]{adjustbox}
\usepackage{hyperref}
\hypersetup{
    colorlinks=true,
    linkcolor=orange!80!black,
    filecolor=magenta,      
    urlcolor=RoyalBlue,
    citecolor=RoyalBlue,
}
\usepackage{enumitem}

\usepackage[finalizecache=false,frozencache=true]{minted}
\setminted[julia]{frame=lines,
rulecolor=\color{white!80!black},
fontsize=\small,
numbers=right,
numbersep=-5pt,
obeytabs=true,
tabsize=4}

\newtheorem{theorem}{Theorem}[section]
\newtheorem{theorem*}[theorem]{Theorem*}

\newtheorem{lemma}[theorem]{Lemma}
\newtheorem{corollary}[theorem]{Corollary}
\newtheorem{proposition}[theorem]{Proposition}
\newtheorem{assumption}{Assumption}

\theoremstyle{definition}
\newtheorem{examplex}[theorem]{Example}
\newenvironment{example}
{\pushQED{\qed}\examplex}
{\popQED\endexamplex}
\newtheorem{definition}[theorem]{Definition}
\newtheorem{remark}[theorem]{Remark}

\theoremstyle{remark}

\def\sol{\mathrm{Sol}}

\title{\bf Four Lectures on Euler Integrals}
\author{Saiei-Jaeyeong Matsubara-Heo, Sebastian Mizera, Simon Telen}
\date{\today}

\begin{document}

\maketitle

\begin{abstract}
    These lecture notes provide a self-contained introduction to Euler integrals, which are frequently encountered in applications. In particle physics, they arise as Feynman integrals or string amplitudes. Our four selected topics demonstrate the diverse mathematical techniques involved in the study of Euler integrals, including polyhedral geometry, very affine varieties, differential equations, and computational algebra.
\end{abstract}

\hypersetup{linkcolor=black}
\tableofcontents
\hypersetup{linkcolor=orange!80!black}

\pagebreak

\addcontentsline{toc}{section}{Introduction}
\section*{Introduction} 
Consider $\ell$ Laurent polynomials $f_1, \ldots, f_\ell$ in $n$ variables $x = (x_1, \ldots, x_n)$ with complex coefficients. By an \emph{Euler integral}, we mean any integral of the following form: 
\begin{equation} \label{eq:integralintro}
\int_\Gamma \,  \frac{x_1^{\nu_1} \cdots x_n^{\nu_n}}{f_1^{s_1} \cdots f_\ell^{s_\ell}} \, \frac{{\rm d} x_1}{x_1} \wedge \cdots \wedge \frac{{\rm d} x_n}{x_n} \, = \, \int_\Gamma f^{-s} x^\nu \, \frac{{\rm d} x}{x} \, .
\end{equation}
The right-hand side is our shorthand notation. The first example is the Euler \emph{beta~function} 
\begin{equation} \label{eq:betafunc}
B(\nu, 1-s) \, = \, \int_0^1 \frac{x^{\nu}}{(1-x)^{s}} \, \frac{{\rm d} x}{x} \, \textcolor{RoyalBlue}{ = } \, \frac{\Gamma(\nu) \Gamma(1-s)}{\Gamma(\nu + 1-s)}, \quad \text{where} \quad \Gamma(u) \, = \, \int_0^\infty t^{u-1} e^{-t} \, {\rm d} t 
\end{equation}
is the \emph{gamma function}. The equality in blue will be derived below. Such integrals have been called many different names, depending on the context in which they are studied. They were called \emph{generalized Euler integrals} by Gelfand, Kapranov and Zelevinsky \cite{gelfand1990generalized}. This was motivated by Euler's integral representation of Gauss' hypergeometric function. In fact, the integral \eqref{eq:integralintro} represents a \emph{generalized} hypergeometric function, and the name \emph{hypergeometric integral} has appeared in the literature as well \cite{aomoto2011theory}. When $s_1 = \cdots = s_\ell = 1$ and $\Gamma = \mathbb{R}^n_+$, our integral is a function of $\nu$ called the \emph{Mellin transform} of $(f_1 \cdots f_\ell)^{-1}$ \cite{nilsson2013mellin}. This lead the authors of \cite{berkesch2014euler} to use the name \emph{Euler-Mellin integrals} for general $s$ and $\Gamma = \mathbb{R}^n_+$.
Seminal contributions like \cite{aomoto1975equations,gelfand1986general} justify the name \emph{Aomoto-Gelfand integrals}.
In physics, \emph{Feynman integrals} in quantum field theory and  \emph{string amplitudes} in superstring theory take the form \eqref{eq:integralintro} for particular choices of $f_i$. We elaborate on these specific polynomials below. In Bayesian statistics, Euler integrals appear as \emph{marginal likelihood integrals} \cite{borinsky2023bayesian}. 
In our title, we chose to use \emph{Euler integrals} as an umbrella term for all these instances of \eqref{eq:integralintro}. 

In different sections, we will view the integral \eqref{eq:integralintro} as a function of different sets of parameters. For instance, in Section~\ref{sec:1}, we will fix $\Gamma = \mathbb{R}^n_{+}$ and think of \eqref{eq:integralintro} as a function of $s$ and $\nu$. On the other hand, in Section \ref{sec:3}, we think of the integrand as an element of a cohomology vector space. Hence, the integral gives a linear function which sends $\Gamma$ to \eqref{eq:integralintro}. We will also consider the case where the coefficients of $f_i$ depend on some parameters $z$. In this case our integral is a function of $z$ satisfying some interesting differential equations, see Section \ref{sec:4}. 

As mentioned above, Euler integrals appear in particle physics. The first important example comes from quantum field theory, where \emph{Feynman integrals} are used to describe particle scattering processes. For a complete introduction to the subject, we refer to the recent book by Weinzierl \cite{weinzierl2022feynman}. In the \emph{Lee-Pomeransky} representation \cite{lee2013critical}, up to a prefactor involving gamma functions in $s, \nu$, the Feynman integral of a graph $G$ takes the form 
\begin{equation} \label{eq:feynman}
{\cal I}_G \, = \, \int_{\mathbb{R}^n_+} \frac{x^\nu}{({\cal U}_G + {\cal F}_G)^s} \, \frac{{\rm d}x}{x},
\end{equation}
where $n$ is the number of \emph{internal edges} of $G$, and ${\cal U}_G, {\cal F}_G$ are the first and second \emph{Symanzik polynomials} associated to the graph. We illustrate this with one of our running examples. 

\begin{example}\label{ex:triangle}
Consider the triangle diagram $G$ in \eqref{eq:trianglediag} with three massless internal edges. The internal edges carry variables $(x_1, x_2, x_3)$. The three external (open) edges attached to each vertex carry the \emph{kinematic parameters} $(t_1, t_2, t_3)$:
\begin{equation} \label{eq:trianglediag}
\begin{gathered}
\begin{tikzpicture}[scale=1,line width=1.2pt]
\coordinate (t1) at (0,0);
\coordinate (t2) at (0.5,0.866);
\coordinate (t3) at (1,0);
\draw[] (t1) --  (t2) node[midway, left] {\footnotesize $x_3$}  -- (t3) node[midway, right] {\footnotesize $x_1$} -- (t1) node[midway, below] {\footnotesize $x_2$};
\draw[] (t1) -- ++(210:0.4) node[left] {\footnotesize $t_1$};
\draw[] (t2) -- ++(90:0.4) node[above] {\footnotesize $t_2$};
\draw[] (t3) -- ++(-30:0.4) node[right] {\footnotesize $t_3$};
\end{tikzpicture}
\end{gathered}
\end{equation}
The polynomial $\mathcal{U}_G$ is the sum over spanning trees of $G$, with each term given by the $x_i$'s not present in the tree:
\begin{equation}
\mathcal{U}_G =
\begin{gathered}
\begin{tikzpicture}[scale=0.5,line width=1.2pt]
\coordinate (t1) at (0,0);
\coordinate (t2) at (0.5,0.866);
\coordinate (t3) at (1,0);
\draw[lightgray!50] (t1) -- (t2) -- (t3) -- (t1);
\draw[] (t3) -- (t1) -- (t2);
\draw[] (t1) -- ++(210:0.4);
\draw[] (t2) -- ++(90:0.4);
\draw[] (t3) -- ++(-30:0.4);
\end{tikzpicture}
\end{gathered}
+
\begin{gathered}
\begin{tikzpicture}[scale=0.5,line width=1.2pt]
\coordinate (t1) at (0,0);
\coordinate (t2) at (0.5,0.866);
\coordinate (t3) at (1,0);
\draw[lightgray!50] (t1) -- (t2) -- (t3) -- (t1);
\draw[] (t1) -- (t2) -- (t3);
\draw[] (t1) -- ++(210:0.4);
\draw[] (t2) -- ++(90:0.4);
\draw[] (t3) -- ++(-30:0.4);
\end{tikzpicture}
\end{gathered}
+
\begin{gathered}
\begin{tikzpicture}[scale=0.5,line width=1.2pt]
\coordinate (t1) at (0,0);
\coordinate (t2) at (0.5,0.866);
\coordinate (t3) at (1,0);
\draw[lightgray!50] (t1) -- (t2) -- (t3) -- (t1);
\draw[] (t2) -- (t3) -- (t1);
\draw[] (t1) -- ++(210:0.4);
\draw[] (t2) -- ++(90:0.4);
\draw[] (t3) -- ++(-30:0.4);
\end{tikzpicture}
\end{gathered}
= x_1 + x_2 + x_3.
\end{equation}
The $\mathcal{F}_G$ polynomial is given similarly as a sum of spanning two-forests (disjoint unions of two trees), each weighted with minus the corresponding kinematic variable:
\begin{equation}
\mathcal{F}_G =
\begin{gathered}
\begin{tikzpicture}[scale=0.5,line width=1.2pt]
\coordinate (t1) at (0,0);
\coordinate (t2) at (0.5,0.866);
\coordinate (t3) at (1,0);
\draw[lightgray!50] (t1) -- (t2) -- (t3) -- (t1);
\draw[] (t2) -- (t3);
\draw[] (t1) -- ++(210:0.4);
\draw[] (t2) -- ++(90:0.4);
\draw[] (t3) -- ++(-30:0.4);
\end{tikzpicture}
\end{gathered}
+
\begin{gathered}
\begin{tikzpicture}[scale=0.5,line width=1.2pt]
\coordinate (t1) at (0,0);
\coordinate (t2) at (0.5,0.866);
\coordinate (t3) at (1,0);
\draw[lightgray!50] (t1) -- (t2) -- (t3) -- (t1);
\draw[] (t3) -- (t1);
\draw[] (t1) -- ++(210:0.4);
\draw[] (t2) -- ++(90:0.4);
\draw[] (t3) -- ++(-30:0.4);
\end{tikzpicture}
\end{gathered}
+
\begin{gathered}
\begin{tikzpicture}[scale=0.5,line width=1.2pt]
\coordinate (t1) at (0,0);
\coordinate (t2) at (0.5,0.866);
\coordinate (t3) at (1,0);
\draw[lightgray!50] (t1) -- (t2) -- (t3) -- (t1);
\draw[] (t1) -- (t2);
\draw[] (t1) -- ++(210:0.4);
\draw[] (t2) -- ++(90:0.4);
\draw[] (t3) -- ++(-30:0.4);
\end{tikzpicture}
\end{gathered}
=
-t_1 \cdot x_2 x_3 - t_2 \cdot x_3 x_1 - t_3 \cdot x_1 x_2 .
\end{equation}
The associated integral is given by
\begin{equation} \label{eq:Itriangle}
{\cal I}_{G} \, = \, \int_{\mathbb{R}^3_+} \frac{x_1^{\nu_1} x_2^{\nu_2} x_3^{\nu_3}}{( x_1 + x_2 + x_3 - t_1 \cdot x_2 x_3 - t_2 \cdot x_1 x_3 - t_3 \cdot x_1 x_2)^s} \, \frac{{\rm d}x_1{\rm d}x_2{\rm d}x_3}{x_1x_2x_3}.
\end{equation}
The exponents $\nu_i$ are typically taken to be non-negative integers and $s = \mathrm{D}/2$ is half the space-time dimension $\mathrm{D}$. It is often convenient to think of $(\nu_1, \nu_2, \nu_3)$ and $s$ as generic parameters, which is referred to as \emph{analytic} and \emph{dimensional regularization} respectively.
\end{example}

The second application of Euler integrals in physics comes from scattering amplitudes in string theory. Instead of \emph{particles}, one computes the probability of \emph{strings} interacting with each other. See \cite{Mafra:2022wml} for a comprehensive review. This offers a nice immediate connection to algebraic geometry: the integration is on the moduli space ${\cal M}_{0,m}$ of genus zero curves with $m$ marked points. Equivalently, this is the space of configurations of $m$ distinct points on $\mathbb{P}^1$ up to its automorphisms ${\rm PSL}(2)$. We can represent these points as the columns of a $2 \times m$ matrix with nonzero $2 \times 2$ minors. Two such matrices $M_1, M_2$ represent equivalent configurations if there is an invertible $2 \times 2$ matrix $T$ and an $n \times n$ invertible diagonal matrix $D$ such that $T \cdot M_1 \cdot D = M_2$. We can use the action of $T$ and $D$ to fix 3 out of $m$ points, leaving $n = m-3$ degrees of freedom. Following \cite[Eq.~(1.5)]{arkani2021stringy}, we write a point of ${\cal M}_{0,m}$ as
\begin{equation} \label{eq:positiveparam}
    M \, = \,  \begin{pmatrix}
1 & 1 & 1 & 1 & \cdots & 1 & 0 \\
0 & 1 & 1+x_1 & 1+x_1+x_2 & \cdots & 1+x_1+ \cdots +x_n & 1
\end{pmatrix},
\end{equation}
where $n = m-3$ and the $2 \times 2$ minors $f_{ij} = M_{1i}M_{2j}-M_{1j}M_{2i}, i < j$ are nonzero. The genus zero contribution to the $m$-point string amplitude is given by an Euler integral depending on an extra parameter $\alpha'$:
\begin{equation} \label{eq:stringamplitude}
{\cal I}_{m} = (\alpha')^n \cdot \int_{{\cal M}_{0,m}^+} \frac{x_1^{\alpha' \nu_1} \cdots x_n^{\alpha' \nu_n}}{\prod_{1 < i+1<j<m} f_{ij}^{\alpha' s_{ij}}} \, \frac{{\rm d} x}{x}.
\end{equation}
The pairs $(i,j)$ excluded in the product in the denominator are those for which the minor $f_{ij}$ is either constant or one of the $x$-variables. The integration is over the \emph{positive part} ${\cal M}_{0,m}^+$ of ${\cal M}_{0,m}$, which is the subset of points $M$ satisfying $f_{ij} >0$, for all $1 \leq i < j \leq m$. Using the parameterization \eqref{eq:positiveparam}, one checks that this is $\mathbb{R}^n_+$.

There are two physically interesting limits: $\alpha' \rightarrow 0$ and $\alpha' \rightarrow \infty$. The first one is called the \emph{field theory limit} in which strings become particles, and the second is the \emph{high-energy limit}. We will see in Section \ref{sec:2} that both of them admit an elegant geometric description. 

\begin{example}[$m=4$]\label{ex:0.2}
The moduli space ${\cal M}_{0,4}$ has dimension 1. The four-point string amplitude \eqref{eq:stringamplitude} is ${\cal I}_4 = \alpha' \cdot B(\alpha' \nu, - \alpha' \nu + \alpha' s_{13})$, where $B$ is the beta function from \eqref{eq:betafunc}.
\end{example}

\begin{example}[$m=5$]
The matrix parameterizing $\mathcal{M}_{0,5}$ is
\begin{equation}
M \, = \,  \begin{pmatrix}
1 & 1 & 1 & 1 & 0 \\
0 & 1 & 1+x_1 & 1+x_1+x_2 & 1
\end{pmatrix},
\end{equation}
which has only $5$ ordered minors depending on the variables $(x_1,x_2)$:
\begin{equation}
f_{13} = 1 + x_1, \quad f_{14} = 1 + x_1 + x_2, \quad f_{23} = x_1, \quad f_{24} = x_1 + x_2, \quad f_{34} = x_2.
\end{equation}
The minors $f_{23}$ and $f_{34}$ are not included in the integrand of \eqref{eq:stringamplitude}, since they would only shift the exponents of $x_j^{\alpha' \nu_j}$. The five-point string amplitude is given by
\begin{equation} \label{eq:I5}
{\cal I}_{5} \, = \, (\alpha')^2 \cdot \int_{\mathbb{R}^2_+} \frac{x_1^{\alpha' \nu_1} x_2^{\alpha' \nu_2}}{(1+x_1)^{\alpha' s_{13}} (1 + x_1+x_2)^{\alpha' s_{14}} (x_1 + x_2)^{ \alpha' s_{24}} } \, \frac{{\rm d}x_1 {\rm d}x_2}{x_1x_2} .
\end{equation}
The parameters $(\nu_1, \nu_2, s_{13}, s_{14}, s_{24})$ describe momenta and angles of the $5$ strings involved in the scattering process.
\end{example}

Euler integrals have many other applications, including marginal likelihood integrals \cite{borinsky2023bayesian}, wave functions in cosmology \cite{Arkani-Hamed:2017fdk}, and correlation functions of conformal field theories \cite{Britto:2021prf,Dotsenko:1984nm}.

These notes present the basics on Euler integrals from different points of view. They provide a roadmap through the literature for a reader who is new to the subject. At the same time, we hope they serve as a helpful overview of important results for experts. Section \ref{sec:1} discusses convergence and meromorphic continuation, which leads us to study convex polytopes and polyhedral cones. Section \ref{sec:2} is about certain limits which are meaningful in physics applications. This brings in algebraic equations, very affine varieties and Euler characteristics. Section \ref{sec:3} develops the theory of (algebraic) twisted (co)homology on these very affine varieties. Section \ref{sec:4} identifies difference and differential equations satisfied by Euler integrals. Finally, Section \ref{sec:5} contains a list of open problems.

\section{Newton polytopes and convergence} \label{sec:1}
This section discusses convergence of the integral \eqref{eq:integralintro}, viewed as a function of the exponents $s,\nu$. The integration contour $\Gamma = \mathbb{R}^n_+$ is fixed throughout the section. We set 
\begin{equation} \label{eq:Eulermellin}
    {\cal I}(s,\nu) \, = \, \int_{\mathbb{R}^n_+} f^{-s} x^\nu \, \frac{{\rm d} x}{x}.
\end{equation}
To ensure that the integrand is finite on $\mathbb{R}^n_+$, we make the following assumption. 
\begin{assumption} \label{assum:poscoeffs}
The coefficients of $f_i$ are real, positive numbers. That is,
\begin{equation} \label{eq:fi} f_i \, = \, \sum_{\alpha \,  \in \, {\rm supp}(f_i)} c_{i,\alpha} \cdot x^\alpha, \quad i = 1, \ldots, \ell, 
\end{equation}
where $c_{i,\alpha} \in \mathbb{R}_+$, ${\rm supp}(f_i) \subset \mathbb{Z}^n$ is the \emph{support} of $f_i$ (see Definition \ref{def:newton}) and $x^\alpha = x_1^{\alpha_1} \cdots x_n^{\alpha_n}$. 
\end{assumption}
Before studying convergence, we should address \emph{how to evaluate} the integrand $f^{-s}x^\nu$. Since $s$ and $\nu$ are complex vectors, this function may be multi-valued. For instance, we have
\[ f_i(x)^{s_i} \, = \, \exp({s_i \, {\rm log} f_i(x)}), \]
and ${\rm log}$ is only defined up to translates by integer multiples of $2 \pi \sqrt{-1}$. When $s$ is not an integer, ${\rm exp}(s_i F) \neq {\rm exp}(s_i(F+ 2 \pi \sqrt{-1}k))$ for some integer $k$, i.e., there are multiple branches of $f_i(x)^{s_i}$. 
Assumption \ref{assum:poscoeffs} ensures that $f_i$ takes positive values on $\mathbb{R}^n_+$, so that there is precisely one positive branch of $\log f_i$ and $\log x_j$. In this section, our integrand is 
\[ f^{-s} x^{\nu} \, = \, \exp({-s_1 \log f_1 - \cdots - s_\ell \log f_\ell + \nu_1 \log x_1 + \cdots + \nu_n \log x_n}), \]
where the unique positive branches of $\log f_i$ and $\log x_j$ are intended. 

As it turns out, statements about convergence of \eqref{eq:Eulermellin} involve \emph{convex polytopes} and \emph{polyhedral cones}. We start by introducing these objects, and then switch to convergence results from \cite{arkani2021stringy,berkesch2014euler,nilsson2013mellin}. In \cite{berkesch2014euler,nilsson2013mellin}, \eqref{eq:Eulermellin} was called an \emph{Euler-Mellin integral} and weaker assumptions on $f_i$ are used. In this text, we stick with Assumption \ref{assum:poscoeffs} for simplicity.

\subsection{A little polyhedral geometry}
This section introduces properties of convex polytopes and polyhedral cones that we need later on. We omit most proofs, and refer the reader to the standard textbook \cite{ziegler2012lectures} for more details.
A subset $P \subset \mathbb{R}^n$ is called \emph{convex} if for any $p_1, p_2 \in P$, the line segment $p_1p_2$ is contained in $P$. The \emph{convex hull} of $A \subset \mathbb{R}^n$ is the smallest convex subset $P \subset \mathbb{R}^n$ such that $A \subset P$. We denote this by ${\rm conv}(A)$. A \emph{convex polytope} in $\mathbb{R}^n$ is the convex hull of finitely many points. Since we will not encounter any non-convex polytopes in this text, we will sometimes omit the adjective \emph{convex}. If $P$ is a polytope and $s$ is a nonnegative number, the \emph{$s$-dilation} of $P$~is the convex polytope
\[ s \cdot P \, = \, \{ s\cdot p \, : \, p \in P \}. \]
Here $s \cdot p$ is the usual scalar multiplication for vectors in $\mathbb{R}^n$. It is easy to check that $s \cdot P$ is indeed a convex polytope. The \emph{Minkowski sum} of two polytopes $P, Q$ is a new polytope
\[ P + Q \, = \, \{ p + q \, : \, p \in P, \, q \in Q \}. \]
This binary operation is commutative and associative. An example is shown in Figure \ref{fig:MS} (left), where we take the sum of three polytopes in $\mathbb{R}^2$. Each is the convex hull of the points represented by black bullets. The \emph{dimension} of a polytope is the dimension of the smallest affine space containing it. Figure \ref{fig:MS} (left) shows two polytopes of dimension two (these are also called \emph{polygons}), and two polytopes of dimension one (i.e., line segments). In the right part of that figure, we show a three-dimensional convex polytope in $\mathbb{R}^3$. 
\begin{figure}
    \centering
    \includegraphics[width = 10cm, valign=c]{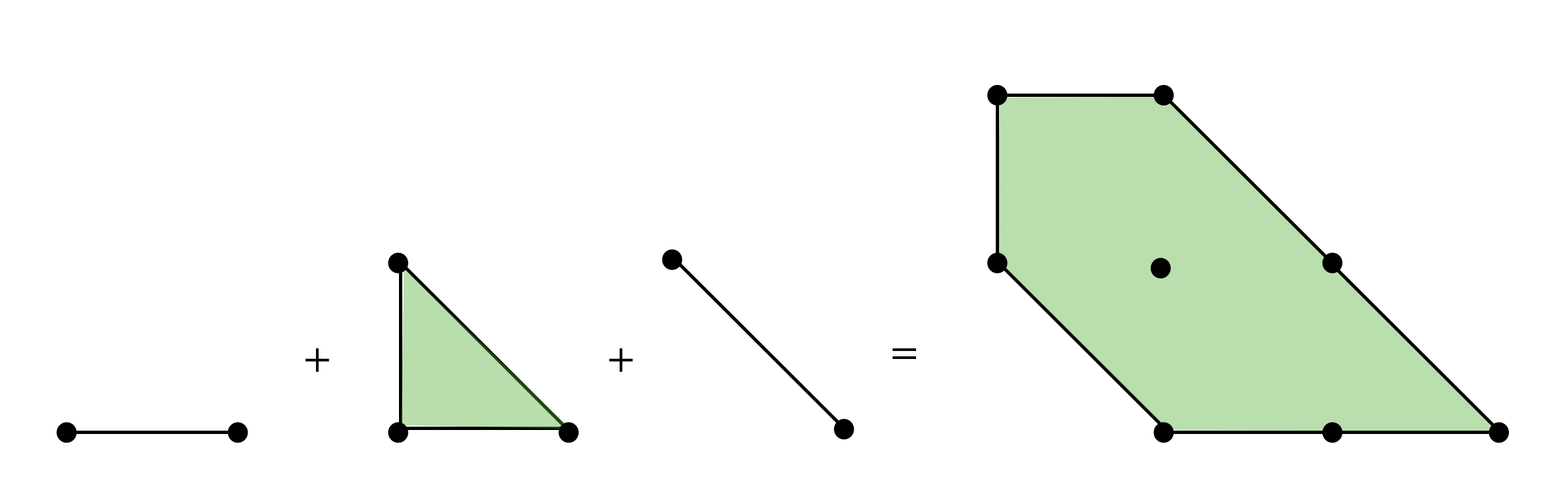}
    \adjustbox{valign=c}{\begin{tikzpicture}
    \node[] at (0,0) {\includegraphics[width = 4cm, valign=c]{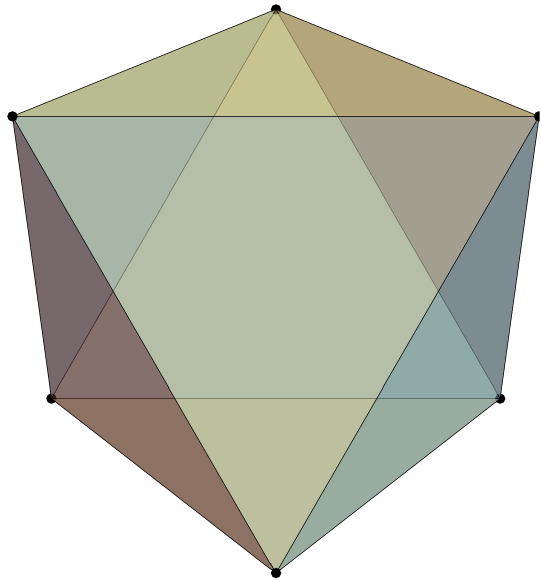}};
    \node[] at (0,-2.3) {$\scriptstyle (1,1,0)$};
    \node[] at (0,2.3) {$\scriptstyle (0,0,1)$};
    \node[] at (2.5,1) {$\scriptstyle (0,1,1)$};
    \node[] at (-2.5,1) {$\scriptstyle (1,0,1)$};
    \node[] at (2.2,-1) {$\scriptstyle (0,1,0)$};
    \node[] at (-2.2,-1) {$\scriptstyle (1,0,0)$};
    \end{tikzpicture}}
    \caption{Left: Minkowski sum of three polytopes in $\mathbb{R}^2$. Right: The polytope $\Delta(\mathcal{U}_G + \mathcal{F}_G)$ of the triangle Feynman diagram.}
    \label{fig:MS}
\end{figure}
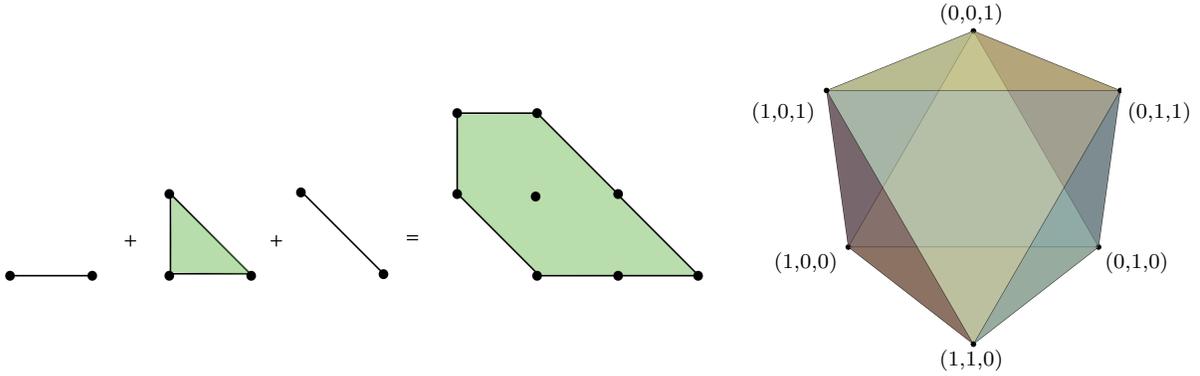
The polytopes we will encounter in this text arise as the \emph{Newton polytope} of a Laurent polynomial. 
\begin{definition}[Newton polytope] \label{def:newton}
Let $f = \sum_{\alpha \in \mathbb{Z}^n} c_\alpha \, x^\alpha \in \mathbb{C}[x_1^{\pm 1}, \ldots, x_n^{\pm 1}]$ be a Laurent polynomial. The \emph{support} of $f$ is the set ${\rm supp}(f) = \{ \alpha \in \mathbb{Z}^n \, : \, c_\alpha \neq 0 \}$. The \emph{Newton polytope} $\Delta(f) \subset \mathbb{R}^n$ is defined as the convex hull of the support, i.e., $\Delta(f) = {\rm conv}({\rm supp}(f))$.
\end{definition}
At the level of Laurent polynomials, Minkowski addition corresponds to multiplication. That is, for two Laurent polynomials $f, g$, we have $\Delta(fg) = \Delta(f) + \Delta(g)$.
\begin{example}
    The polytopes in Figure \ref{fig:MS}, from left to right, are the Newton polytopes of 
    \[ 1 + x_1, \quad 1 + x_1 + x_2, \quad x_1 + x_2, \quad (1+x_1)(1+x_1+x_2)(x_1+x_2), \]
    and the polytope $\Delta({\cal U}_G + {\cal F}_G)$ for the denominator of \eqref{eq:Itriangle}.
\end{example}
A nonzero vector $y \in \mathbb{R}^n$ defines a \emph{face} $P_y$ of a polytope $P$ as follows:
\[ P_y = \{ p \in P \, : \, y \cdot p = \min_{q \in P} y \cdot q \}. \]
In particular, $P$ is a face of itself: $P_0 = P$. Every face $P_y$ of $P$ is a polytope itself, and a face of a face of $P$ is a face of $P$ itself. If $\dim P_y = \dim P - 1$, then $P_y$ is called a \emph{facet} of $P$. Faces of dimension 0 and 1 are called \emph{vertices} and \emph{edges} respectively. For example, the polygon in the middle of Figure \ref{fig:MS} has 5 vertices, 5 facets (or edges), and one 2-dimensional face.

The faces of $P$ divide up $\mathbb{R}^n$ into finitely many regions. For a given face $Q \subseteq P$, we set
\[ C_Q \, = \, \{ y \in \mathbb{R}^n \, : \, Q \subseteq P_y \}. \]
For any face $Q \subseteq P$, $C_Q$ is a \emph{polyhedral cone}. I.e., there is a finite set $A \subset \mathbb{R}^n$ such that 
\begin{equation} \label{eq:CQ}
    C_Q \, = \,  {\rm pos}(A) = \left \{ \sum_{r \in A} c_r \, r \, : \, c_r \in \mathbb{R}_{\geq 0}  \right \} . 
\end{equation} 
All our cones are polyhedral, so we will sometimes just refer to them as \emph{cones}. The dimension of a cone is the dimension of the smallest linear space containing it. For our cones $C_Q$, we have $\dim C_Q = n - \dim Q$. E.g., if $v \in P$ is a vertex, we have $\dim C_v = n$. If $\dim P = n$ and $Q$ is a facet, then $C_Q$ is a one-dimensional cone. These are called \emph{rays}. When $Q$ runs over all faces, the cones $C_Q$ tile up $\mathbb{R}^n$. The same is true for the vertices $v$. In symbols:
\begin{equation} \label{eq:complete}
    \mathbb{R}^n \, = \, \bigcup_{Q} C_Q \, = \, \bigcup_v C_v. 
\end{equation} 
A cone $C$ is called \emph{pointed} if $C \cap (-C) = \{0\}$. If $P \subset \mathbb{R}^n$ is \emph{full-dimensional}, i.e., $\dim P = n$, the cone $C_Q$ is pointed for each face $Q \subset P$. If $C_Q$ is pointed and $A \subset \mathbb{R}^n$ is the minimal subset such that \eqref{eq:CQ} holds, the elements of $A$ are called \emph{ray generators} of $C_Q$. The reason is that each $r \in A$ generates the ray $C_{Q'} \subset C_Q$ of a facet $Q'\supset Q$. A $k$-dimensional cone is called \emph{simplicial} if it has a set of $k$ ray generators. This always holds when $k \leq 2$.

The collection of cones $\Sigma_P = \{  C_Q \, : \, Q \text{ face of } P \}$ is closed under taking intersections. In fact, one can check that $C_{Q_1} \cap C_{Q_2} = C_{Q_{12}}$, where $Q_{12} \subset P$ is the smallest face of $P$ containing both $Q_1$ and $Q_2$. This fact, together with the observation that $\Sigma_P$ is closed under taking faces (we leave the definition of a \emph{face} of a cone to the reader), makes $\Sigma_P$ into a \emph{polyhedral fan}, called the \emph{normal fan} of $P$.

\begin{example}
A pentagon $P$ in $\mathbb{R}^2$ has five vertices. These give five pointed full-dimensional cones in its normal fan. The ray separating two neighboring cones $C_{v_1}$ and $C_{v_2}$ is the cone $C_{v_1v_2}$ corresponding to the edge containing $v_1$ and $v_2$. This is illustrated in Figure \ref{fig:normalfan}. The cone $C_P = \{0\}$ is the only zero-dimensional one. The normal fan $\Sigma_P$ is invariant under translations of $P$. I.e., $\Sigma_P = \Sigma_{P+w}$ for $w \in \mathbb{R}^n$. We encourage the reader to check this. 
    \begin{figure}[h!]
    \centering
    \includegraphics[width = 11cm]{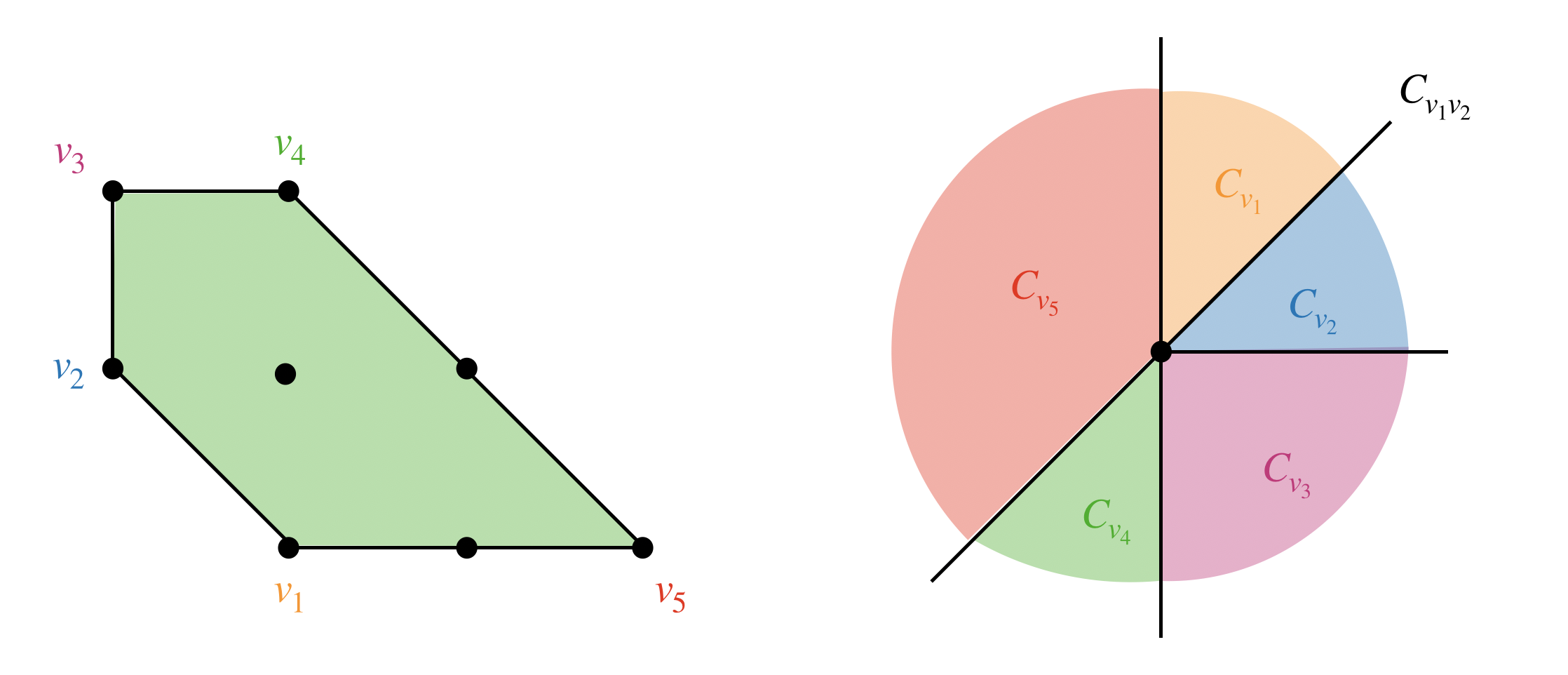}
    \caption{The normal fan of a pentagon has five two-dimensional cones.}
    \label{fig:normalfan}
\end{figure}
\end{example}
Our final construction is the \emph{polar dual} $P^\circ$ of a polytope $P \subset \mathbb{R}^n$. This is given by 
\[ P^\circ \, = \, \{ y \in \mathbb{R}^n \, : \, y \cdot p \geq -1, \text{ for all } p \in P \}. \]
If $P$ is full-dimensional and it contains the origin in its interior ${\rm int}(P)$, then $P^\circ$ is again a polytope.
Its vertices lie on the rays of the normal fan $\Sigma_P$. Hence, the normal fan induces a subdivision 
\begin{equation} \label{eq:dualdecomp}
    P^\circ \, =\, \bigcup_{v} B_v, \quad \text{where} \quad B_v \, = \, C_v \cap P^\circ. 
\end{equation}
Here $B_v$ is the convex polytope $\{ y \in C_v \, : \, y \cdot v \geq -1 \}.$
\begin{example}
The polar dual and its subdivision are illustrated in Figure \ref{fig:polardual}. To satisfy $0 \in {\rm int}(P)$, we translated our polytope so that $(0,0)$ is an interior lattice point.
    \begin{figure}[h!]
        \centering
        \includegraphics[width = 11cm]{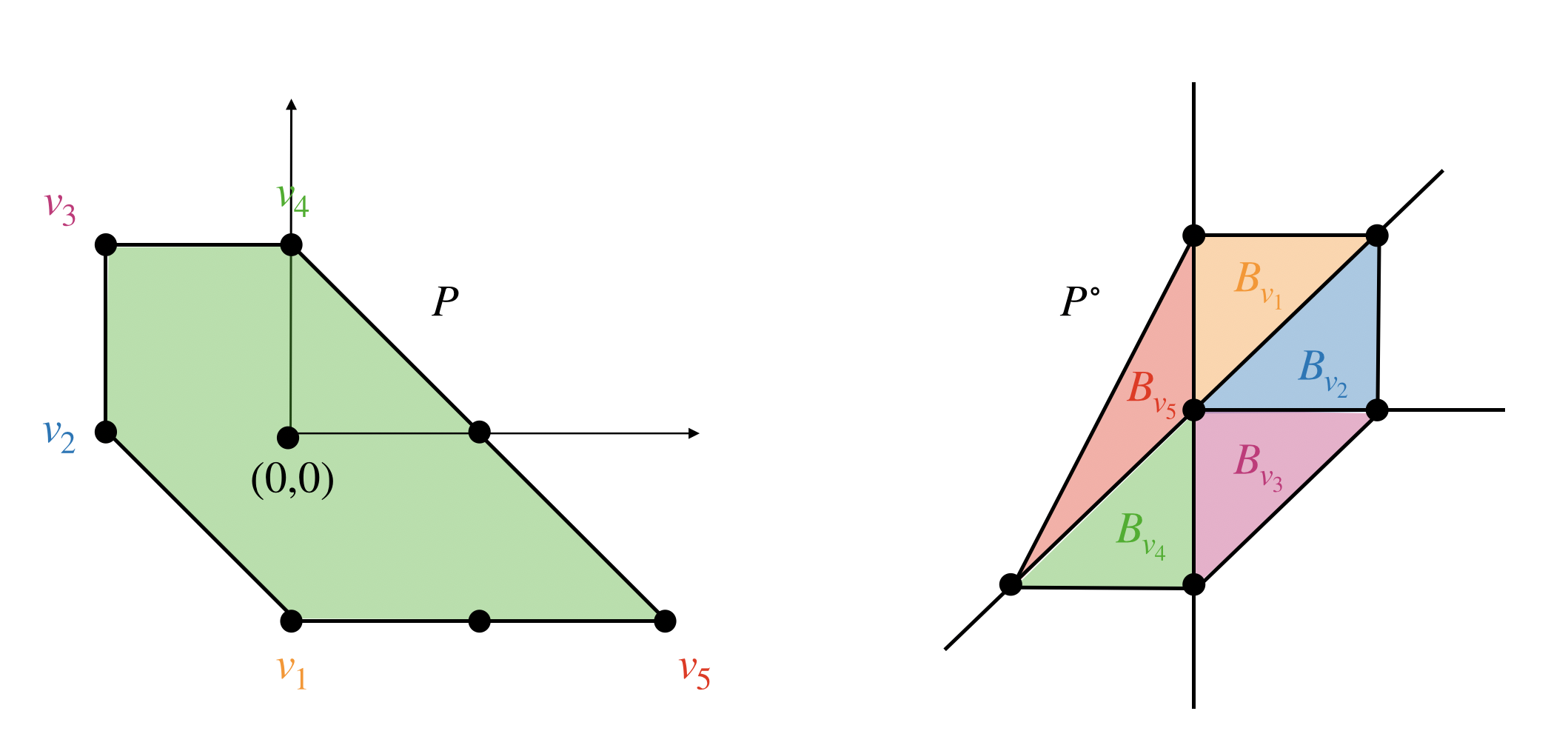}
        \caption{The polar dual of $P$ and its subdivision induced by the normal fan $\Sigma_P$. }
        \label{fig:polardual}
    \end{figure}
\end{example}
The following lemma will be useful in our discussion on convergence. 
\begin{lemma} \label{lem:allneg}
Let $P$ be a full-dimensional polytope in $\mathbb{R}^n$. We have $0 \in {\rm int}(P)$ if and only if for all vertices $v$ of $P$, $y \cdot v < 0$ for all $y \in C_v \setminus \{0\}$. 
\end{lemma}
\begin{proof}[Proof (sketch)]
The proof uses the \emph{facet description} of ${\rm int}(P)$: 
\[ {\rm int}(P) \, = \, \{ p \in \mathbb{R}^n \, : \, r_Q \cdot p > r_Q \cdot v_Q, \text{ for all facets } Q\subset P\}.\]
Here $r_Q$ is any ray generator of the ray $C_Q$, and $v_Q$ is any vertex contained in $Q$. If and only if all right-hand sides $r_Q \cdot v_Q$ are negative, $p = 0$ belongs to this set. For a given vertex $v$ of $P$, a vector $y \in C_v$ can be written as $y = \sum_{v \in Q} c_Q \, r_Q$, with $c_Q \geq 0$. The lemma follows.
\end{proof}

\subsection{\label{sec:convergence}Nilsson-Passare convergence}

The main theorem of this section identifies a region in $(s, \nu)$-space $\mathbb{C}^{\ell + n}$ in which the integral \eqref{eq:Eulermellin} converges. The result for $\ell = 1$ and $s = 1$ is due to Nilsson and Passare \cite{nilsson2013mellin}. This was generalized to integrals of the form \eqref{eq:Eulermellin} in \cite{berkesch2014euler}. The papers \cite{berkesch2014euler,nilsson2013mellin} use weaker assumptions on $f$. Our proof below is inspired by that of \cite[Section~2]{Panzer:2019yxl} and \cite[Claim 1]{arkani2021stringy}.
\begin{theorem} \label{thm:convergence}
Let ${\rm Re}(s_i) >0$ for $i = 1, \ldots, \ell$ and suppose that $\Delta(f_1) + \cdots + \Delta(f_\ell)$ has dimension $n$. The integral \eqref{eq:Eulermellin} with $f_i$ satisfying Assumption \ref{assum:poscoeffs} converges absolutely if and only if ${\rm Re}(\nu) \in {\rm int}(P(s))$, where $P(s) = {\rm Re}(s_1) \cdot \Delta(f_1) + \cdots + {\rm Re}(s_\ell) \cdot \Delta(f_\ell)$. 
\end{theorem}
\begin{example} \label{ex:pentagon}
Ignoring the $\alpha'$ parameter for now, the string amplitude ${\cal I}_5$ from \eqref{eq:I5} is
\begin{equation} \label{eq:I5noalpha}
    {\cal I}_{5} \, = \,  \int_{\mathbb{R}^2_+} \frac{x_1^{\nu_1} x_2^{\nu_2}}{(1+x_1)^{s_{13}} (1 + x_1+x_2)^{s_{14}} (x_1 + x_2)^{ s_{34}} } \, \frac{{\rm d}x_1 {\rm d}x_2}{x_1x_2} .
\end{equation}
Suppose $(s_{13}, s_{14}, s_{34}) = (1,1,1)$. By Theorem \ref{thm:convergence}, the integral converges if $(\nu_1,\nu_2) \in {\rm int}(P)$, where $P$ is the pentagon in the middle of Figure \ref{fig:MS}. For $(\nu_1,\nu_2) = (1,1)$, we find using
\begin{minted}{julia}
Integrate[((1+x1)(1+x1+x2)(x1+x2))^(-1), {x1,0,Infinity}, {x2,0,Infinity}]
\end{minted}
in \texttt{Mathematica} that ${\cal I}_5 = \pi^2/6$. Multiplying the integrand with \texttt{x1*x2}, i.e., using $(\nu_1,\nu_2) = (2,2)$, the program prints a message saying that the integral does not converge.
\end{example}
The \emph{normalized volume} of a compact set $B\subset \mathbb{R}^n$ is defined as ${\rm Vol}(B) = n! \cdot \int_B 1 {\rm d} x$. Our proof of Theorem \ref{thm:convergence} bounds the Euler integral \eqref{eq:Eulermellin} in terms of the normalized volume of polytopes, see Equation \eqref{eqn:sandwich2} below. It uses Lemma \ref{lem:allneg}, as well as the following lemma. 
\begin{lemma} \label{lem:volB}
    Let $C \subset \mathbb{R}^n$ be an $n$-dimensional polyhedral cone and let $v \in \mathbb{R}^n$ be such that $y \cdot v < 0$ for all $y \in C \setminus \{0\}$. Then $B = \{y \in C \, : \, y \cdot v \geq -1 \}$ is a polytope with volume 
    \begin{equation}\label{eq:volume}
        {\rm Vol}(B) \, = \, \int_C {\rm exp}(y \cdot v) \, {\rm d} y.
    \end{equation} 
    If, instead, $y \cdot v \geq 0$ for some $y \in C \setminus \{0\}$, then the integral above diverges. 
\end{lemma}
\begin{proof}
    It suffices to show this in the case where $C$ is simplicial, with $n$ ray generators $r_1, \ldots, r_n$. This is because if $C$ is not simplicial, it can be subdivided into finitely many simplicial cones $C_1, \ldots, C_k$, and we would conclude
    \[ {\rm Vol}(B) \, = \, \sum_{i = 1}^k {\rm Vol}(B \cap C_i) \, = \, \sum_{i=1}^k \int_{C_i} {\rm exp}(y \cdot v) \, = \, \int_C {\rm exp}(y \cdot v). \]
    Since $y \cdot v < 0$ for all $y \in C$, we may also assume that the ray generators $r_i$ are scaled so that $r_i \cdot v = -1$. This means that ${\rm Vol}(B) = |{\rm det}(A)|$, where $A = (r_1, \ldots, r_n)$ is a matrix whose columns are the ray generators. Since $C$ is simplicial, a point $y = (y_1, \ldots, y_n) \in C$ can be written uniquely as $y = A z$, where $z = (z_1, \ldots, z_n)$ are new nonnegative coordinates. Hence
    \[ \int_C {\rm exp}(y \cdot v) \, = \, |{\rm det}(A)| \int_{\mathbb{R}^n_+} \exp({v^\top A z}) {\rm d} z \, = \, |{\rm det}(A)| \int_{\mathbb{R}^n_+} \prod_{i=1}^n \exp(z_i(r_i \cdot v)) {\rm d} z.\]
    We now perform the integration for each variable $z_i$ separately to conclude
    \[\int_C {\rm exp}(y \cdot v) \, = \, |{\rm det}(A)| \prod_{i=1}^n \left [ \frac{\exp(z_i(r_i \cdot v))}{r_i \cdot v} \right ]^{z_i = \infty}_{z_i = 0}.\]
    The integral is finite if and only if $r_i \cdot v < 0$ for all rays, which is equivalent to $y \cdot v < 0$ for all $y \in C \setminus \{0\}$. In this case, it equals $|\det(A)| = {\rm Vol}(B)$, as desired. 
\end{proof}

\begin{proof}[Proof of Theorem \ref{thm:convergence}]
       We start with a change of variables $x_j = {\rm exp}(y_j)$: 
    \[ {\cal I}(s,\nu) \, = \, \int_{\mathbb{R}^n_+} f^{-s} x^\nu \, \frac{{\rm d} x}{x} \, =\, \int_{\mathbb{R}^n} \frac{{\rm exp}(y \cdot \nu)}{f({\rm exp}(y))^s} \, {\rm d} y .\]
    To show absolute convergence, we need to prove that
    \[ \int_{\mathbb{R}^n} \left | \frac{{\rm exp}(y \cdot \nu)}{f({\rm exp}(y))^s} \right |  {\rm d} y \, = \, \int_{\mathbb{R}^n} \frac{{\rm exp}(y \cdot {\rm Re}(\nu))}{f({\rm exp}(y))^{{\rm Re}(s)}} \, {\rm d} y \, < \, \infty. \]
    The equality in this display uses $|r^{a+\sqrt{-1}b}| = |{\rm exp}(\log(r)(a + \sqrt{-1}b))| = |{\rm exp}(a \,\log(r))| = r^a$ for a positive real number $r$ and real numbers $a,b$. Notice that this means we may assume $s$ and $\nu$ are real. We first consider the case where $\ell = 1$. Let $P$ be the Newton polytope $\Delta(f)$ of $f = f_1$. We have seen in \eqref{eq:complete} that its normal fan subdivides $\mathbb{R}^n$ into $n$-dimensional polyhedral cones $C_v$, where $v$ runs over the vertices of $P$. Therefore
    \begin{equation}  \label{eq:sumovervtcs}
    {\cal I}(s,\nu) \, = \, \sum_v {\cal I}_v(s,\nu) \, = \, \sum_{v} \int_{-C_v} \frac{{\rm exp}(y \cdot \nu)}{f({\rm exp}(y))^s} \, {\rm d} y \, .
    \end{equation}
    Notice that we use the cones $-C_v$ instead of $C_v$ for this decomposition, because these are the domains on which we can find easy bounds for the integrand. Let $f = \sum_{\alpha} c_\alpha \cdot x^\alpha$ with $c_\alpha > 0$, as in Assumption \ref{assum:poscoeffs}. With our change of variables, this becomes $f(e^y) = \sum_\alpha c_\alpha \cdot e^{y \cdot \alpha}$. Since $v$ is a vertex of $\Delta_f$, one of the exponents $\alpha$ equals $v$. For $y \in -C_v$, $y \cdot v \geq y \cdot \alpha$ for $\alpha \in {\rm supp}(f) \setminus \{v \}$. This gives the following chain of inequalities:
    \begin{equation} \label{eq:boundf} c_v \cdot  \exp( y \cdot v ) \, \leq \, f(\exp(y)) \, \leq \sum_\alpha c_\alpha \cdot \exp(y \cdot v). 
    \end{equation}
    This leads to a chain of inequalities of integrals. Let $M = \sum_{\alpha} c_\alpha$. Since $s > 0$, we have 
    \[ M^{-s} \cdot \int_{-C_v} {\rm exp}(y \cdot (\nu - sv)) \, {\rm d} y \, \leq \, {\cal I}_v(s,\nu) \, \leq \,   c_v^{-s} \cdot \int_{-C_v} {\rm exp}(y \cdot (\nu - sv)) \, {\rm d}y.  \]
       The integral appearing on the left and right of this expression can be written as 
    \[ \int_{C_v} {\rm exp}(y \cdot w) \, {\rm d} y, \quad \text{with} \quad w = sv - \nu. \]
    By Lemma \ref{lem:volB}, this integral converges if and only if $y \cdot w < 0$ for all $y \in C_v \setminus \{0\}$. Notice that $w$ is a vertex of the polytope $s \cdot P - \nu$, and the cone $C_{sv - \nu}$ in its normal fan equals $C_v$. By Lemma \ref{lem:allneg}, $y \cdot w < 0$ for all $y \in C_w \setminus \{0\} = C_v \setminus \{0\}$ for all vertices $w$ if and only if $0 \in {\rm int}(s \cdot P - \nu)$. This is equivalent to $\nu \in {\rm int}(s \cdot P)$, which proves the theorem for $\ell = 1$.

    When $\ell > 1$, the formula \eqref{eq:sumovervtcs} generalizes to a sum over the vertices of $P(s)$:
    \[ {\cal I}(s,\nu) \, = \, \sum_v {\cal I}_v(s,\nu) \, = \,  \sum_{v} \int_{-C_v} \frac{{\rm exp}(y \cdot \nu)}{f_1({\rm exp}(y))^{s_1} \cdots f_\ell({\rm exp}(y))^{s_\ell} } \, {\rm d} y .
    \]
    Here $P(s) = s_1 \cdot \Delta(f_1) + \cdots + s_\ell \cdot \Delta(f_\ell)$. Each vertex $v$ of $P(s)$ is a sum $v = s_1 \cdot v_1 + \cdots + s_\ell \cdot v_\ell$ of vertices $s_i \cdot v_i$ of the summands $s_i \cdot \Delta(f_i)$. Here $v_i$ is the face $\Delta(f_i)_y \subset \Delta(f_i)$ for any interior point of $C_v$. With the notation for coefficients as in Assumption \ref{assum:poscoeffs}, we set $M_i = \sum_{\alpha} c_{i,\alpha}$. Bounding each of the $f_i$ via $c_{i,v_i} \exp(y \cdot v_i) \leq f_i(\exp(y)) \leq M_i \exp(y \cdot v_i)$ as in \eqref{eq:boundf}, we find 
    \begin{equation} \label{eq:sandwich}
    \prod_{i=1}^\ell M_i^{-s_i} \cdot \int_{-C_v} {\rm exp}(y \cdot (\nu - v)) \, {\rm d} y \, \leq \, {\cal I}_v(s,\nu) \, \leq \, \prod_{i=1}^\ell c_{i,v_i}^{-s_i} \cdot \int_{-C_v} {\rm exp}(y \cdot (\nu  - v)) \, {\rm d}y.  
    \end{equation}
    Again, the integral in these expressions converges if and only if $0 \in {\rm int}(P(s) - \nu)$.
\end{proof}
\begin{remark} \label{rem:tropical}
    The decomposition of the integral ${\cal I }(s,\nu)$ in \eqref{eq:sumovervtcs} is referred to as \emph{sector decomposition} in the physics literature \cite[Section~3.4]{Heinrich:2020ybq}. This is used in state-of-the-art algorithms for evaluating Feynman integrals numerically, which use \emph{Monte Carlo sampling} and \emph{tropical geometry} \cite{borinsky2023tropical,Borinsky:2023jdv}.
\end{remark}
By Lemma \ref{lem:volB}, if ${\rm Re}(s_i) > 0$ and $\nu \in {\rm int}(P(s))$, the bounds in \eqref{eq:sandwich} can be written as 
\begin{equation}\label{eqn:sandwich2}
 \prod_{i=1}^\ell M_i^{-{\rm Re}(s_i)} \cdot {\rm Vol}(B_{v-\nu}) \, \leq \, |{\cal I}_v(s,\nu)| \, \leq \, \prod_{i=1}^\ell c_{i,v_i}^{-{\rm Re}(s_i)} \cdot {\rm Vol}(B_{v-\nu}). 
 \end{equation}
Here $B_{v-\nu} = \{ y \in C_v \, : \, y \cdot (v - \nu) \geq -1 \}$ is the portion of $(P(s)-\nu)^\circ$ corresponding to the vertex $v$, see Figure \ref{fig:polardual}. This will be important in our discussion on field theory limits. 

Theorem \ref{thm:convergence} identifies a region of convergence of the integral \eqref{eq:Eulermellin}, which is geometrically described by a convex polytope. As pointed out in \cite[Example 2.3]{berkesch2014euler}, the integral might converge on a larger domain. E.g., the assumption ${\rm Re}(s_i) > 0$ is in general not necessary. However, as it turns out, our domain is large enough to allow a unique meromorphic continuation of ${\cal I}(s,\nu)$ to the entire parameter space $\mathbb{C}^{\ell + n}$. This is similar in spirit to the fact that the integral representation of the gamma function seen in \eqref{eq:betafunc} only converges for ${\rm Re}(u) >0$. The meromorphic function $\Gamma(u)$ is obtained by extending that integral function on $\mathbb{R}_+$ to a function on $\mathbb{C} \setminus \mathbb{Z}_{\leq 0}$ satisfying $\Gamma(u+1) = u\, \Gamma(u)$. Let us now consider the beta function.

\begin{example} \label{ex:betaextend}
    The coordinate change $x = \frac{y}{1+y}$ brings the integral in \eqref{eq:betafunc} into the form~\eqref{eq:Eulermellin}:
    \begin{equation}
        \int_0^1 \frac{x^\nu}{(1-x)^s} \, \frac{{\rm d}x}{x} \, = \, \int_0^\infty \frac{y^\nu}{(1+y)^{\tilde{s}}} \, \frac{{\rm d} y}{y}, \quad \text{with} \quad \tilde{s} \, = \, \nu+1-s.
        \label{eq:positive_Beta}
    \end{equation}
    Theorem \ref{thm:convergence} predicts convergence when ${\rm Re}(\tilde{s}) > 0$ and $\nu \in {\rm int}(P(\tilde{s}))$, where
    \begin{equation} \label{eq:betapolytope} P(\tilde{s}) \, = \, \{ p \in \mathbb{R} \, : \, p \geq 0, \, -p \geq -{\rm Re}(\tilde{s}) \}.
    \end{equation}
    To justify the blue equality in \eqref{eq:betafunc}, we observe that when these convergence conditions hold,
    \begin{align*}
        \Gamma(\nu + 1 - s) \cdot B(\nu,1 -s) \, &= \, \int_0^\infty t^{\nu-s} e^{-t} {\rm d} t \, \int_0^\infty \frac{y^\nu}{(1+y)^{\nu + 1 - s}} \, \frac{{\rm d} y}{y} \\
        \, &= \, \int_0^\infty \int_0^\infty \left(\frac{ty}{1+y}\right)^{\nu-1} \left( \frac{t}{1+y} \right)^{-s} e^{-t} \, \frac{t \, {\rm d}y \, {\rm d}t}{(1+y)^2}.
    \end{align*} 
    With the coordinate change $u = \frac{ty}{1+y}, \, w = \frac{t}{1+y}$ we have $u + w = t$ and $\frac{t{\rm d} y {\rm d} t}{(1+y)^2} = {\rm d}u {\rm d}w$. Hence
    \[\Gamma(\nu + 1 - s) \cdot B(\nu,1 -s) \, = \, \int_0^\infty u^{\nu-1}e^{-u} \, {\rm d} u \, \int_0^\infty w^{-s} e^{-w} \, {\rm d} w \, = \, \Gamma(\nu) \Gamma(1-s). \]
    While the integrals in these equalities only make sense in their respective convergence regions, we may use the definition of the gamma function to extend the beta function to a meromorphic function on $\mathbb{C}^2$:
    \[ B(\nu,1-s) \, = \, \frac{\Gamma(\nu) \Gamma(1-s)}{\Gamma(\nu + 1-s)}.\]
    Its poles are countably many lines in $(s,\nu)$-space, given by $\nu, \, 1-s \in \mathbb{Z}_{\leq 0}$. 
\end{example}
The fact that the beta function extends to a meromorphic function whose poles are given by some gamma functions is an example of a general phenomenon, proved in \cite{berkesch2014euler}. We include the statement, but omit the proof. We refer the reader to \cite[Theorem 2.4]{berkesch2014euler} for full details. 
\begin{theorem} \label{thm:meromorphiccontinuation}
    Suppose the Minkowski sum $\Delta(f_1) + \cdots + \Delta(f_\ell)$ has dimension $n$ and the polytope $P(s) = \sum_{i = 1}^\ell {\rm Re}(s_i) \cdot \Delta(f_i)$ is given by $N < \infty$ inequalities: 
    \[ P(s) \, = \, \{ p \in \mathbb{R}^n \, : \, r_i \cdot p \, \geq \, w_i \cdot {\rm Re}(s), \quad i = 1, \ldots, N \}, \quad r_i \in \mathbb{R}^n, \, \, w_i \in \mathbb{R}^\ell.\]
    Under Assumption \ref{assum:poscoeffs}, ${\cal I}(s,\nu)$ from \eqref{eq:Eulermellin} admits a meromorphic continuation of the form 
    \begin{equation} \label{eq:meromcont} 
    \Phi_f(s,\nu) \cdot \prod_{i = 1}^N \Gamma( r_i \cdot \nu - w_i \cdot s ), 
    \end{equation}
    where $\Phi_f(s,\nu)$ is an entire function. 
    \end{theorem}
It is worth noting that only the entire factor $\Phi_f(s,\nu)$ in Theorem \ref{thm:meromorphiccontinuation} depends on the specific positive coefficients of $f$. The gamma factors only depend on the polyhedral data coming from $P(s)$. For fixed, positive $s$, the poles of this meromorphic continuation are hyperplanes emanating from the boundary of $P(s)$. This is illustrated in Figure \ref{fig:poles} for \eqref{eq:I5noalpha}.
\begin{figure}
    \centering
    \includegraphics[width = 6cm]{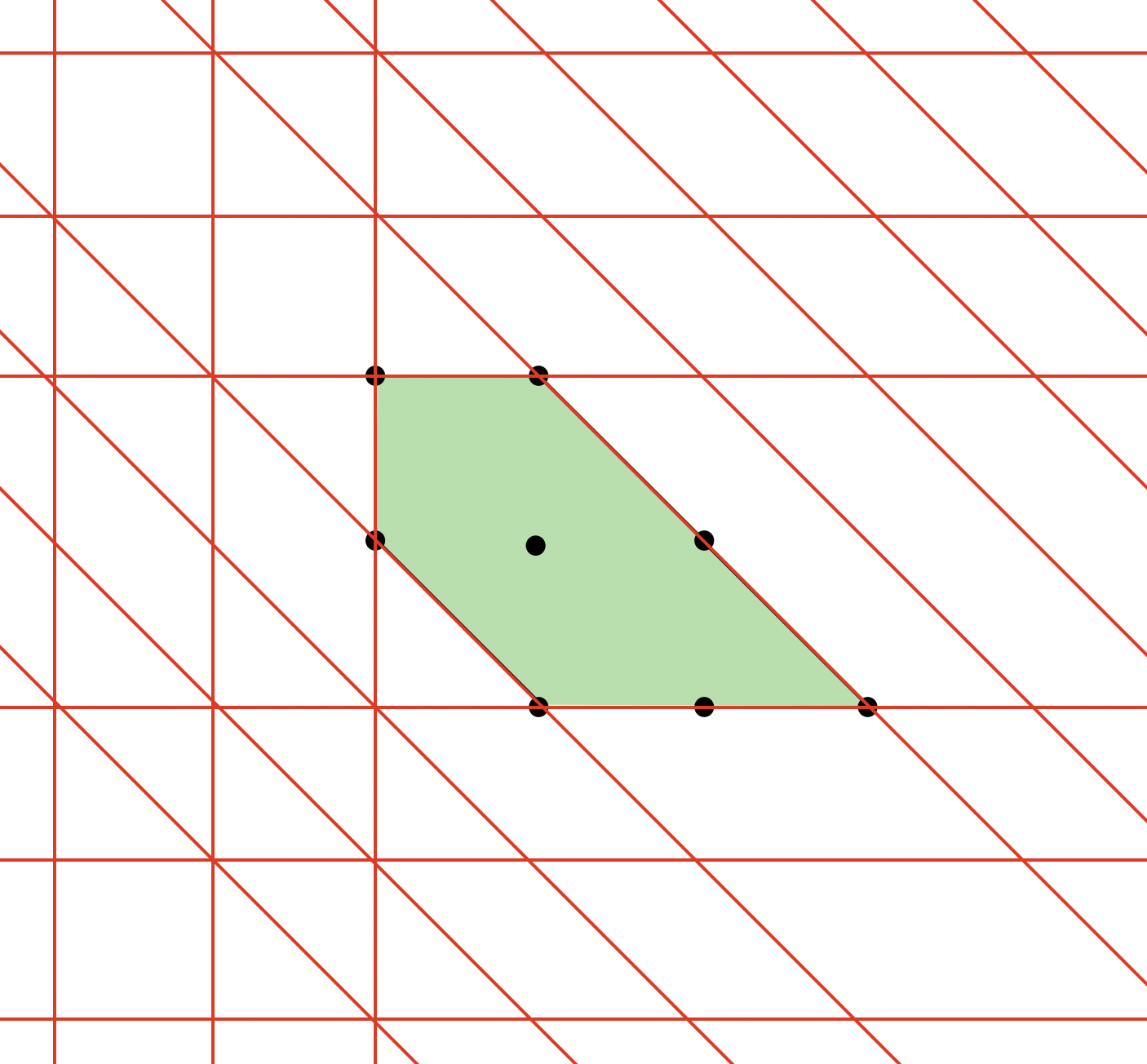}
    \caption{The poles of the meromorphic continuation of $I_5$ with $(s_{13},s_{14},s_{34}) = (1,1,1)$ are lines emanating from the boundary of the pentagon $P$ from Example \ref{ex:pentagon}.}
    \label{fig:poles}
\end{figure}
\begin{example}
    In Example \ref{ex:betaextend}, we read off from \eqref{eq:betapolytope} that $N = 2$ and $r_1 = 1, \, w_1 = 0, \, r_2 = -1, \, w_2 = -1$. The gamma factors in Theorem \ref{thm:meromorphiccontinuation} are $\Gamma(r_1 \cdot \nu - w_1 \cdot \tilde{s}) = \Gamma(\nu)$ and $\Gamma(r_2 \cdot \nu - w_2 \cdot \tilde{s}) = \Gamma(1-s)$. The entire function $\Phi_{1 + y}(s,\nu)$ equals $\Gamma(\nu + 1 - s)^{-1}$. 
\end{example}

\section{Limits and critical points} \label{sec:2}
In this section, we continue to use the integration contour $\Gamma = \mathbb{R}^n_+$ and we work under Assumption \ref{assum:poscoeffs}. We now think of the parameters $s, \nu$ to be fixed, satisfying the conditions of Theorem \ref{thm:convergence}. The main novelty with respect to Section \ref{sec:1} is that we introduce a new parameter $\delta$, of which the integral ${\cal I}$ is now a function: 
\begin{equation} \label{eq:Idelta}
    {\cal I}(\delta) \, = \,  \frac{1}{\delta^n}  \cdot \int_{\mathbb{R}^n_+} \frac{x_1^{\frac{\nu_1}{\delta}}\cdots x_n^{\frac{\nu_n}{\delta}}}{f_1^{\frac{s_1}{\delta}} \cdots f_\ell^{\frac{s_\ell}{\delta}} } \, \frac{{\rm d} x_1}{x_1} \wedge \cdots \wedge \frac{{\rm d} x_n}{x_n} \, = \, \frac{1}{\delta^n}  \int_\Gamma \left ( f^{-s}  x^\nu \right )^{\frac{1}{\delta}} \, \frac{{\rm d} x}{x} \, .
\end{equation}
Notice that $\delta^{-1}$ plays the role of the \emph{inverse string tension} $\alpha'$ in the string amplitude \eqref{eq:stringamplitude}. We are interested in the opposite limits $\lim_{\delta \rightarrow \infty} {\cal I}(\delta)$ and $\lim_{\delta \rightarrow 0^+} {\cal I}(\delta)$. Motivated by the physics application, these are called \emph{field theory limit} and \emph{high energy limit} respectively \cite{Mizera:2019gea,arkani2021stringy}. They are the leading terms in the series expansions of ${\cal I}(\delta)$ around $\delta = \infty$ and $\delta = 0$. 

As it turns out, both the field theory limit and the high energy limit can be expressed in terms of complex critical points of the \emph{potential function} or \emph{log-likelihood function}
\[ {\rm log} \, L \, = \, {\rm log} \, f^{-s} x^\nu \, = \, -s_1 \, \log f_1 - \cdots - s_\ell \, \log f_\ell + \nu_1 \, \log x_1 + \cdots + \nu_n \, \log x_n. \]
These critical points are the complex solutions to the $n$ rational function equations 
\begin{equation} \label{eq:criteq}
\frac{\partial (\log \, f^{-s}x^\nu)}{\partial x_j} \, = \, \frac{\nu_j}{x_j} - s_1 \, \frac{\frac{\partial{f_1}}{\partial x_j}}{f_1} - \cdots - s_\ell \, \frac{\frac{\partial{f_\ell}}{\partial x_j}}{f_\ell} \, = \, 0, \quad j = 1, \ldots, n. 
\end{equation}
These rational functions are defined where neither $x_j$ nor $f_i(x)$ are zero. We define 
\begin{equation} \label{eq:veryaffinevariety}
X \, = \, \{ x \in \mathbb{C}^n \, : \, x_1 \cdots x_n \cdot f_1(x) \cdots f_\ell(x) \, \neq \, 0 \, \}.
\end{equation}
This is an example of a \emph{very affine variety}, see \cite[page 6]{huh2014likelihood}. The set of complex critical points of $\log L$ is ${\rm Crit}(\log  L) = \{ x \in X \, : \, x \text{ satisfies } \eqref{eq:criteq}  \}$. Since \eqref{eq:criteq} consists of $n$ equations in $n$ unknowns, we expect ${\rm Crit}(\log L)$ to be finite. A solution $x \in {\rm Crit}(\log L)$ is \emph{degenerate} if
\begin{equation} \label{eq:torichessian}
H_{\log L} \, = \,  \det \left ( x_j  \frac{\partial}{\partial x_j} \left ( x_k\frac{\partial}{\partial x_k} \log L (x) \right) \right)_{j,k} \, = \, 0 .
\end{equation}
This determinant is called the \emph{toric Hessian} of $\log L$. It is much like the usual Hessian determinant, but with $\partial/\partial x_j$ replaced by the Euler operator $x_j(\partial/\partial x_j)$. Using the toric version will be convenient later in the section. The following result is Theorem 1 in \cite{huh2013maximum}. 

\begin{theorem} \label{thm:huh}
    There is a dense open subset $U \subset \mathbb{C}^{\ell + n}$ such that for $(s,\nu) \in U$, the number of solutions to \eqref{eq:criteq} equals the signed Euler characteristic $(-1)^n \cdot \chi(X)$ of the very affine variety $X$, and all solutions are non-degenerate, meaning that $H_{\log L}(x) \neq 0$ for all $x \in {\rm Crit}(\log L)$.
\end{theorem}
 Theorem \ref{thm:huh} says that the number of points in ${\rm Crit}(\log L)$ depends only on the topology of the space $X$. We will see how to compute the Euler characteristic $\chi(X)$ below.

 In Section \ref{sec:2.1} we discuss how to compute ${\rm Crit}(\log L)$ using numerical homotopy continuation. Sections \ref{sec:2.2} and \ref{sec:2.3} explain how these critical points are used to compute field theory and high energy limits respectively.

\subsection{Computing critical points} \label{sec:2.1}
Our goal is to solve the equations \eqref{eq:criteq}, assuming that $(s,\nu) \in \mathbb{C}^{\ell + n}$ belongs to the set $U$ from Theorem \ref{thm:huh}. This ensures that there are finitely many solutions, all of them non-degenerate, and the number of solutions we find is the Euler characteristic of $X$ (up to sign). 

We use the Julia package \texttt{HomotopyContinuation.jl} (v2.6.4) for all computations in this section \cite{breiding2018homotopycontinuation}. Our approach follows that in \cite{sturmfels2021likelihood} and is illustrated by means of an example, for which we use the integral \eqref{eq:I5noalpha}. The function $\log L$ is 
\[ \log L \, = \, -s_{13} \log (1+x_1) - s_{14} \log (1 + x_1 + x_2) - s_{34} \log(x_1 + x_2) + \nu_1 \log x_1 + \nu_2 \log x_2 . \]
Its partial derivatives give two rational function equations $g_1 = g_2 = 0$ in the unknowns~$x_1, x_2$: 
\begin{equation} \label{eq:criteqM05}
g_1 \, = \, -\frac{s_{13}}{1 + x_1} - \frac{s_{14}}{1 + x_1 + x_2} - \frac{s_{34}}{x_1 + x_2} + \frac{\nu_1}{x_1}, \quad g_2 \, = \, - \frac{s_{14}}{1 + x_1 + x_2} - \frac{s_{34}}{x_1 + x_2} + \frac{\nu_2}{x_2}.  
\end{equation}
Importantly, we think of $s$ and $\nu$ as parameters at this stage. We will emphasize the dependence of $g_i$ on these parameters by writing $g_i(x;s,\nu)$. The fixed complex parameters we want to solve for are denoted by $(s^*,\nu^*) \in U$. Here is how to code this up in Julia:
\begin{minted}{julia}
using HomotopyContinuation              # load the package 
n = 2; l = 3; @var ν[1:n] s[1:l] x[1:n] # declare variables and parameters
f = [1 + x[1]; 1 + x[1] + x[2]; x[1] + x[2]] 
logL = - sum([s[i]*log(f[i]) for i = 1:l]) + sum([ν[j]*log(x[j]) for j = 1:n])
g = differentiate(logL, x) 
g_sys = System(g, parameters = [s; ν])  # system of equations with parameters
s_star = [1;1;1]; ν_star = [1;1]        # choice of target parameters
\end{minted}
Here we chose $s^* = (1,1,1)$ and $\nu^* = (1,1)$, like in Example \ref{ex:pentagon}. The strategy for solving $g_1(x;s^*, \nu^*) = g_2(x;s^*,\nu^*) = 0$ consists of two steps: 
\begin{enumerate}
    \item Solve $g_1(x; \tilde{s}, \tilde{\nu}) = g_2(x; \tilde{s}, \tilde{\nu}) = 0$ for a \emph{different} set of parameters $(\tilde{s}, \tilde{\nu}) \in U$.
    \item Deform the \emph{start parameters} $(\tilde{s}, \tilde{\nu})$ continuously into the \emph{target parameters} $(s^*, \nu^*)$ and, along the way, keep track of the solutions to $g_1(x;s,\nu) = g_2(x;s,\nu) = 0$. 
\end{enumerate}
Both these steps require numerically tracking solution paths as we vary the parameters. This can be phrased as numerically solving an ordinary differential equation called the \emph{Davidenko equation}. For details, we refer to the standard textbook \cite{sommese2005numerical}.

Step 1 is done using the \emph{monodromy method} \cite{duff2019solving}. We explain how this works in a nutshell, using Figure \ref{fig:monodromy} as an illustration. 
\begin{figure}
    \centering
    \includegraphics[height = 6cm]{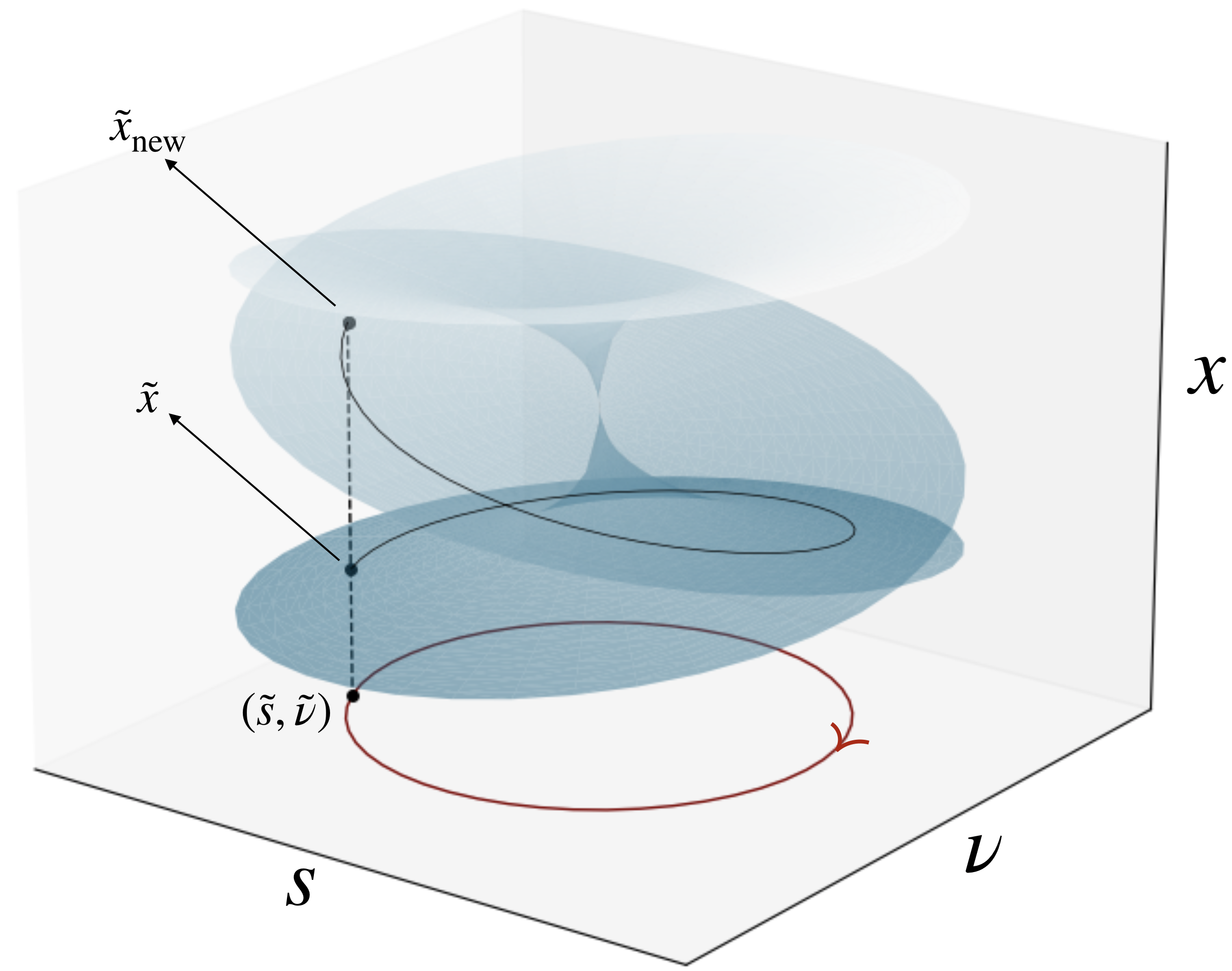}
    \caption{Illustration of the monodromy method.}
    \label{fig:monodromy}
\end{figure}
In that cartoon, the solutions for a fixed point $(s,\nu)$ are represented by the points on the blue surface lying directly above it. This surface represents the \emph{incidence space} $\{(x,s,\nu) \in X \times \mathbb{C}^{\ell + n} \, : \, g_1(x;s,\nu) = g_2(x;s,\nu) = 0 \}$. Choose a random point $\tilde{x} \in X$, and let $(\tilde{s}, \tilde{\nu})$ be any solution to the \emph{linear} system of equations $g_1(\tilde{x}; s, \nu) = g_2(\tilde{x}; s, \nu) = 0$. Clearly, $\tilde{x}$ is a point lying above $(\tilde{s},\tilde{\nu})$. This is called the \emph{seed solution}. Now walk a loop in $(s, \nu)$-space while keeping track of the seed solution $\tilde{x}$ along the way. When we arrive back at $(\tilde{s},\tilde{\nu})$, there is a good chance we picked up a new solution $\tilde{x}_{\rm new}$ to the system $g_1(x;\tilde{s},\tilde{\nu})= g_2(x;\tilde{s},\tilde{\nu}) = 0$. Now repeat this procedure to populate the solution set. In practice, this technique is extremely effective. In Julia, all this happens via
\begin{minted}{julia}
R1 = monodromy_solve(g_sys)
start_pars = parameters(R1); start_sols = solutions(R1)
\end{minted}
The variable \texttt{start\_pars} stands for \emph{start parameters}. It contains $\ell + n = 5$ complex numbers, the first $\ell = 3$ of which give $\tilde{s}$, and the last $n = 2$ give $\tilde{\nu}$. If all went well, the variable \texttt{start\_sols} contains all solutions to $g_1(x;\tilde{s},\tilde{\nu})= g_2(x;\tilde{s},\tilde{\nu}) = 0$. Hence, step 1 is completed. 

By Theorem \ref{thm:huh}, the number of solutions in \texttt{start\_solutions} equals $(-1)^n \cdot \chi(X)$. This gives a way of computing $\chi(X)$, which has been applied in some challenging cases \cite{agostini2023likelihood,sturmfels2021likelihood}. In our example, the number of solutions is $2$. Here is a way to verify that $\chi(X) = 2$. The real part $X_{\mathbb{R}}$ of the very affine variety $X$ is the complement of an arrangement of five lines in $\mathbb{R}^2$. These lines are given by $\{x_1 = 0\}, \{x_2 = 0\}, \{1+x_1 = 0\}, \{1 + x_1 + x_2 = 0\}$ and $\{ x_1 + x_2 = 0 \}$. By \cite[Theorem 1.2.1]{varchenko1995critical}, the signed Euler characteristic of $X$ is the number of bounded cells of $X_{\mathbb{R}}$. In our case, the bounded cells are two triangles. 

In step 2, we use our start solutions as initial conditions for path tracking from $(\tilde{s}, \tilde{\nu})$ to $(s^*,\nu^*)$. That is, we use the solutions \texttt{start\_sols} for parameters \texttt{start\_pars} to compute the solutions \texttt{solutions(R2)} for the target parameters \texttt{[s\_star; ν\_star]}:
\begin{minted}{julia}
R2 = solve(g_sys, start_sols; start_parameters = start_pars, 
                              target_parameters = [s_star;ν_star])
solutions(R2)
\end{minted}
The last line prints an accurate numerical approximation of the two critical points:
\begin{equation}\label{eq:two_critical_points}
\left ( \frac{\sqrt{5}-1}{2},\, 1  \right ), \quad \left ( \frac{-\sqrt{5}-1}{2},\, 1  \right ).    
\end{equation} 

\subsection{Dual volumes in field theory limits} \label{sec:2.2}

We switch back to the integral \eqref{eq:Idelta}. The field theory limit of ${\cal I}(\delta)$ is $\lim_{\delta \rightarrow \infty} {\cal I}(\delta)$. This has a nice description in terms of our polytope $P(s) = \sum_{i = 1}^\ell {\rm Re}(s_i) \cdot \Delta(f_i)$ from Theorem \ref{thm:convergence}, and in terms of the critical points ${\rm Crit}(\log L)$ computed in the previous section. The statement uses the toric Hessian determinant $H_{-\log L}$ of \emph{minus} the log-likelihood function, see \eqref{eq:torichessian}.  

\begin{theorem} \label{thm:main2.2}
Let ${\rm Re}(s_i) >0$ for $i = 1, \ldots, \ell$ and suppose that $\Delta(f_1) + \cdots + \Delta(f_\ell)$ has dimension $n$. If $f_1, \ldots, f_\ell$ satisfy Assumption \ref{assum:poscoeffs} and $\nu \in {\rm int}(P(s))$, then we have
    \begin{equation}
        \lim_{\delta \rightarrow \infty}{\cal I}(\delta) \, = \, {\rm Vol}((P(s)-\nu)^\circ) \, = \, \sum_{x \in {\rm Crit}(\log L)} H_{-\log L}(x)^{-1}.
        \label{eq:main2.2}
    \end{equation}
\end{theorem}
\begin{proof}
We prove the first equality. The second equality uses \cite[Section 7.1, Claim 4]{arkani2021stringy}. See also \cite[Theorem 13]{sturmfels2021likelihood}. For any fixed $\delta \in \mathbb{R}_+$, we have $\nu/\delta \in {\rm int}(P(s/\delta))$, and a vertex $v$ of $P(s)$ gives a vertex $(v-\nu)/\delta$ of $P(s/\delta) - \nu/\delta$. We can use \eqref{eqn:sandwich2} to obtain the estimate
\begin{equation}\label{eq:ineq_d}
     \sum_{v}\prod_{i=1}^\ell M_i^{-{\rm Re}(s_i)/\delta} \cdot {\rm Vol}(B_{\frac{v-\nu}{\delta}}) \, \leq \, \delta^n \cdot {\cal I}(\delta) \, \leq \, \sum_{v}\prod_{i=1}^\ell c_{i,v_i}^{-{\rm Re}(s_i)/\delta} \cdot {\rm Vol}(B_{\frac{v-\nu}{\delta}}). 
\end{equation}
The sums are over vertices of $P(s)$. The factor $\delta^n$ in the middle comes from ${\cal I}(\delta) = \delta^{-n} \cdot {\cal I}(s/\delta,\nu/\delta)$, with ${\cal I}(s,\nu)$ as in \eqref{eq:Eulermellin}.
Using the scaling property of the volume ${\rm Vol}(B_{(v-\nu)/\delta}) = \delta^n \cdot  {\rm Vol}(B_{v-\nu})$, we can cancel $\delta^n$ from \eqref{eq:ineq_d}. Taking the limit $\delta\to\infty$ gives 
\[ \lim_{\delta \rightarrow \infty} {\cal I}(\delta) \, = \, \sum_{v} {\rm Vol}(B_{v-\nu}) \, = \, {\rm Vol}((P(s)-\nu)^\circ) \]
as desired. For the last equality, see \eqref{eq:dualdecomp} and Figure \ref{fig:polardual}.
\end{proof}

\begin{example} \label{ex:beta2.2}
    Let us verify the formula \eqref{eq:main2.2} for the integral representation \eqref{eq:positive_Beta} of the beta function.
    The Newton polytope $P({\tilde{s}})$ is a segment given by \eqref{eq:betapolytope}.
    The dual polytope $(P(\tilde{s})-\nu)^\circ$ is given by $(\frac{1}{\nu-\tilde{s}},\frac{1}{\nu})$.
    Its volume ${\rm Vol}((P(\tilde{s})-\nu)^\circ)$ is $\frac{1}{\nu}+\frac{1}{\tilde{s}-\nu}=\frac{\tilde{s}}{\nu(\tilde{s}-\nu)}$. Checking that this equals our field theory limit can be done in one line of \texttt{Mathematica} code:
    \begin{minted}{julia}
Limit[1/δ*Integrate[y^(ν/δ - 1)/(1 + y)^(s/δ), {y, 0, Infinity}], δ -> Infinity]
    \end{minted}
    Finally, we compute the sum of $H_{- \log L}(y)$ evaluated at the critical points of ${\log L}$. We solve
    \[ \frac{d}{dy}\log L(y) \, = \, \frac{\nu}{y}-\frac{\tilde{s}}{1+y} \, = \, 0. \]
    The unique solution is $y=\frac{\nu}{\tilde{s}-\nu}$. The toric Hessian determinant of ${-\log L}$ is 
    \begin{equation}\label{eq:H_for_beta}
    H_{-\log L}(y) \, = \, -y \frac{d}{d y} \left( y \frac{d}{dy} \log L(y) \right) \, = \, -y \frac{d}{d y}  \left( \nu - \frac{\tilde{s}y}{1+y} \right )  \, = \, \frac{\tilde{s}y}{(1+y)^2}.
    \end{equation}
    The value at $y=\frac{\nu}{\tilde{s}-\nu}$ is $\frac{\nu(\tilde{s}-\nu)}{\tilde{s}}$. We have now confirmed \eqref{eq:main2.2}.
\end{example}

\begin{example}\label{ex:2.4}
    Consider the integral ${\cal I}(\delta)$ from \eqref{eq:I5} with $\alpha'=\frac{1}{\delta}$.
    The Newton polytope $P(s) = P(s_{13},s_{14},s_{24})$ of $f_1(x_1,x_2)=1+x_1,f_2(x_1,x_2)=1+x_1+x_2,f_3(x_1,x_2)=x_1+x_2$ is two dimensional for positive $s$.
    When $s_{13}=s_{14}=s_{24}=1$, it is as in Figure \ref{fig:MS}.
    We take $\nu_1=\nu_2=1$.
    The values of $H_{-\log L}$ at the two critical points \eqref{eq:two_critical_points} are
    \[
    \frac{1}{2} \left(25+11 \sqrt{5}\right),\quad\frac{1}{2} \left(25-11
   \sqrt{5}\right).
    \]
    The sum of the reciprocals of these two numbers is $5$. This is the area of $(P(s) - \nu)^\circ$ in the right part of Figure \ref{fig:polardual}, normalized by a factor $2! = 2$ (recall our definition of ${\rm Vol}$ in the discussion preceding Lemma \ref{lem:volB}).
Let us illustrate the computation of the dual volume using the Julia package \texttt{OSCAR.jl} (v0.12.0) \cite{OSCAR}.
It calls \texttt{polymake} for polytope computations \cite{gawrilow2000polymake}.
The Newton polytope $P=P(s_{13},s_{14},s_{24})$ of Example \ref{ex:2.4} is computed as follows:

\begin{minted}{julia}
using Oscar                              #load the package
P1 = convex_hull([0 0;1 0])              #Newton polytope of f1
P2 = convex_hull([0 0;1 0;0 1])          #Newton polytope of f2
P3 = convex_hull([1 0;0 1])              #Newton polytope of f3
P = P1+P2+P3                             #Minkowski sum
\end{minted}

The dual polytope $(P(s)-\nu)^\circ$ and its volume are computed by the commands \texttt{polarize} and \texttt{volume} respectively.
The vertices of the dual polytope $(P-\nu)^\circ$ are $(1, 1), (1,0),( 0,-1),(-1,-1),(0,1)$ as in Figure \ref{fig:polardual}.
The normalized volume is 5, as expected.

\if
For a polyhedral cone $C$ in $\mathbb{R}^n$, its {\it dual cone} $C^\circ$ is
\[
C^\circ=\{ \phi\in(\mathbb{R}^n)^\vee\mid \phi(v)\geq 0\ \text{for all }v\in C\}.
\]
Here, the symbol $(\mathbb{R}^n)^\vee$ denotes the dual vector space of $\mathbb{R}^n$.
We identify $(\mathbb{R}^n)^\vee$ with $\mathbb{R}^n$ through the standard dot product.
When the vertices of $P$ are $v_1,\dots,v_N$, we associate an $(n+1)$-dimensional cone $C(P)$ spanned by $(1,v_1),\dots,(1,v_N)\in\mathbb{R}^{n+1}$.
The set of these vectors is called {\it Cayley configuration}.
Let us write $\pi:\mathbb{R}^{n+1}\to\mathbb{R}^n$ for the projection omitting the first coordinate.
If $P$ is full-dimensional and it contains the origin in its interior, the dual polytope $P^\circ$
is precisely given by
\[
P^\circ=\pi\left( \{ w\in C(P)^\circ \mid w=(1,w')\text{ for some }w'\in\mathbb{R}^n\}\right).
\]
For example, we can compute the dual polytope $(P-\nu)^\circ$ of Example \ref{ex:2.4} as follows:
\fi
\begin{minted}{julia}           
DP = polarize(P+[-1,-1])                 #dual polytope of P-ν
println(vertices(DP))                    #print the vertices
factorial(2)*volume(DP)                  #normalized volume - output: 5
\end{minted}
The reader is encouraged to repeat this example for the Feynman integral \eqref{eq:Itriangle}.
\end{example}
We point out that the field theory limit $\alpha' \rightarrow 0$ of the string amplitude \eqref{eq:stringamplitude} is the scattering amplitude for a physical model called \emph{bi-adjoint scalar $\phi^3$ theory}. The expression in terms of critical points was first discovered in \cite{cachazo2014scattering}. The critical point equations \eqref{eq:criteq} are called the \emph{scattering equations} in this context. For a connection to algebraic statistics, see \cite{sturmfels2021likelihood}.

As a final remark on field theory limits, note that the coordinates of the individual critical points in ${\rm Crit}(\log L)$ are \emph{algebraic} functions of $s, \nu$. They are usually not rational functions, like in Example \ref{ex:beta2.2}. For instance, eliminating $x_2$ from \eqref{eq:criteqM05} gives a quadratic equation in $x_1$, resulting in the square roots in \eqref{eq:two_critical_points} via the quadratic formula. However, the sum over ${\rm Crit}(\log L)$ in Theorem \ref{thm:main2.2} is a rational function in $s$ and $\nu$ by Galois theory. This is called the \emph{canonical function} of $P(s)$. On the domain ${\rm Re}(s_i) > 0$, $\nu \in {\rm int}(P(s))$, it evaluates to ${\rm Vol}((P(s)-\nu)^\circ)$. When $s$ with positive real part is fixed and $P(s)$ is viewed as a polytope in $n$-dimensional $\nu$-space $\mathbb{R}^n$, the canonical function defines a meromorphic top form on $\mathbb{R}^n$. That form is called the \emph{canonical form} of $P(s)$ in the theory of \emph{positive geometries} \cite{arkani2017positive}. 

\subsection{Saddle point approximation in high energy limits} \label{sec:2.3}

While the field theory limit $\lim_{\delta \rightarrow \infty} {\cal I}(\delta)$ is obtained by summing over the complex points ${\rm Crit}(\log L)$, the high energy limit $\lim_{\delta \rightarrow 0^+} {\cal I}(\delta)$ is governed by a single, positive critical point. 
\begin{theorem} \label{thm:main2.3}
    Let $s_i \in \mathbb{R}_+$ for $i = 1, \ldots, \ell$ and suppose that $\Delta(f_1) + \cdots + \Delta(f_\ell)$ has dimension $n$. If $f_1, \ldots, f_\ell$ satisfy Assumption \ref{assum:poscoeffs} and $\nu \in {\rm int}(P(s))$, then ${\rm Crit}(\log L) \cap \mathbb{R}^n_+$ consists of one point $\{a\}$. Moreover, the following formula holds:
    \begin{equation}\label{eq:high_energy_limit}
    \lim_{\delta \rightarrow 0^+}\left (\frac{2\pi}{\delta}\right)^{-\frac{n}{2}}L(a)^{-\frac{1}{\delta}}{\cal I}(\delta) \, = \, \left(H_{-\log L}(a)\right)^{-\frac{1}{2}}.
    \end{equation}
\end{theorem}
We used the following conventions in \eqref{eq:high_energy_limit}. As above, for a positive real number $r$, we take the branch of the logarithm for which $\log r \in\mathbb{R}$. For the square root, we set $(-r)^{\frac{1}{2}}=e^{i\frac{\pi}{2}}r^{\frac{1}{2}}$. Notice that, rather than requiring ${\rm Re}(s_i) >0$, here we only allow real values for $s_i$. The reader can check that, using the values  $s_{13} = s_{14}  = s_{34} - \sqrt{-1} = 1$ in the example of Section \ref{sec:2.1} instead, there are no positive critical points. 
    
To prove Theorem \ref{thm:main2.3}, we use two lemmas. The first is on the concavity of $\log L$.

\begin{lemma}\label{lem:convexity}
Let $f_1, \ldots, f_\ell$ satisfy Assumption \ref{assum:poscoeffs} and let $s_1, \ldots, s_\ell \in \mathbb{R}_+$.
The function $\log L(e^{z_1},\ldots,e^{z_n})$ with $L(x)=f^{-s}x^{\nu}$ is strictly concave in $z\in \mathbb{R}^n$. The toric Hessian~matrix
\begin{equation}\label{eq:toric_hessian_matrix}
\left ( x_j  \frac{\partial}{\partial x_j} \left ( x_k\frac{\partial}{\partial x_k} \log L (x) \right) \right)_{j,k}    
\end{equation}
of $\log L$ is negative definite for any  $x\in\mathbb{R}^n_+$.
\end{lemma}
\begin{proof}
    Substituting $x_j = {\rm exp}(z_j)$, the toric Hessian matrix \eqref{eq:toric_hessian_matrix} is the usual Hessian matrix of $-\sum_{i=1}^\ell s_i \log f_i(e^{z_1}, \ldots, e^{z_n})$.
    Each summand is strictly concave for $z \in \mathbb{R}^n$ by \cite[Theorem 1.13]{brown1986fundamentals}, so the toric Hessian is indeed a negative definite matrix. 
\end{proof}
To state the next lemma, we introduce a version of the \emph{algebraic moment map} $\mu_{\mathbb{C}}: X \longrightarrow \mathbb{C}^n$. This name comes from toric geometry, see Fulton's book \cite[Section 4.2]{fulton1993introduction}. Our moment map is slightly different from the one used by Fulton. It is given by 
\begin{equation} 
\mu_{\mathbb{C}}(x) \, = \, \left ( x_1 \sum_{i=1}^\ell s_i \, f_i(x)^{-1} \, \frac{\partial f_i}{\partial x_1}(x), \, \ldots, \, x_n \sum_{i=1}^\ell s_i \, f_i(x)^{-1} \, \frac{\partial f_i}{\partial x_n}(x) \right ) .
\end{equation}
This map is closely related to our critical points. From \eqref{eq:criteq}, it is clear that $x \in {\rm Crit}(\log L)$ if and only if $\mu_{\mathbb{C}}(x) = \nu$. The crucial properties of $\mu_{\mathbb{C}}$ are summarized in the following Lemma.

\begin{lemma}\label{lem:moment_map}
    Let $s_i \in \mathbb{R}_+$ for $i = 1, \ldots, \ell$ and suppose that $\Delta(f_1) + \cdots + \Delta(f_\ell)$ has dimension $n$. If $f_1, \ldots, f_\ell$ satisfy Assumption \ref{assum:poscoeffs}, then $\mu_{\mathbb{C}}(\mathbb{R}^n_+) \subset {\rm int}(P(s))$. Moreover, the restriction 
    \begin{equation} \label{eq:momentmap}
        \mu \, = \, (\mu_\mathbb{C})_{|\mathbb{R}^n_+}\, = \, (\mu_1,\dots,\mu_n): \mathbb{R}^n_+ \longrightarrow {\rm int} (P(s))
    \end{equation}
    of $\mu_{\mathbb{C}}$ to $\mathbb{R}^n_+$ is a diffeomorphism, 
    and the Jacobian matrix $\left( \frac{\partial\mu_k}{\partial x_j}\right)_{j,k=1}^n$ is positive definite.
\end{lemma}
\begin{proof}
    This theorem follows from \cite[Section 4.2]{fulton1993introduction} if $\ell=1$.
    The statement for $\ell\geq 1$ has appeared in \cite[Claim 4]{arkani2021stringy}.
    For the readers' convenience, we provide a concise proof.
    Let $A_i = {\rm supp}(f_i) \subset\mathbb{Z}^n$ be the set of exponents appearing in $f_i(x)=\sum_{\alpha \in A_i}c_{i,\alpha}x^{\alpha}$. 
    We construct the \emph{Cayley configuration} $A \subset \mathbb{Z}^{\ell + n}$ of $A_1, \ldots, A_\ell$. This is given by 
    \begin{equation}\label{eq:CayleyConfiguration}
        A \, = \, \{ (e_i,\alpha) \, :\, \alpha \in A_i, \,  i = 1, \ldots, \ell \} \quad \subset \mathbb{Z}^{\ell + n},
    \end{equation}
    where $e_i$ is the $i$-th standard basis vector of $\mathbb{Z}^{\ell}$.
    Consider the Laurent polynomial $\hat{f} = s_1 y_1 f_1 + \cdots + s_\ell y_\ell f_\ell \in \mathbb{R}_+[y_1, \ldots, y_\ell, x_1^{\pm1}, \ldots, x_n^{\pm 1}]$. We define the map 
    \[ \hat{\mu}: \mathbb{R}^{\ell + n}_+ \rightarrow {\rm int}({\rm pos}(A)), \quad (y,x) \longmapsto \left(y_1 \frac{\partial \hat{f}}{\partial y_1}, \ldots, y_\ell \frac{\partial \hat{f}}{\partial y_\ell}, x_1 \frac{\partial \hat{f}}{\partial x_1}, \ldots, x_n \frac{\partial \hat{f}}{\partial x_n} \right).\]
    Here, ${\rm pos}(A)$ is the positive hull seen in \eqref{eq:CQ}. The reader who is unfamiliar with moment maps should check that the image of $\hat{\mu}$ indeed lies in the interior of the cone ${\rm pos}(A)$. By the statement labeled $(A_n)$ in \cite[page 83]{fulton1993introduction}, $\hat{\mu}$ is a diffeomorphism. 
    
    We consider the polytope $\hat{P}(s)$ consisting of all points in ${\rm pos}(A)$ whose first $\ell$ coordinates are $(s_1, \ldots, s_\ell)$. The preimage of $\hat{P}(s)$ under $\hat{\mu}$ is given by 
    \[ \hat{X}_+ \, = \, \{(y,x) \in \mathbb{R}^{\ell + n}_+ \, :\, y_1f_1(x)=\cdots=y_\ell f_\ell(x)=1\}. \]
    In fact, $\hat{\mu}_{|\hat{X}_+}: \hat{X}_+ \rightarrow {\rm int}(\hat{P}(s))$ is a diffeomorphism. 
    We now relate this to the moment map $\mu$ in \eqref{eq:momentmap}. To identify the domains of $\mu$ and $\hat{\mu}$, we introduce the map $\kappa: \mathbb{R}^n_+ \rightarrow \hat{X}_+$ with $\kappa(x) = (f_1(x)^{-1}, \ldots, f_\ell(x)^{-1}, x_1, \ldots ,x_n)$. For the co-domains, note that $\iota:{\rm int}(P(s)) \rightarrow {\rm int}(\hat{P}(s))$ with $\iota(v)=  (s_1,\dots,s_\ell,v)$ is an isomorphism. We obtain a diagram of diffeomorphisms
    \[
    \xymatrix{
        \hat{X}_+\ar[r]^-{\hat{\mu}}& \, \, {\rm int}(\hat{P}(s))\\
        \mathbb{R}^n_+\ar[u]^-{\kappa}\ar[r]^-{\mu}& \, \, {\rm int}({P}(s))\ar[u]_{\iota}
    }.
    \]
    The Jacobian matrix of $\mu$ is positive definite on $\mathbb{R}^n_+$ if and only if the \emph{toric Jacobian matrix} 
    \[\left( x_j\frac{\partial\mu_k}{\partial x_j}\right)_{j,k} \, = \, \left ( - x_j  \frac{\partial}{\partial x_j} \left ( x_k\frac{\partial}{\partial x_k} \log L (x) \right) \right)_{j,k} \, = \, \sum_{i=1}^\ell s_i \left ( x_j  \frac{\partial}{\partial x_j} \left ( x_k\frac{\partial}{\partial x_k}   \log f_i\right) \right)_{j,k}. \] 
    is positive definite on $\mathbb{R}^n_+$.
    The positivity follows from Lemma \ref{lem:convexity}.
\end{proof}

Note that Lemma \ref{lem:moment_map} implies that, under our assumptions, ${\rm Crit}(\log L) \cap \mathbb{R}^n_+ = \mu^{-1}(\nu)$ consists of a single point $\{a \}$. This is the first claim in Theorem \ref{thm:main2.3}. While Theorem \ref{thm:main2.3} uses the fiber $\mu^{-1}(\nu)$, Theorem \ref{thm:main2.2} sums over the fiber $\mu_{\mathbb{C}}^{-1}(\nu)$ of the complexified map $\mu_{\mathbb{C}}$.

\begin{proof}[Proof of Theorem \ref{thm:main2.3}]
We have established that ${\rm Crit}(\log L)\cap\mathbb{R}^n_+ = \{a\}$ is a singleton (Lemma \ref{lem:moment_map}). Let $U$ be a small open neighborhood of the positive critical point $a$.
The idea of the proof is to decompose ${\cal I}(\delta)$ into two parts:
\[
\delta^{\frac{n}{2}}L(a)^{-\frac{1}{\delta}}{\cal I}(\delta) \, = \, \delta^{-\frac{n}{2}}L(a)^{-\frac{1}{\delta}}\left( \int_{\mathbb{R}^n_+\setminus U} L(x)^{\frac{1}{\delta}}\frac{dx}{x} +\int_U L(x)^{\frac{1}{\delta}}\frac{dx}{x} \right) .
\]
We will show that the integral over $\mathbb{R}^n_+ \setminus U$ does not contribute to the limit $\delta\to 0^+$, and the integral over $U$ gives rise to a Gaussian integral. 
By Lemma \ref{lem:convexity}, $\log L$ attains its unique global maximum at $x=a$ and $\log L(e^{z_1},\dots,e^{z_n})$ is a concave function of $(z_1,\dots,z_n)\in\mathbb{R}^n$ (see Figure \ref{fig:likelihoodmax}). 
There exists a positive number $\varepsilon$ such that the inequality $\log L(x)-\log L(a)\leq -\varepsilon$ is true for $\mathbb{R}_+^n\setminus U$.
We obtain the following inequality for $0<\delta<1$:
\begin{align}
    L(a)^{-\frac{1}{\delta}}\int_{\mathbb{R}^n_+\setminus U}{L(x)}^{\frac{1}{\delta}}\frac{{\rm d}x}{x}&=e^{-\frac{\varepsilon}{\delta}}\int_{\mathbb{R}^n_+\setminus U}{\rm exp} \left (\delta^{-1}(\log L(x)-\log L(a)+{\varepsilon}) \right) \, \frac{{\rm d}x}{x}\\
    &\leq e^{-\frac{\varepsilon}{\delta}}\int_{\mathbb{R}^n_+\setminus U}{\rm exp} \left (\log L(x)-\log L(a)+{\varepsilon} \right) \, \frac{{\rm d} x}{x}. \label{eq:38}
\end{align}
Here, we used the fact that $\log L(x)-\log L(a)+\varepsilon$ is negative for any $x\in \mathbb{R}^n_+ \setminus U$, and the final integral in \eqref{eq:38} is bounded because of Theorem \ref{thm:convergence}.
The inequality \eqref{eq:38} shows that the integral over $\mathbb{R}^n_+\setminus U$ converges to zero as $\delta\to0^+$.

Let $P$ be an orthogonal matrix which diagonalizes the Hessian matrix of $\log L$: 
\[ P^\top \cdot \left(\frac{\partial^2\log L}{\partial x_j\partial x_k}(a)\right)_{j,k} \cdot P \, = \,  D.\]
Here $D$ is an $n\times n$ diagonal matrix, with negative diagonal entries. 
We perform a linear change of coordinates $y=P^T(x-a)$.
The Taylor expansion of $\log L(x)-\log L(a)$ around $y = 0$ looks like $\frac{1}{2}(y^TDy+r(y))$. Plugging this into our integral gives
\begin{equation*}
    L(a)^{-\frac{1}{\delta}}\int_{U}{L(x)}^{\frac{1}{\delta}}\frac{dx}{x} \, = \, \int_{P^T (U-a)}\exp\left( \frac{y^TDy+r(y)}{2\delta}\right)\frac{dy}{\prod_i(Py+a)_i}.
\end{equation*}
The last denominator is the product of the entries of $Py + a$. 
Without loss of generality, we may assume that $P^T(U-a)$ is a product of small intervals $(-\epsilon,\epsilon)^n$.
Replacing $y_i$ with $y_i/\sqrt{\delta}$, the last integral becomes
\[
\delta^{\frac{n}{2}}\int_{(-\epsilon/\sqrt{\delta},\epsilon/\sqrt{\delta})^n}
\exp\left(\frac{1}{2}y^TDy+\frac{r(\sqrt{\delta}y)}{\delta}\right)\frac{dy}{\prod_i(\sqrt{\delta}Py+a)_i}.
\]
The function $r(\sqrt{\delta}y)/\delta$ is bounded for $0<\delta<1$ and $y\in (-\epsilon/\sqrt{\delta},\epsilon/\sqrt{\delta})^n$ and it converges to $0$ when $\delta$ tends to $0$.
Therefore, Lebesgue's dominance convergence theorem proves
\[
\lim_{\delta\to0^+}\delta^{-\frac{n}{2}}L(a)^{-\frac{1}{\delta}}\int_{U}{L(x)}^{\frac{1}{\delta}}\frac{dx}{x}
\, = \, 
\int_{\mathbb{R}^n}e^{\frac{1}{2}y^TDy}\frac{dy}{a_1\cdots a_n}.
\]
This leaves us with a Gaussian integral. To finish the proof, recall that 
\[ \int_{-\infty}^\infty e^{\frac{\lambda}{2} \, z^2} \, {\rm d} z \, = \, \sqrt{\frac{2 \pi}{-\lambda}}, \quad \text{ for } \lambda < 0,\]
and use the fact that $H_{-\log L}(a) = (a_1 \cdots a_n)^2 \cdot \prod_i (-D_{ii})$. 
\end{proof}

\begin{figure}[H]
    \centering
    \includegraphics[height = 6cm]{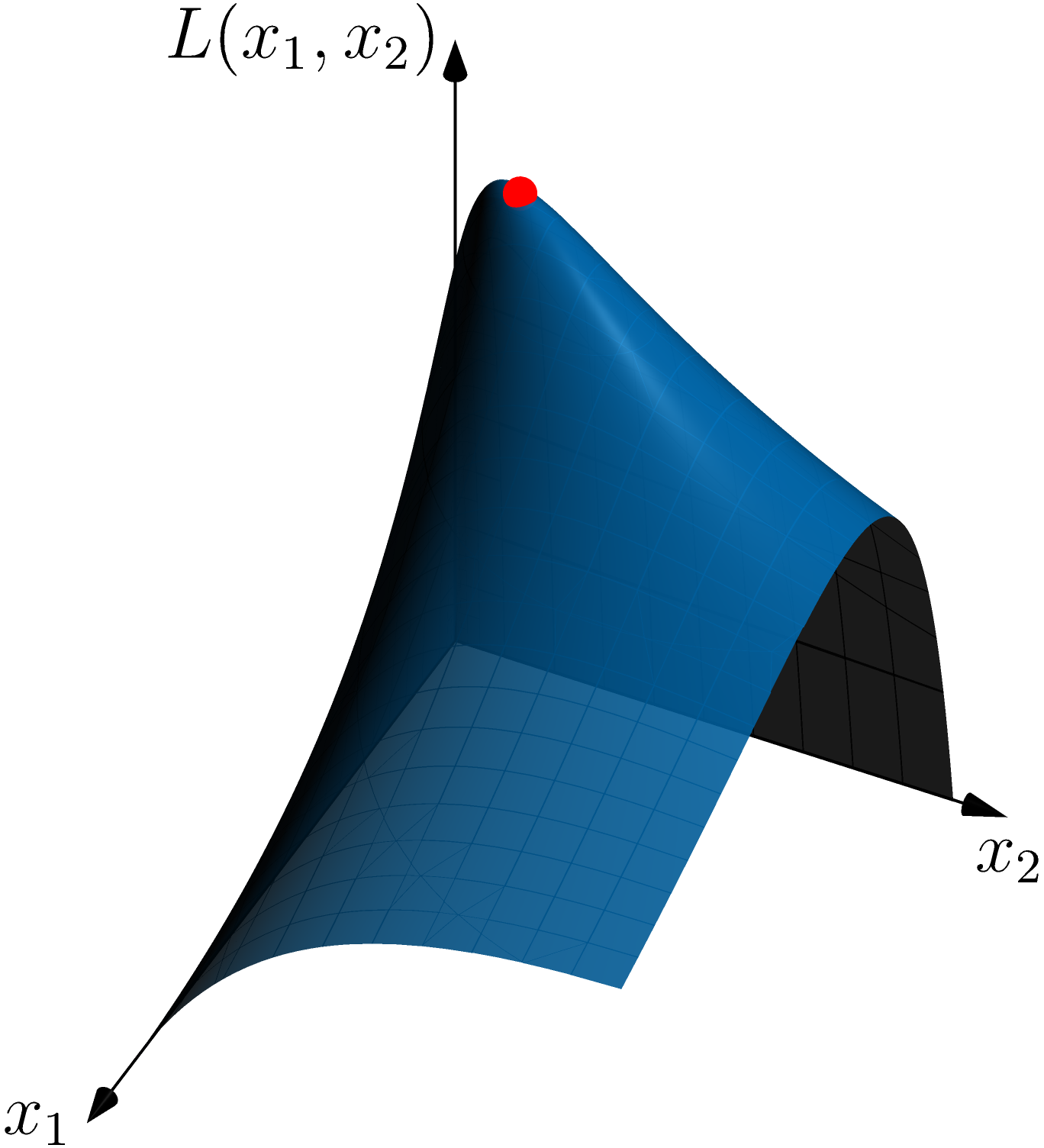}
    \caption{The likelihood function of Example \ref{ex:2.4} attains its maximum at $(x_1,x_2) \approx(0.618,1)$.}
    \label{fig:likelihoodmax}
\end{figure}

\begin{example}
    Let us verify the formula \eqref{eq:high_energy_limit} for the integral representation \eqref{eq:positive_Beta} of the beta function.
    The integral \eqref{eq:positive_Beta} is expressed by Gamma functions as in \eqref{eq:betafunc}:
    \[
    {\cal I}(\delta) \, = \, B\left(\frac{\Tilde{s}-\nu}{\delta},\frac{\nu}{\delta}\right)\, =\, \frac{\Gamma(\frac{\Tilde{s}-\nu}{\delta})\Gamma(\frac{\nu}{\delta})}{\Gamma(\frac{\Tilde{s}}{\delta})}.
    \]
    The function $L(a)^{-\frac{1}{\delta}}$, where $a$ is the unique critical point $\frac{\nu}{\Tilde{s}-\nu}$ from \eqref{ex:beta2.2}, is given by
    \[
    \left(\left(\frac{\Tilde{s}}{\Tilde{s}-\nu }\right)^{-\Tilde{s}} \left(\frac{\nu }{\Tilde{s}-\nu
   }\right)^{\nu }\right)^{-\frac{1}{\delta} }
    \, .\]
    Stirling's formula $\Gamma(x)\sim \sqrt{2\pi}e^{-x}x^{x-\frac{1}{2}}$ $(x\to+\infty)$ shows that the left-hand side of \eqref{eq:high_energy_limit} is given by $\sqrt{\frac{\tilde{s}}{\nu(\tilde{s}-\nu)}}$. We have seen in Example \ref{ex:beta2.2} that this equals $H_{-\log L}(a)^{-\frac{1}{2}}$.
\end{example}

\section{Twisted (co)homology}
\label{sec:3}
In this section, we abandon the concrete integration contour $\mathbb{R}^n_+$, and we drop Assumption \ref{assum:poscoeffs}. We fix $f = (f_1, \ldots, f_\ell) \in \mathbb{C}[x_1^{\pm 1}, \ldots, x_n^{\pm 1}]^\ell$, $s = (s_1, \ldots, s_\ell) \in \mathbb{C}^\ell$ and $\nu = (\nu_1, \ldots, \nu_n) \in \mathbb{C}^n$.
The perspective we take is that the Euler integral \eqref{eq:integralintro} is the result of a pairing between the integration contour $\Gamma$ and the differential $n$-form $\frac{{\rm d} x}{x}$. More generally, an $n$-form $\phi$ gives 
\begin{equation} \label{eq:pairing} 
\langle \Gamma, \phi \rangle \, = \, \int_\Gamma f^{-s} x^{\nu} \, \phi. 
\end{equation}
This works nicely when $\Gamma$ is a \emph{twisted $n$-cycle} and $\phi$ is a \emph{twisted $n$-cocycle}. We will introduce these concepts, and see that the pairing \eqref{eq:pairing} is a perfect pairing of finite dimensional $\mathbb{C}$-vector spaces. In particular, the integral \eqref{eq:pairing} always evaluates to a (finite) complex number.

This story is reminiscent of the classical duality between \emph{singular homology} and \emph{de Rham cohomology}, where one pairs an integration contour $\Delta$ on a complex manifold $X$ with a differential form $\phi$ by evaluating $\int_\Delta \phi$ \cite[Chapter 0]{griffiths2014principles}. In our setting, $X$ is the very affine variety seen in \eqref{eq:veryaffinevariety}. The material in this section is like that standard theory, but with a \emph{twist}. For instance, recall from the beginning of Section \ref{sec:1} that our integrand $f^{-s}x^\nu$ is multi-valued. To make sense of the integral \eqref{eq:pairing}, we need to specify a branch. We will see that this information is carried by our twisted cycle $\Gamma$. Also, a central role in de Rham's cohomology theory is played by Stokes' theorem, which says that for an $(n-1)$-form $\psi$,
\begin{equation} \label{eq:classic stokes}
    \int_\Delta {\rm d} \psi \, = \, \int_{\partial \Delta} \psi .
\end{equation}
In our setting, the \emph{twisted} boundary operator $\partial_\omega$ takes the choice of branch into account, and the \emph{twisted differential} $\nabla_\omega$ replaces the ordinary differential ${\rm d}$ to accommodate our integrals:
\begin{equation} \label{eq:stokestwisted}
\int_\Gamma f^{-s} x^\nu \, \nabla_\omega \psi \, = \, \int_{\partial_\omega \Gamma} f^{-s} x^\nu \, \psi. 
\end{equation}
The meaning of the index $\omega$ will become clear soon. For now, it simply indicates the twist.  

The theory of twisted (co)homology goes back to the seminal work of Deligne and Grothendieck \cite{deligne2006equations}. It has been investigated in the context of Euler integrals and hypergeometric functions by several authors, among which we mention Aomoto, Gelfand, Iwasaki, Kapranov, Kita, Matsumoto and Zelevinsky. See \cite{aomoto2011theory,gelfand1990generalized} and references therein. The relevance of this theory in particle physics was first realized by Mastrolia and Mizera \cite{mastrolia2019feynman}. 

The section is organized as follows. We start by discussing twisted chains and cycles, leading to a twisted version of the usual chain complex of $X$. Next, we switch to the dual complex, called the \emph{twisted de Rham complex} of $X$. We discuss properties of the (co)homology of these complexes, ultimately leading to the perfect pairing in \eqref{eq:pairing}.

\subsection{Twisted chains} \label{sec:3.1}
 
Throughout the section, $X$ is the very affine variety from \eqref{eq:veryaffinevariety}.
A \emph{singular $k$-simplex} $\Delta$ in $X$ is a continuous (not necessarily injective) map from the standard $k$-simplex to $X$. The $\mathbb{C}$-vector space generated by singular $k$-simplices is the space of \emph{singular $k$-chains}, denoted~by
\begin{equation} \label{eq:singchains}
 C_k(X) \, = \, \bigoplus_{\Delta \subset X, \text{ $k$-simplex}} \mathbb{C} \cdot \Delta.
 \end{equation}
\begin{example} \label{ex:chains}
    Examples of $1$-simplices $X = \mathbb{C} \setminus \{0,1\}$ are illustrated in Figure~\ref{fig:chains}. Here $\varepsilon \in \mathbb{R}$ lies in $(0,1/2)$. There are 5 simplices in total. Four of them are semicircles, parameterized~by
    \[\Delta_{a,\theta}(\varepsilon) \, = \, \{\, t \mapsto a + \varepsilon \cdot {\rm exp}(\sqrt{-1}(\theta + t\pi) \, \}, \quad t \in [0,1].  \]
    Here, $a$ is $0$ or $1$ and $\theta$ is $0$ or $\pi$.
    The remaining simplex is the line segment $[\varepsilon,1-\varepsilon]$, parameterized by $t \mapsto (1-t)\varepsilon + t(1-\varepsilon)$. These parameterizations fix the orientations visualized by the arrows in Figure \ref{fig:chains}.
    \begin{figure}[!t]
\centering
\begin{tikzpicture}[scale=3]
\draw[->, draw=lightgray] (-1.5,0) -- (0.5,0);
\draw[->, draw=lightgray] (-1,-0.5) -- (-1,0.5);
\draw[gray, thick, decorate,decoration=zigzag] (0,0) -- (0,-0.5);
\draw[gray, thick, decorate,decoration=zigzag] (-1,0) -- (-1,-0.5);
\filldraw[black, thick] (0,0) circle (0.02) node[anchor=south west]{\footnotesize $1$};
\filldraw[black] (-1,0) circle (0.02) node[anchor=south west]{\footnotesize $0$};
\draw[->, very thick, draw=RoyalBlue] (-0.8,0) -- (-0.2,0);
\draw[->, very thick, draw=RoyalBlue] (-0.2,-0.03) arc (-170:0:0.2);
\draw[->, very thick, draw=RoyalBlue] (0.2,0.02) arc (0:175:0.2);
\draw[->, very thick, draw=RoyalBlue] (-0.8,0.03) arc (10:175:0.2);
\draw[->, very thick, draw=RoyalBlue] (-1.2,-0.02) arc (180:355:0.2);
\node[] at (-0.51,0.15) {\footnotesize $[\varepsilon, 1-\varepsilon]$};
\node[] at (0.35,0.25) {\footnotesize $\Delta_{1,0}(-\varepsilon)$};
\node[] at (0.35,-0.25) {\footnotesize $\Delta_{1,\pi}(-\varepsilon)$};
\node[] at (-1.35,0.25) {\footnotesize $\Delta_{0,0}(\varepsilon)$};
\node[] at (-1.35,-0.25) {\footnotesize $\Delta_{0,\pi}(\varepsilon)$};
\end{tikzpicture}
\caption{\label{fig:chains}Five simplices in $\mathbb{C} \setminus \{0,1\}$.}
\end{figure}
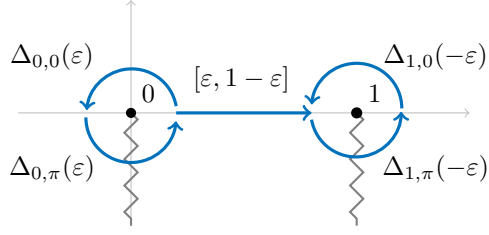
\end{example}

We need to modify this standard construction to account for multi-valuedness of $f^{-s} x^\nu$.
The branches of $f^{-s}x^\nu$ generate the space of sections of a line bundle ${\cal L}_{-\omega}$ on $X$ called a \emph{local system}. On an open subset $U \subset X$, these sections are 
\begin{equation}  \label{eq:diffLomega} 
{\cal L}_{-\omega}(U) \, = \, \{ \tau: U \rightarrow \mathbb{C} \, :\, \tau \text{ is holomorphic and }  {\rm d}\tau - {\rm dlog} (f^{-s}x^\nu) \, \tau = 0 \}.
\end{equation}
One checks that each branch $\tau$ of $f^{-s}x^\nu$ indeed satisfies the equation ${\rm d}\tau - {\rm dlog} (f^{-s}x^\nu) \, \tau = 0$. The one-form ${\rm dlog}(f^{-s}x^\nu)$ is of crucial importance in this section, so we introduce the notation $\omega = {\rm dlog}(f^{-s}x^\nu)$. We have $\omega = g_1 {\rm d}x_1 + \cdots + g_n {\rm d}x_n$, where $g_j$ are the rational functions in \eqref{eq:criteq}. The symbol ${\cal L}_{-\omega}$ stresses the term $-\omega$ in the operator applied to $\tau$ in \eqref{eq:diffLomega}. 

Picking a (linear combination of) branch(es) of $f^{-s}x^\nu$ on a singular $k$-simplex $\Delta$ means picking a section $\tau$ of ${\cal L}_{-\omega}$ on a sufficiently small open subset $U \supset \Delta$. We formalize this intuition by considering the \emph{direct limit} 
\begin{equation}\label{eq:direct_limit}
{\cal L}_{-\omega}(\Delta) \, = \, \lim_{\substack{\longrightarrow \\ U \supset \Delta}} {\cal L}_{-\omega}(U).
\end{equation} 
Here the open sets $U$ containing $\Delta$ are ordered by inclusion, and when $U \subset U'$, the map ${\cal L}_{-\omega}(U') \rightarrow {\cal L}_{-\omega}(U)$ is given by restriction. A reader who is unfamiliar with direct limits can simply think of elements in ${\cal L}_{-\omega}(\Delta)$ as branches of $f^{-s}x^\nu$, restricted to $\Delta$.

\begin{example} \label{ex:3.2}
    We continue Example \ref{ex:chains}. On each of the five simplices we define a section of ${\cal L}_{-\omega}$, where $\omega = (-s(1-x)^{-1} + \nu\,x^{-1})\,{\rm d} x$ is the logarithmic differential of $f^{-s}x^\nu = (1-x)^{-s}x^{\nu}$. Here $s$ and $\nu$ are fixed complex numbers. Notice that $f^{-s}x^\nu$ is the integrand of the beta function in \eqref{eq:betafunc}. At $x = \varepsilon$, both $(1-x)$ and $x$ are positive. Let $\zeta = {\rm exp}(-s \log(1-\varepsilon) + \nu \log \varepsilon) \in \mathbb{C}$, where we evaluate the positive branch of the logarithm. The conditions
    \begin{equation} \label{eq:IVP}
    {\rm d}\tau_{0,0} - \omega \tau_{0,0} \, = \, 0 \quad \text{ and } 
    \quad \tau_{0,0}(\varepsilon) = \zeta 
    \end{equation}
    uniquely define $\tau_{0,0} \in {\cal L}_{-\omega}(\Delta_{0,0}(\varepsilon))$. Constraints like $\tau_{0,0}(\varepsilon) = \zeta$ are called \emph{initial conditions}. On our other simplices, we choose $\tau_{a,\theta} \in {\cal L}_{-\omega}(\Delta_{a,\theta}(\varepsilon))$, $\tau_{-} \in {\cal L}_{-\omega}([\varepsilon,1-\varepsilon])$, with
    \[ \tau_{0,\pi}(-\varepsilon) = \tau_{0,0}(-\varepsilon), \,\, \tau_{-}(\varepsilon) = \tau_{0,0}(\varepsilon),\,\, \tau_{1,\pi}(1-\varepsilon) = \tau_{-}(1-\varepsilon), \,\, \tau_{1,0}(1+\varepsilon) = \tau_{1,\pi}(1+\varepsilon).\]
    To compute these boundary values, one can make use of the parameterizations in Example \ref{ex:chains}. For instance, $\tau_{0,0}(\varepsilon \exp(\sqrt{-1} t\pi )$ is given by $(1-\varepsilon \exp(\sqrt{-1} \pi t))^{-s} \varepsilon^\nu \exp(\sqrt{-1} \nu t \pi)$, for $t \in [0,1]$. At $t = 1$, this gives $\tau_{0,0}(-\varepsilon) = (1+\varepsilon)^{-s}\varepsilon^\nu \exp(\sqrt{-1}\pi \nu)$. A similar computation for $\tau_{0,\pi}$ shows that $\tau_{0,\pi}(\varepsilon) = {\rm exp}(\sqrt{-1}\nu 2\pi) \, \zeta$. In particular, $\tau_{0,0}(\varepsilon) \neq \tau_{0,\pi}(\varepsilon)$, for non-integer $\nu$.

    We remark that having nontrivial sections of ${\cal L}_{-\omega}$ is the reason why we split up the circle $S_0(\varepsilon) = \{ t \mapsto \varepsilon \exp(\sqrt{-1}t 2\pi)\}, t \in [0,1]$ into two semicircles $\Delta_{0,0}(\varepsilon)$ and $\Delta_{0,\pi}(\varepsilon)$. Indeed, because of the nontrivial monodromy around $x=0$, there are no nonzero holomorphic solutions of ${\rm d}\tau - \omega \tau = 0$ on $S_0(\varepsilon)$: ${\cal L}_{-\omega}(S_0(\varepsilon)) = 0$. 
\end{example}

This gives us all ingredients to define the space of \emph{twisted $k$-chains} on $X$, with twist $\omega$: 
\begin{equation} \label{eq:twistedchains}
    C_k(X,-\omega) \, = \, \bigoplus_{\Delta \subset X, \text{ $k$-simplex}} \Delta \otimes {\cal L}_{-\omega}(\Delta). 
\end{equation}
Comparing \eqref{eq:twistedchains} with \eqref{eq:singchains} motivates why this  construction is sometimes called the space of $k$-chains \emph{with coefficients in ${\cal L}_{-\omega}$}. In words, $C_k(X,-\omega)$ consists of finite $\mathbb{C}$-linear combinations of elements of the form $\Delta \otimes \tau$, where $\tau : \Delta \rightarrow \mathbb{C}$ is an element of ${\cal L}_{-\omega}(\Delta)$. We say that $\Delta$ is \emph{loaded} with the branch $\tau$. We note that $C_k(X,-\omega)$ is non-zero for any $k\geq 0$.

Let us now clarify the meaning of \eqref{eq:pairing} for a twisted $k$-chain $\Gamma \in C_k(X,-\omega)$. 
\begin{definition} \label{def:integral}
    Let $\Gamma =\sum_p d_p \, \Delta_p \otimes \tau_p \in C_k(X,-\omega)$ be a twisted $k$-chain on $X$, with $d_p \in \mathbb{C}$. Let $\phi = g(x) \, {\rm d} x$ be a holomorphic $k$-form on $X$. We define
\begin{equation} \label{eq:meaningofintegral}
\langle \Gamma, \phi \rangle \, = \, \int_{\Gamma} f^{-s} x^\nu \, \phi \, = \, \sum_p d_p \, \int_{\Delta_p \otimes \tau_p} f^{-s} x^\nu \, \phi \, = \, \sum_p d_p \, \int_{\Delta_p} \tau_p(x) \, \phi \, . 
\end{equation}
The integrals on the right are the familiar integrals of single valued $k$-forms on $k$-simplices.
\end{definition}
We now explain how to take boundaries in this twisted setting. A $k$-simplex $\Delta$ on $X$ has boundary $\partial \Delta = \partial \Delta_0 + \cdots + \partial \Delta_{k} \in C_{k-1}(X)$. That is, if $\Delta$ is given by the parameterization $\varphi: \Delta \rightarrow X$ with $\Delta$ the standard $k$-simplex in $\mathbb{R}^n$, then $\partial \Delta_{i}$ is the $(k-1)$-simplex in $X$ coming from the restriction of $\varphi$ to the $i$-th boundary component of $\Delta$. Let $\Gamma = \Delta \otimes \tau \in C_k(X,-\omega)$ be a $k$-simplex $\Delta$ on $X$, loaded with $\tau$. Observe that there is a natural restriction map $\rho_{\Delta,\Delta'}: {\cal L}_{-\omega}(\Delta) \rightarrow {\cal L}_{-\omega}(\Delta')$ whenever $\Delta' \subset \Delta$. The \emph{twisted boundary} $\partial_\omega(\Gamma)$ of $\Gamma$ is 
\[ \partial_\omega(\Gamma) \, = \,  \partial \Delta_0 \otimes \rho_{\Delta,\partial \Delta_0}(\tau) + \cdots + \partial \Delta_{k} \otimes \rho_{\Delta,\partial \Delta_{k}}(\tau). \]
Extending this $\mathbb{C}$-linearly gives the \emph{twisted boundary operator}
\begin{equation}
\partial_{\omega}: \quad C_k(X, -\omega) \, \longmapsto \,  C_{k-1}(X, -\omega) .
\end{equation}
Morally, this boundary operator simply keeps track of the branch of the twisted chain. 

\begin{example}\label{ex:3.7}
Consider again the five simplices illustrated in Figure~\ref{fig:chains}, loaded with the branches specified in Example \ref{ex:3.2}. The twisted boundaries are
\begin{align} \label{eq:boundaries}
    \begin{split}
        \partial_\omega(\Delta_{0,0}(\varepsilon) \otimes \tau_{0,0}) &\, = \, \{-\varepsilon\} \otimes \tau_{0,0}(-\varepsilon) - \{ \varepsilon \} \otimes \tau_{0,0}(\varepsilon), \\
        \partial_\omega(\Delta_{0,\pi}(\varepsilon) \otimes \tau_{0,\pi}) &\, =\, 
        \{\varepsilon\} \otimes \tau_{0,\pi}(\varepsilon) - \{ -\varepsilon \} \otimes \tau_{0,\pi}(-\varepsilon), \\
        \partial_\omega([\varepsilon,1-\varepsilon] \otimes \tau_-) &\, =\, \{1-\varepsilon\} \otimes \tau_-(1-\varepsilon) - \{\varepsilon\} \otimes \tau_-(\varepsilon), \\
        \partial_\omega(\Delta_{1,\pi}(\varepsilon) \otimes \tau_{1,\pi}) & \, = \, \{1+\varepsilon\} \otimes \tau_{1,\pi}(1+\varepsilon) - \{1-\varepsilon\} \otimes \tau_{1,\pi}(1-\varepsilon), \\ 
        \partial_\omega(\Delta_{1,0}(\varepsilon) \otimes \tau_{1,0}) & \, = \, \{1-\varepsilon\} \otimes \tau_{1,0}(1-\varepsilon) - \{1+\varepsilon\} \otimes \tau_{1,0}(1+\varepsilon). \\ 
    \end{split}
\end{align}
The orientation of the boundary components of a 1-simplex is like in standard singular homology: end point minus starting point. The restrictions of our sections to these boundary points are simply given by their value at the point. It is instructive to reduce the number of parameters in \eqref{eq:boundaries} by using our definitions and findings from Example \ref{ex:3.2}. For instance, we have $\tau_{0,\pi}(-\varepsilon) = \tau_{0,0}(-\varepsilon)$, $\tau_{-}(\varepsilon) = \tau_{0,0}(\varepsilon)$, $\tau_{0,\pi}(\varepsilon) = \exp(\sqrt{-1}\nu 2\pi) \, \tau_{0,0}(\varepsilon)$, and so on.
\end{example}

Notice that, for a twisted $k$-chain $\Delta \otimes \tau$, we have $\partial_\omega \partial_\omega (\Delta \otimes \tau) = 0$. This follows from the fact that the boundary of a boundary is empty, i.e., $\partial \partial \Delta = 0$, and the fact that $\partial_\omega$ simply restricts $\tau$ to $\partial \Delta$. We suggest that the reader checks this carefully for a two-dimensional simplex in $X = \mathbb{C} \setminus \{0,1\}$ from our running example. In homological algebra, $\partial_\omega \partial_\omega = 0$ is the key property of a boundary operator in a \emph{chain complex}. 
\begin{definition} \label{def:twistedchain}
Let $X$ be the very affine variety from \eqref{eq:veryaffinevariety}. Let $\omega = {\rm dlog}(f^{-s} x^\nu)$, and let $C_k(X,-\omega)$ be the space \eqref{eq:twistedchains} of twisted $k$-chains on $X$. The \emph{twisted chain complex} is
\begin{equation} \label{eq:chaincomplex}
(C_\bullet(X, -\omega), \partial_{\omega}): \; \cdots \xrightarrow[]{\phantom{\partial_\omega}} C_{k}(X, -\omega) \xrightarrow[]{\partial_\omega} C_{k-1}(X, -\omega) \xrightarrow[]{\partial_\omega} \cdots \xrightarrow[]{\partial_\omega} C_0(X,-\omega) \xrightarrow[]{\phantom{\partial_\omega}} 0 \, .
\end{equation}
\end{definition}
The \emph{homology} of this complex is obtained by considering all twisted chains whose twisted boundary is zero, modulo those that are twisted boundaries themselves.
\begin{definition} 
	The \emph{$k$-th homology vector space} of $(C_\bullet(X, -\omega), \partial_{\omega})$ is the quotient space
	\begin{equation}
	H_k(X, -\omega) \,= \, \frac{\{ \Gamma \in C_k(X, -\omega) \, :\, \partial_\omega (\Gamma) = 0 \} }{ \partial_\omega C_{k+1}(X, -\omega)}\, .
	\end{equation}
\end{definition}
Elements of $H_k(X,-\omega)$ are called \emph{twisted $k$-cycles} (or sometimes \emph{loaded $k$-cycles}).
We will primarily be interested in the $n$-th homology space $H_n(X,-\omega)$, because these are the cycles on which we can integrate $n$-forms. While it is easy to construct cycles in the usual (non-twisted) singular homology, this is a bit more complicated in our twisted setting. The running example of this section illustrates a standard construction \cite[Section 3.2.4]{aomoto2011theory}.

\begin{example}\label{ex:3.10}
None of the five twisted chains in Example~\ref{ex:3.7} are twisted cycles, since they have non-zero twisted boundaries. However, we can use the expressions \eqref{eq:boundaries} to find a linear combination of these chains whose twisted boundary is zero. For ease of notation, let us write $\Gamma_{a,\theta} = \Delta_{a,\theta}(\varepsilon) \otimes \tau_{a,\theta} \in C_1(X,-\omega)$, and $\Gamma_{-} = [\varepsilon,1-\varepsilon] \otimes \tau_- \in C_1(X,-\omega)$. Consider
\begin{equation} \label{eq:twisted-cycle} \Gamma \, = \, \frac{\Gamma_{0,0}+\Gamma_{0,\pi}}{\exp(2\pi\sqrt{-1}\nu ) - 1} \, + \, \Gamma_- \, - \, \frac{\Gamma_{1,\pi} + \Gamma_{1,0}}{\exp(-2 \pi\sqrt{-1}s ) - 1} \, \, \in \, C_1(X,-\omega).
\end{equation}
Note that \eqref{eq:twisted-cycle} only makes sense when $\nu$ and $s$ are non-integer. This \emph{genericity} assumption will appear in our theorems below. One checks that $\partial_\omega(\Gamma) = 0$ by expanding it using \eqref{eq:boundaries}, and applying identities like at the end of Example \ref{ex:3.7}. The class $[\Gamma] \in H_1(X,-\omega)$ of $\Gamma$ is non-zero. We will show this in Example \ref{ex:gammanonzero}. Hence, $\Gamma$ is not a boundary of a $2$-chain.
\end{example}

\subsection{Twisted cochains} \label{sec:3.2}

While chains tell us \emph{where to integrate}, co-chains tell us \emph{what to integrate}. This is to be taken with a grain of salt in our twisted setting. We have seen above that twisted chains also carry some information about the integrand: they specify a branch of $f^{-s}x^\nu$. In \eqref{eq:pairing}, the multivalued function $f^{-s}x^\nu$ is multiplied with an $n$-form $\phi = g \, {\rm d}x$, where $g$ is a single-valued function on $X$. This section explains which $n$-forms $\phi$ we consider. It constructs a(n algebraic) \emph{twisted de Rham complex}, dual to the twisted chain complex in Definition \ref{def:twistedchain}. 

The vector spaces $\Omega^k(X)$ of the twisted de Rham complex are rather easy to describe. 
They are the \emph{regular} $k$-forms on $X$, which have coefficients in the coordinate ring of $X$:
\begin{equation} \label{eq:regularforms}
\Omega^k(X) = \left\{  \sum_{1 \leq j_1 \leq \cdots \leq j_k \leq n} g_{j_1,\ldots, j_k} \mathrm{d} x_{j_1} \wedge \cdots \wedge \mathrm{d} x_{j_k} \, :\, g_{j_1,\ldots,j_k} \in \sum_{\substack{a \in \mathbb{Z}^{\ell} \\ b \in \mathbb{Z}^n}} \mathbb{C} \cdot f^a x^b \right\}\, .
\end{equation}
Notice that $\omega = {\rm d log}(f^{-s}x^\nu) \in \Omega^1(X)$. We will integrate regular $n$-forms $\phi \in \Omega^n(X)$. In particular, setting $ \phi = \frac{{\rm d}x}{x} \in \Omega^n(X)$ in \eqref{eq:meaningofintegral} gives our integral \eqref{eq:integralintro}. In analogy with the usual de Rham complex, we want to regard regular $k$-forms $\phi$ modulo those that integrate to zero in \eqref{eq:meaningofintegral} on a twisted cycle $\Gamma = \Delta \otimes \tau$. A first step towards formalizing this is the following important observation. 
\begin{lemma} \label{lem:modnablaomega}
    For any $\psi \in \Omega^{k-1}(X)$ and any twisted $k$-chain $\Gamma \in C_k(X,-\omega)$, we have
    \begin{equation}\label{eq:total-differential} 
\int_\Gamma f^{-s} x^{\nu}  (\mathrm{d} + \omega\wedge ) \psi \, = \, \int_\Gamma \mathrm{d} \left( f^{-s} x^{\nu} \psi \right) \, = \, \int_{\partial_\omega(\Gamma)} f^{-s} x^{\nu} \psi.
\end{equation}
\end{lemma}

\begin{proof}
    The first equality is checked by expanding ${\rm d}(f^{-s} x^{\nu} \psi) = {\rm d}(f^{-s}x^\nu) \, \psi + f^{-s}x^\nu \, {\rm d}\psi$. The second identity is Stokes' theorem \eqref{eq:classic stokes}. More precisely, if $\Gamma = \Delta \otimes \tau$ is a simplex loaded with $\tau$, the integral is $\int_{\Delta} {\rm d}(\tau(x) \psi) = \int_{\partial \Delta} \tau(x) \psi$, which agrees with \eqref{eq:total-differential} via Definition \ref{def:integral}.
\end{proof}
Equation \eqref{eq:total-differential} will be our twisted version of Stokes' theorem \eqref{eq:stokestwisted}, where the \emph{twisted differential} $\nabla_\omega$ is given by ${\rm d} + \omega \wedge$. That is, for any $0 \leq k \leq n$ we define 
\begin{equation}
\nabla_\omega : \Omega^k(X) \to \Omega^{k+1}(X)\, \quad \text{with} \quad \nabla_\omega(\phi) = \mathrm{d}\phi + \omega\wedge \phi.
\end{equation}
A regular $k$-form $\phi$ is \emph{closed} if its twisted differential is zero, i.e., $\nabla_\omega(\phi) = 0$. In particular, all $n$-forms are closed, since $\Omega^{n+1}(X) = 0$. A regular $k$-form $\phi$ is called \emph{exact} if it is the twisted differential of some $(k-1)$-form: $\phi = \nabla_\omega(\psi)$. 
Here is a consequence of Lemma \ref{lem:modnablaomega}.
\begin{lemma} \label{lem:modeverything}
    Let $\Gamma \in C_k(X,-\omega)$ be a twisted cycle and let $\phi \in \Omega^k(X)$ be a closed $k$-form, i.e., $\partial_\omega (\Gamma) = 0$ and $\nabla_\omega(\phi) = 0$. If $\Gamma$ is a twisted boundary or $\phi$ is exact, i.e., $\Gamma = \partial_\omega(\Gamma')$ for some $\Gamma' \in C_{k+1}(X,-\omega)$ or $\phi = \nabla_\omega(\psi)$ for some $\psi \in \Omega^{k-1}(X)$, we have $\int_\Gamma f^{-s}x^\nu \phi = 0$.
\end{lemma}
Every exact $k$-form is closed. Indeed, using $\mathrm{d} \mathrm{d} \psi = \mathrm{d} \omega = \omega\wedge \omega = 0$, we find that
\begin{equation*}
\nabla_\omega \nabla_\omega \psi \, = \,  \mathrm{d} \mathrm{d} \psi + \mathrm{d} \omega \wedge \psi - \omega \wedge \mathrm{d} \psi + \omega \wedge \mathrm{d} \psi + \omega \wedge \omega \wedge \psi \, = \, 0, \quad \text{ for any $\psi \in \Omega^{k-1}(X)$}.
\end{equation*}
Here ${\rm d}\omega =0$ because $\omega = {\rm dlog}(f^{-s} x^\nu)$. The property $\nabla_\omega \nabla_\omega = 0$ means that $\nabla_\omega$ defines a \emph{flat connection} on $X$.  Lemma \ref{lem:modeverything} says that an exact $k$-form $\phi$ satisfies $\langle \Gamma,\phi \rangle = 0$, for any $\Gamma \in H_k(X,-\omega)$. It is therefore natural to regard $k$-forms modulo the exact $k$-forms $\nabla_\omega(\Omega^{k-1}(X))$. This amounts to considering the cohomology of the following cochain complex. 
\begin{definition} \label{def:twisteddeRham}
Let $X$ be the very affine variety from \eqref{eq:veryaffinevariety}. Let $\omega = {\rm dlog}(f^{-s} x^\nu)$, and let $\Omega^k(X)$ be the space \eqref{eq:regularforms} of regular $k$-forms. The \emph{(algebraic) twisted de Rham complex} is
\begin{equation} \label{eq:deRham}
(\Omega^\bullet(X), \nabla_\omega): \; 0 \xrightarrow[]{\phantom{\nabla_\omega}} \Omega^0(X) \xrightarrow[]{\nabla_\omega} \Omega^1(X) \xrightarrow[]{\nabla_\omega} \cdots \xrightarrow[]{\nabla_\omega} \Omega^{n}(X) \xrightarrow[]{\phantom{\nabla_\omega}} 0.
\end{equation}
\end{definition}
This complex is also called the \emph{twisted cochain complex}, to emphasize its duality with \eqref{eq:chaincomplex} (see below). As said above, since exact forms integrate to zero, we pass to cohomology.

\begin{definition}
The \emph{$k$-th twisted cohomology vector space} of \eqref{eq:deRham} is the quotient space
\begin{equation*}
H^k(X, \omega) \, = \,  \frac{\{ \phi \in \Omega^k(X) \, : \,  \nabla_\omega (\phi) = 0 \} }{ \nabla_\omega \Omega^{k-1}(X)} \, = \, \frac{\text{closed $k$-forms}}{\text{exact $k$-forms}} \, .
\end{equation*}
\end{definition}
We regard $\phi$ in \eqref{eq:pairing} as a \emph{twisted cocycle}, i.e., an element of the $n$-th twisted cohomology
\begin{equation}
H^n(X, \omega) \, = \,  \frac{\Omega^n (X)}{\nabla_\omega \Omega^{n-1}(X)}\,.
\end{equation}

The twisted de Rham complex in Definition \ref{def:twisteddeRham} is called \emph{algebraic} because we work with the regular $k$-forms $\Omega^k(X)$ in the sense of algebraic geometry. One can build an analogous complex using \emph{holomorphic} $k$-forms, for which the coefficients $g_{j_1,\ldots, j_k}$ in \eqref{eq:regularforms} can be any holomorphic functions on $X$. By the Grothendieck-Deligne comparison theorem \cite[Corollaire 6.3]{deligne2006equations}, the cohomology of this holomorphic twisted de Rham complex is isomorphic to that of \eqref{eq:deRham}. Since our cocycles $\phi$ will be regarded as elements in this cohomology, it suffices to work with the algebraic complex \eqref{eq:deRham}. This is also the preferred setting for doing computations, because the regular $k$-forms \eqref{eq:regularforms} have a very concrete description. 

Here is an example of how to compute relations in twisted cohomology.
\begin{example}\label{eq:3.5}
Consider again the Euler beta integral \eqref{eq:betafunc}. The twisted differential is
\begin{equation} \label{eq:nablabeta}
\nabla_\omega = \mathrm{d} + \left( \frac{s}{1 - x} + \frac{\nu}{x}  \right) \mathrm{d} x \wedge\, .
\end{equation}
Applying this to $1 \in \Omega^0(X)$, we obtain the following equality in $H^1(X,\omega)$:
\[ \left [ \frac{\mathrm{d} x}{1-x} \right ] \, = \, \left [ \frac{-\nu}{s} \frac{{\rm d}x}{x} \right ]. \]
More generally, we shall derive in Section \ref{sec:4} that for $a, b \in \mathbb{Z}$, we have the relation
\begin{equation} \label{eq:betarelation}
\left [ \frac{x^b }{(1-x)^a} \frac{{\rm d}x}{x} \right ] \, = \, \left [  \frac{(1-s)_{-a}(\nu)_b}{(1+\nu-s)_{b-a}}\frac{{\rm d}x}{x} \right ].  
\end{equation}
Here, for a complex number $\gamma$ and an integer $a$, we used the following notation:
\[
(\gamma)_a:=
\begin{cases}
    \gamma(\gamma+1)\cdots(\gamma+a-1)&(a>0)\\
    1&(a=0)\\
    (\gamma-1)^{-1}(\gamma-2)^{-1}\cdots(\gamma+a)^{-1} &(a<0)
\end{cases}
. \qedhere
\]
\end{example}
While $\Omega^1(X)$ and $\nabla_\omega(\Omega^0(X))$ are infinite-dimensional $\mathbb{C}$-vector spaces, Example \ref{eq:3.5} claims that the quotient $H^1(X,\omega)$ is one-dimensional. Indeed, each regular $1$-form can be written as a constant multiple of $[\frac{{\rm d}x}{x}]$. This holds, at least, when $s, \nu, s-\nu$ are non-integer. We will now proceed towards the underlying theorem (Theorem \ref{thm:3.4}). First, we state a vanishing result for twisted cohomology. 
\begin{theorem}[Vanishing theorem]\label{thm:3.3}
Let $X$ be the very affine variety from \eqref{eq:veryaffinevariety}. There exists a dense open subset $U \subset \mathbb{C}^{\ell + n}$ such that, for each $(s, \nu) \in U$, we have  
\begin{equation}\label{eq:vanishing-theorem}
H^k(X, \omega) = 0 \quad\text{for all}\quad k \neq n, \quad \text{with} \quad \omega = {\rm dlog}(f^{-s}x^\nu).
\end{equation}
\end{theorem}
A description of the open subset $U \subset \mathbb{C}^{\ell + n}$ follows from the proof in \cite[Theorem~A.1]{Agostini:2022cgv}. In our running example, one can take $U = \{(s,\nu) \in \mathbb{C}^{\ell + n} \, :\, s, \nu, s-\nu \notin \mathbb{Z} \}$.

One of the consequences of this vanishing theorem is a geometric description of the dimension of $H^{n}(X,\omega)$. By \cite[Theorem 2.2]{aomoto2011theory}, the topological Euler characteristic $\chi(X)$ of the very affine variety $X$ is given by the alternating sum of cohomology dimensions:
\begin{equation}
\chi(X) \, = \,  \sum_{k=0}^{n} (-1)^{k} \dim_{\mathbb{C}} H^k(X, \omega)\, .
\end{equation}
Combined with \eqref{eq:vanishing-theorem}, this immediately gives us the following result.
\begin{theorem}\label{thm:3.4}
Let $X$ be the very affine variety from \eqref{eq:veryaffinevariety}. Fix $(s,\nu) \in U$, where $U \subset \mathbb{C}^{\ell + n}$ is as in Theorem \ref{thm:3.3}, and let $\omega = {\rm dlog}(f^{-s}x^\nu)$. We have 
$\dim_{\mathbb{C}} H^n(X,\omega) = |\chi(X)|$.
\end{theorem}
The Euler characteristic $\chi(X)$ also appeared in Theorem~\ref{thm:huh}: it is the number of critical points in $\mathrm{Crit}(\log L)$ (up to a sign). Hence, for generic $(s,\nu)$, we can compute the dimension of $H^n(X, \omega)$ using the homotopy continuation techniques explained in Section~\ref{sec:2.1}.
\begin{example} \label{ex:eulercharbeta}
 In the case of our running example, $X = \mathbb{C}^\ast \setminus \{0,1\}$ is topologically the Riemann sphere $S^2$ with three points removed. Using the inclusion-exclusion principle and $\chi(S^2) = 2$, $\chi(\mathrm{point}) =1$, we get $\chi(X) = -1$. Hence $\dim_{\mathbb{C}}H^1(X,\omega) = 1$. This confirms what we saw in Example \ref{eq:3.5}. A basis for $H^1(X,\omega)$ is $[\frac{\mathrm{d}x}{x}]$.
\end{example}
\begin{example}
Consider $m$-point string amplitudes, for which $X = {\cal M}_{0,m}$. The projection map ${\cal M}_{0,m} \to {\cal M}_{0,m-1}$ is obtained by dropping one of the marked points. The fiber of this map at a given configuration of $m-1$ distinct points in $\mathbb{P}^1$ is $\mathbb{P}^1$ with these $m-1$ points removed. The product property of the Euler characteristic for fibrations gives the recursion
\begin{equation}
\chi(\mathcal{M}_{0,m}) = \chi(S^2-\{m{-}1 \text{ points}\}) \cdot \chi(\mathcal{M}_{0,m-1})\, .
\end{equation}
By the same arguments as in Example~\ref{ex:eulercharbeta}, the first factor on the right-hand side is $3-m$. The endpoint of the recursion is $m=3$, for which $\chi(\mathcal{M}_{0,3}) = \chi(\text{point}) = 1$. We conclude
\begin{equation} \label{eq:eulercharM0m}
\chi(\mathcal{M}_{0,m}) = (-1)^{m-3} (m-3)!\,.
\end{equation}
We have seen this number for $m = 5$ in Section \ref{sec:2.1}, where we explained that it also counts bounded cells of hyperplane arrangements in $\mathbb{R}^{m-3}$. 
Equation \eqref{eq:eulercharM0m} implies that the dimension of $H^{m-3}(\mathcal{M}_{0,m},\omega)$ is $(m-3)!$, under the genericity assumptions of Theorem \ref{thm:3.4}. A basis consists of $(m-3)!$ regular $(m-3)$-forms. In the physics literature, it is common to use the so-called \emph{Parke-Taylor} basis, see, e.g., \cite[Definition~3.2]{Mizera:2017cqs} and \cite[Appendix~A]{Brown:2019wna}.
\end{example}

\subsection{Back to Euler integrals} \label{sec:3.3}

We have been using the shorthand notation $\langle \Gamma, \phi \rangle$ for our integrals, see Definition \ref{def:integral}. Equipped with the tools from Sections \ref{sec:3.1} and \ref{sec:3.2}, we can now formally introduce the pairing $\langle \cdot, \cdot \rangle$ as a bilinear map on homology and cohomology. This was alluded to in \eqref{eq:pairing}. 
\begin{theorem} \label{thm:perfectpairing}
Let $X$ be as in \eqref{eq:veryaffinevariety} and let $\omega = {\rm dlog}(f^{-s}x^\nu)$. For any $k$, the $\mathbb{C}$-bilinear~map
    \begin{equation} \label{eq:perpairing}
\langle \cdot, \cdot \rangle : H_k(X,-\omega) \times H^k(X,\omega) \longrightarrow \mathbb{C}, \quad ([\Gamma],[\phi]) \longmapsto \langle \Gamma, \phi \rangle = \int_\Gamma f^{-s}x^\nu \, \phi 
\end{equation}
is well-defined. Moreover, the induced maps $H_k(X, -\omega) \rightarrow H^k(X,\omega)^\vee$ and $H^k(X,\omega) \rightarrow H_k(X,-\omega)^\vee$ (see below) are isomorphisms. In other words, the pairing \eqref{eq:perpairing} is perfect. 
\end{theorem}
\begin{proof}
    Well-definedness follows from Lemma \ref{lem:modeverything}. The pairing is perfect by \cite[Lemma 2.9(1)]{aomoto2011theory}, using the Deligne-Grothendieck comparison theorem \cite[Corollaire 6.3]{deligne2006equations}.
\end{proof}
The pairing \eqref{eq:perpairing} is called the \emph{period pairing} between twisted homology and cohomology. In the theorem, $V^\vee = {\rm Hom}_{\mathbb{C}}(V, \mathbb{C})$ denotes the dual vector space of a $\mathbb{C}$-vector space $V$. The maps $H_k(X, -\omega) \rightarrow H^k(X,\omega)^\vee$ and $H^k(X,\omega) \rightarrow H_k(X,-\omega)^\vee$ are given by 
\[[\Gamma] \mapsto ([\phi] \mapsto \langle \Gamma, \phi \rangle), \quad \text{and} \quad [\phi] \mapsto ([\Gamma] \mapsto \langle \Gamma, \phi \rangle) \]
respectively. Notice that Theorem \ref{thm:perfectpairing} makes no assumptions on $s$ and $\nu$. We spell out three important implications (Corollaries \ref{cor:samedim}, \ref{cor:dimchi} and \ref{cor:relations}).
\begin{corollary} \label{cor:samedim}
Let $X, \omega$ be as above. For any $k$, $\dim_{\mathbb{C}} H_k(X, -\omega)=\dim_{\mathbb{C}} H^k(X, \omega)$.  
\end{corollary}
\begin{corollary} \label{cor:dimchi}
If $(s, \nu)$ lies in the open subset $U$ from Theorem \ref{thm:3.3}, the vanishing theorem extends to twisted homology: $H_k(X,-\omega) = 0$ when $k \neq n$, and $\dim_{\mathbb{C}} H_n(X,-\omega) = |\chi(X)|$.
\end{corollary}
This means that, when $(s,\nu) \in U$, we can find a set of $\chi = |\chi(X)|$ basis elements $[\phi_1], \ldots, [\phi_\chi]$ for $H^n(X,\omega)$, and a set of $\chi$ basis elements $[\Gamma_1], \ldots, [\Gamma_\chi]$ for $H_n(X,-\omega)$. In particular, for any $(a,b) \in \mathbb{Z}^{\ell + n}$, there exist coefficients $c^{a,b}_1, \ldots, c^{a,b}_\chi \in \mathbb{C}$ such that 
\begin{equation} \label{eq:expansion}
    \left [\frac{x^b}{f^a} \, \frac{{\rm d}x}{x} \right ] \, = \, c_1^{a,b} \, [\phi_1] + \cdots + c^{a,b}_\chi \, [\phi_\chi] \quad \text{in} \, \, H^n(X,\omega).
\end{equation}
\begin{example}
    For the beta integral, we have seen in Example \ref{eq:3.5} that for $\phi_1 = [{\rm d}x/x]$, 
    \[ c_1^{a,b} \, = \, \frac{(1-s)_{-a}(\nu)_b}{(1+\nu-s)_{b-a}}. \qedhere \]
\end{example}
\begin{corollary} \label{cor:relations}
    Let $X, \omega$ be as above. A regular $n$-form $\phi \in \Omega^n(X)$ is zero in twisted cohomology, i.e., $[\phi] = 0$ in $H^n(X,\omega)$, if and only if 
    \[ \langle \Gamma, \phi \rangle \, = \, \int_\Gamma f^{-s}x^\nu \, \phi \, = \, 0 \quad \text{ for all } \quad  [\Gamma] \in H_n(X,-\omega). \]
    Here it suffices to let $[\Gamma]$ run over a $\mathbb{C}$-basis for $H_n(X,-\omega)$.
\end{corollary}
Corollary \ref{cor:relations} says that relations in cohomology like \eqref{eq:betarelation} are equivalent to relations between Euler integrals which hold \emph{for any twisted cycle}. That is, Equation \eqref{eq:betarelation} implies 
\[ \int_\Gamma \frac{x^{\nu + b} }{(1-x)^{s+a}} \frac{{\rm d}x}{x} \, = \,   \frac{(1-s)_{-a}(\nu)_b}{(1+\nu-s)_{b-a}} \int_\Gamma \frac{x^{\nu } }{(1-x)^{s}}  \frac{{\rm d}x}{x}.   \]
for any $[\Gamma] \in H_1(\mathbb{C} \setminus \{0,1\},-\omega)$. More generally, the expansion \eqref{eq:expansion} in terms of a basis gives 
\[ \int_\Gamma \frac{x^{\nu + b} }{f^{s+a}} \frac{{\rm d}x}{x} \, = \, c_1^{a,b} \int_\Gamma \frac{x^{\nu } }{f^{s}} \, \phi_1 \,  + \,  \cdots \,  + \,  c_\chi^{a,b} \int_\Gamma \frac{x^{\nu } }{f^{s}} \, \phi_\chi, \quad \text{for all } [\Gamma] \in H_n(X,-\omega).   \]
The integrals on the right-hand side are called a set of \emph{master integrals} in physics, see \cite{weinzierl2022feynman}.
\begin{example} \label{ex:gammanonzero}
    In Example \ref{ex:3.10}, we promised to show that $\Gamma$ from \eqref{eq:twisted-cycle} is nonzero in twisted homology. For this, let us fix values of $s, \nu$ such that $0 < \nu < \nu - s + 1$. These are precisely the convergence conditions derived in Example \ref{ex:betaextend}. We will come back to this later. The twisted cycle $\Gamma$ in \eqref{eq:twisted-cycle} depends on $\varepsilon$, but by the Cauchy-Goursat theorem, the value of the integral $\langle \Gamma, \frac{{\rm d}x}{x} \rangle$ is independent of $\varepsilon \in (0,1/2)$.
    According to the three terms in \eqref{eq:twisted-cycle}, we split the integral up into three parts: $\langle \Gamma, \frac{{\rm d}x}{x} \rangle = I_0(\varepsilon) + I_-(\varepsilon) + I_1(\varepsilon)$. The first summand is 
    \[I_0(\varepsilon) \, = \, \frac{1}{e^{2 \pi\sqrt{-1}\nu }-1} \int_{0}^{2\pi} (1-\varepsilon e^{\sqrt{-1}\theta})^{-s} \, \varepsilon^\nu \, e^{\sqrt{-1} \theta \nu} {\rm d} \theta. \]
    Because of the assumption $\nu >0$, we have $\lim_{\varepsilon \rightarrow 0^+} I_0(\varepsilon) = 0$. Analogously, one shows $\lim_{\varepsilon \rightarrow 0^+} I_1(\varepsilon) = 0$. Finally, by \eqref{eq:betafunc}, we have 
    \[ \left \langle \Gamma, \frac{{\rm d}x}{x}  \right \rangle \, = \, \lim_{\varepsilon \rightarrow 0^+} I_-(\varepsilon) \, = \, \lim_{\varepsilon \rightarrow 0^+}  \int_{\varepsilon}^{1-\varepsilon} \frac{x^\nu}{(1-x)^s} \, \frac{{\rm d}x}{x} \, = \, B(\nu,1-s). \]
    Since the result is nonzero, Corollary \ref{cor:relations} implies $[\Gamma] \neq 0$.
\end{example}

Our final goal in this section is to connect the Euler integrals $\langle \Gamma, \phi \rangle$ obtained from the pairing in Theorem \ref{thm:perfectpairing} with the Euler-Mellin integrals from Section \ref{sec:1}. For that, we first need to reinterpret $\langle \Gamma, \phi \rangle$ as a function of $s, \nu$. We write $\omega = \omega(s,\nu)$ to emphasize the dependence on these parameters. On the cohomology side, the natural thing to do is to fix $\phi \in \Omega^n(X)$, and regard it as an element in the varying cohomology vector space $H^n(X,\omega(s,\nu))$. On the homology side, we need to take into account the fact that the line bundle ${\cal L}_{-\omega(s,\nu)}$ depends on $s, \nu$. Let $\Gamma(s,\nu) = \Delta \otimes \tau(s,\nu)$ be the singular $n$-simplex $\Delta$ loaded with $\tau(s,\nu) \in {\cal L}_{-\omega(s,\nu)}(\Delta)$. To see how $\tau$ varies with $s, \nu$, note that it is a $\mathbb{C}$-linear combination of the branches of 
\[ {\rm exp}( -s_1 \log f_1 - \cdots - s_\ell \log f_\ell + \nu_1 \log x_1 + \cdots + \nu_n \log x_n) \]
restricted to $\Delta$. Such a branch is fixed after fixing the branches of the logarithms $\log f_i, \log x_j$, which are independent of $s$ and $\nu$. Our integral is the following function of $s, \nu$:
\begin{equation} \label{eq:intsnusimplex}
(s,\nu) \longmapsto \int_{\Delta \otimes \tau(s,\nu)} f^{-s}x^\nu \, \phi \, = \, \int_\Delta \tau(s,\nu)(x) \, \phi \, = \, \langle \Delta \otimes \tau(s,\nu),\phi \rangle. 
\end{equation}
Taking derivatives of \eqref{eq:intsnusimplex} in $s,\nu$ can be done under the integration sign \cite[Chapter XVII, Theorem 8.2]{lang2013undergraduate}. This implies the following Proposition.
\begin{proposition} \label{prop:loesersabbah}
    The function $(s,\nu) \mapsto \langle \Delta \otimes \tau(s,\nu), \phi \rangle$ from \eqref{eq:intsnusimplex} is holomorphic.
\end{proposition}
It is straightforward to extend this to the case where $\Gamma(s,\nu) = \sum_{i} d_p(s,\nu) \cdot \Delta_p \otimes \tau_p(s,\nu)$, where the coefficients $d_p(s,\nu)$ are meromorphic functions. Similar to Definition \ref{def:integral}, we set
\begin{equation} \label{eq:intsnu} \langle \Gamma(s,\nu), \phi \rangle \, = \, \sum_{p} d_p(s,\nu) \cdot \int_{\Delta_p} \tau_p(s,\nu)(x) \, \phi. 
\end{equation}
By Proposition \ref{prop:loesersabbah}, this is meromorphic in $s, \nu$. Notice that we can also view this as the sum $\sum_p \langle \Delta_p \otimes \tau_p(s,\nu), d_p(s,\nu) \phi \rangle$, allowing meromorphic coefficients in cohomology. This will be useful in Section \ref{sec:difference}. It turns out we have seen \eqref{eq:intsnu} before.
\begin{theorem*} \label{thm:backtosec1}
    Let $f_1, \ldots, f_\ell$ satisfy Assumption \ref{assum:poscoeffs}, and suppose that the Minkowski sum $\Delta(f_1) + \cdots + \Delta(f_\ell)$ has dimension $n$. Let $X$ be as in \eqref{eq:veryaffinevariety}. There exist finite sets of meromorphic functions $d_p(s, \nu)$, singular $n$-chains $\Delta_p$ on $X$ and sections $\tau_p(s,\nu) \in {\cal L}_{-\omega(s,\nu)}(\Delta_p)$ with the following property. For $\Gamma(s,\nu) = \sum_{p} d_p(s,\nu) \cdot \Delta_p \otimes \tau_p(s,\nu)$, the function $(s, \nu) \mapsto \langle \Gamma(s,\nu), \frac{{\rm d} x}{x} \rangle$ form \eqref{eq:intsnu} is the meromorphic continuation \eqref{eq:meromcont} from Theorem \ref{thm:meromorphiccontinuation}.
\end{theorem*}
This statement is labeled Theorem* \ref{thm:backtosec1} (with an asterisk) because, to the best of our knowledge, there exists no rigorous proof in the literature. Yet, it is widely accepted and used. Our sketch of proof below makes it more than a conjecture. Providing full details is among the proposed problems in Section \ref{sec:5}. This requires tools beyond our present scope. 
\begin{proof}[Sketch of proof of Theorem* \ref{thm:backtosec1}]
    Consider the algebraic moment map in \eqref{eq:momentmap} with $s_i = 1$. By Lemma \ref{lem:moment_map}, $\mu: \mathbb{R}^n_+ \rightarrow {\rm int}(P)$ is a diffeomorphism, with $P = \sum_{i=1}^\ell \Delta(f_i)$. This gives 
    \[ \int_{\mathbb{R}^n_+} f(y)^{-s} y^\nu \, \frac{{\rm d}y}{y} \, = \, \int_P f(\mu^{-1}(x))^{-s} \mu^{-1}(x)^\nu J_{\mu^{-1}} \frac{{\rm d} x}{\prod_j \mu^{-1}(x)_j},\]
    where $J_{\mu^{-1}}$ is the Jacobian determinant of $\mu^{-1}$, and the last denominator is the product of the entries of $\mu^{-1}$. Like in Example \ref{ex:gammanonzero}, for values of $s, \nu$ where the integral on the right converges, we replace $P$ by a twisted cycle $\Gamma(P)$, called the \emph{regularization} of $P$, such that
    \[ \int_{\mathbb{R}^n_+} f(y)^{-s} y^\nu \, \frac{{\rm d}y}{y}  \, = \, \int_{\Gamma(P)} f(\mu^{-1}(x))^{-s} \mu^{-1}(x)^\nu J_{\mu^{-1}} \frac{{\rm d} x}{\prod_j \mu^{-1}(x)_j} \, = \, \left \langle \Gamma(P), J_{\mu^{-1}} \frac{{\rm d} x}{\prod_j \mu^{-1}(x)_j} \right \rangle. \]
    The integral on the right side is the pairing of a holomorphic $n$-form with the twisted cycle $\Gamma(P)$. Here we need to use the analytic version of twisted (co)homology. The manifold is $\tilde{X} = \{x \in  \mathbb{C}^n \, : \, f_1(\mu^{-1}(x)) \cdots f_\ell(\mu^{-1}(x)) \mu^{-1}(x)_1 \cdots \mu^{-1}(x)_n \neq 0 \}$, and the twist is $\tilde{\omega} = {\rm dlog}(f(\mu^{-1}(x))^{-s} \mu^{-1}(x)^\nu)$. The construction of the regularization $\Gamma(P)$ is that in \cite[Sections 3.2.4 and 3.2.5]{aomoto2011theory}. For this, when $P$ is not smooth, one needs to replace the moment map $\mu$ by that of a different toric variety, obtained by blowing up the toric variety of $P$ in its singular locus. The cycle $\Gamma(s,\nu)$ in the Theorem* is the pullback $\mu^*(\Gamma(P))$ of $\Gamma(P)$ under $\mu$.
\end{proof}

Theorem* \ref{thm:backtosec1} replaces \emph{integrating over $\mathbb{R}^n_+$} by \emph{integrating over $\Gamma(s,\nu)$}. The twisted cycle $\Gamma(s,\nu)$ is called the \emph{regularization} of $\mathbb{R}^n_+$ \cite[Sections 3.2.4 and 3.2.5]{aomoto2011theory}. Here is an example.

\begin{example}
    We have seen two integral formulas for $B(\nu,1-s)$ in Example \ref{ex:betaextend}:
    \begin{equation} \label{eq:2betaints}
    \int_{0}^1 \frac{x^\nu}{(1-x)^s} \, \frac{{\rm d}x}{x} \, = \, \int_{0}^1 \frac{x^\nu}{(1-x)^{s-1}} \, \frac{{\rm d}x}{x(1-x)} \, = \,  \int_{\mathbb{R}_+} \frac{y^\nu}{(1+y)^{\tilde{s}}} \, \frac{{\rm d}y}{y}.
    \end{equation}
    The coordinate transformation is $\mu : \mathbb{R}_+ \rightarrow (0,1)$, with $x=\mu(y) = y(1+y)^{-1}$. This is the moment map from Lemma \ref{lem:moment_map} up to scaling by $\tilde{s}$. Its complexification $\mu_{\mathbb{C}}$ is an isomorphism 
    \[ X_y \, = \,  \mathbb{C} \setminus \{0,-1\} \, \overset{\mu_{\mathbb{C}}}{\longrightarrow} \, \mathbb{C} \setminus \{0,1\} \, = \, X_x.\] 
    Let $\omega_x = {\rm dlog}((1-x)^{1-s}x^\nu) \in \Omega^1(X)$ be the regular one-form corresponding to the middle integral of \eqref{eq:2betaints}. We define $\omega_y$ to be the pullback of $\omega_x$ along $\mu_{\mathbb{C}}$. That is, explicitly,
    \[ \omega_y = \mu^*(\omega_x) \, = \,  {\rm dlog}((1-\mu_\mathbb{C}(y))^{1-s}\mu_{\mathbb{C}}(y)^\nu) \,  = \,  {\rm dlog}((1+y)^{-\tilde{s}}y^\nu) .\] 
    The cocycle ${\rm d}x/(x(1-x))$ pulls back to ${\rm d}y/y$ under $\mu$.
    A twisted chain
    $\Gamma = \Delta \otimes \tau \in C_1(X_x, - \omega_x)$ is naturally pulled back to a cycle $\mu^*(\Gamma) \in C_1(X_y,-\omega_y)$ via \[\mu^*(\Gamma) \, = \,  \mu^{-1}(\Delta) \otimes (\tau \circ \mu).\]
    As usual, this definition for simplices is extended linearly to $C_1(X_x,-\omega_x)$. With this notation in place, it is easy to check that for any $\Gamma \in C_1(X_x,-\omega_x)$, we have
    \[ \left \langle \Gamma, \frac{{\rm d}x}{x(1-x)} \right \rangle \, = \, \int_{\Gamma} \frac{x^\nu}{(1-x)^s} \, \frac{{\rm d}x}{x} \, = \, \int_{\mu^*(\Gamma)} \frac{y^\nu}{(1+y)^{\tilde{s}}} \, \frac{{\rm d}y}{y} \, = \, \left \langle \mu^*(\Gamma), \frac{{\rm d} y}{y} \right \rangle. \]
    By Example \ref{ex:gammanonzero}, if we pick $\Gamma$ as in \eqref{eq:twisted-cycle} (with $s$ replaced by $s-1$), the left integral is a meromorphic function in $(s,\nu)$ which agrees with $B(\nu,1-s)$ if $0 < \nu < \nu -s +1$. Hence, the regularization of $\mathbb{R}_+$ is $\mu^*(\Gamma)$. It depends on $\tilde{s}, \nu$ as explained above.
\end{example}

\section{Differential and difference equations} \label{sec:4}

It is common practice to describe a function $F$ via the differential or difference equations it satisfies. More precisely, one attempts to find \emph{differential operators} $P$ such that $P \bullet F = 0$, or \emph{difference/shift operators} $S$ such that $S \bullet F = 0$. Here $P \bullet F$ reads as \emph{$P$ applied to $F$}, and similarly for $S \bullet F$. In this section, $F$ is an Euler integral \eqref{eq:integralintro}, seen as a pairing between a twisted $n$-cycle $[\Gamma]$ and the twisted $n$-cocycle $[{\rm d} x/x]$, see Section \ref{sec:3}. Such an integral satisfies difference equations when seen as a meromorphic function of $s$ and $\nu$, as in \eqref{eq:intsnu}.

\begin{example}\label{ex:4.1}
    We have seen in Example \ref{ex:gammanonzero} that the beta function $B(\nu, 1-s)$ is given by 
    \begin{equation}
        {\cal I}(s,\nu) \, = \, B(\nu,1-s) \, = \, \int_\Gamma \frac{x^\nu}{(1-x)^s} \, \frac{{\rm d}x}{x}, 
    \end{equation}
    where $\Gamma$ is the twisted $1$-cycle from \eqref{eq:twisted-cycle}. This agrees with the integral over $(0,1)$ in \eqref{eq:betafunc} when ${\rm Re}(s)>0$ and $0 < {\rm Re}(\nu) < {\rm Re}(\nu) -{\rm Re}(s) + 1$. The \emph{shift operator in $s$}, denoted $\sigma_s$, acts by
    \[ \sigma_s \bullet {\cal I}(s,\nu) \, = \, {\cal I}(s + 1, \nu).\]
    Similarly, the action of $\sigma_\nu$ is $\sigma_\nu \bullet {\cal I}(s,\nu) = {\cal I}(s,\nu+1)$. We claim that the shift operators
    \begin{equation} \label{eq:shiftbeta}
    S_1 \, = \, 1-\sigma_s\, (1-\sigma_\nu) \quad \text{and} \quad S_2 \, = \, \nu + s \, \sigma_\nu \sigma_s 
    \end{equation}
    annihilate ${\cal I}$, i.e., $S_1 \bullet {\cal I} = 0$ and $S_2 \bullet {\cal I} = 0$. Here coefficients that are rational functions in $s, \nu$ simply act by multiplication. The identity $S_1 \bullet {\cal I} = 0$ is easy to verify: 
    \[ \sigma_s \, (1-\sigma_\nu) \bullet {\cal I } \, = \, \sigma_s \bullet \left ( \int_\Gamma \frac{x^\nu}{(1-x)^s} \, \frac{{\rm d}x}{x} - \int_\Gamma \frac{x\cdot x^\nu}{(1-x)^s} \, \frac{{\rm d}x}{x} \right) \, = \, \int_\Gamma \frac{(1-x)\cdot x^\nu}{(1-x)^{s+1}} \, \frac{{\rm d}x}{x}.\]
    The rightmost integral equals ${\cal I} = 1 \bullet {\cal I}$. To see that $S_2 \bullet {\cal I} = 0$, we apply Lemma \ref{lem:modeverything}:
    \[ S_2 \bullet {\cal I} \, = \, \int_\Gamma \frac{x^\nu}{(1-x)^s} \, \left ( \frac{\nu}{x} + \frac{s}{1-x} \right) {\rm d}x \, = \, \int_\Gamma \frac{x^\nu}{(1-x)^s} \, \nabla_\omega(1) \, = \, 0.\]
    Here $\nabla_\omega$ is as in \eqref{eq:nablabeta}. We reiterate that, in order to view ${\cal I}$ as a meromorphic function of $s,\nu$, it is important to keep in mind that $\omega = \omega(s,\nu)$ varies. Hence, so does the local system ${\cal L}_{-\omega} = {\cal L}_{-\omega(s,\nu)} $, and the twisted cycle $\Gamma = \Gamma(s,\nu)$. This was explained in Section \ref{sec:3.3}.
\end{example}

The goal of Section \ref{sec:difference} is to derive operators like \eqref{eq:shiftbeta} for general Euler integrals. These operators appeared in \cite[Section 3.1]{arkani2021stringy}, \cite[Section 3]{Agostini:2022cgv} and \cite{matsubara2023twisted}. Shift operators for Feynman integrals were studied in \cite{bitoun2019feynman}. Section \ref{sec:differential} discusses differential operators. For that, we must view our Euler integrals as functions of a new set of parameters: the coefficients of $f_i$. 

\begin{example}\label{ex:4.2}
    Fix two generic complex numbers $s, \nu \in \mathbb{C}$. We modify the beta integral \eqref{eq:betafunc} by introducing complex valued parameters $z_1, z_2$ for the coefficients of the denominator $f$: 
    \[ {\cal I}(z_1, z_2) \, = \, \int_\Gamma \frac{x^\nu}{(z_1 + z_2 \, x)^s} \, \frac{{\rm d}x}{x}.\]
    The dependence on $z = (z_1,z_2)$ is subtle. For instance, the very affine variety $X$ from \eqref{eq:veryaffinevariety} depends on $z$ and, necessarily, so does $\Gamma$. In this example, one can think about ${\cal I}(z_1,z_2)$ as a function on a small neighborhood of $(z_1,z_2) = (1,-1)$, which corresponds to our original beta integral. In that neighborhood, one can modify $\Gamma$ by keeping the $1$-simplices in Example \ref{ex:chains} fixed, and varying the sections in Example \ref{ex:3.2} with $z_1,z_2$. For instance, the initial condition $\tau_{0,0}(\varepsilon; z_1, z_2) = \zeta(z_1,z_2)$ from \eqref{eq:IVP} is given by $\zeta(z_1,z_2)={\rm exp}(-s \log (z_1+z_2 \, \varepsilon) + \nu \log \varepsilon)$, and for the first $\log$ we use the analytic continuation of the positive branch near $1-\varepsilon$. 

    The differential operator $\partial_{z_i}$ acts by partial derivation in $z_i$, for $i = 1, 2$, and rational functions in $z$ act by multiplication. Here are two annihilating operators for ${\cal I}(z_1,z_2)$:
    \begin{equation} \label{eq:diffbeta}
    P_1 \, = \, z_1 \, \partial_{z_1} + z_2 \, \partial_{z_2} + s \quad \text{and} \quad P_2 \, = \, z_2 \, \partial_{z_2} + \nu.
    \end{equation}
    The derivatives can be taken under the integration sign. We verify $P_1 \bullet {\cal I} = 0$: 
    \[ (z_1 \, \partial_{z_1} + z_2 \, \partial_{z_2}) \bullet {\cal I }(z) \, = \, \int_\Gamma \frac{-s z_1 \cdot x^\nu}{(z_1+z_2 \, x)^{s+1}} \, \frac{{\rm d}x}{x} + \int_\Gamma \frac{-s z_2x \cdot x^\nu}{(z_1+z_2 \, x)^{s+1}} \, \frac{{\rm d}x}{x} \, = \, -s \cdot {\cal I}(z).\]
    To see that $P_2 \bullet {\cal I}(z) = 0$, we again need Lemma \ref{lem:modeverything}. We compute 
    \[ 0 \, = \, P_2 \bullet {\cal I}(z) \, = \, \int_\Gamma \frac{x^\nu}{(z_1+z_2 \, x)^{s}} \, \nabla_{\omega(z)}(1) \quad \text{with} \quad \omega(z) \, = \, \left( \frac{\nu}{x} - \frac{sz_2}{z_1+z_2 \, x} \right) {\rm d}x. \qedhere\]
\end{example}
The differential operators \eqref{eq:diffbeta} form an \emph{$A$-hypergeometric system} or \emph{GKZ system} of linear partial differential equations. Such systems were introduced by Gelfand, Kapranov and Zelevinsky to study $A$-hypergeometric functions \cite{gelfand1986general,gelfand1990generalized}. We will introduce these systems and recall their relation to Euler integrals in Section \ref{sec:differential}. For a recent survey, see \cite{reichelt2021algebraic}.

\subsection{Difference equations} \label{sec:difference}

We start with difference/shift operators in $(s,\nu)$. In analogy with Example \ref{ex:4.1}, we call these operators $\sigma_{s_i}$ for $i=1,\ldots, \ell$ and $\sigma_{\nu_j}$ for $j=1,\ldots,n$. They act on the integral ${\cal I}= \langle \Gamma, \phi \rangle$ from \eqref{eq:intsnu} as follows. We view \eqref{eq:intsnu} as the sum of pairings
\[ {\cal I}(s,\nu) = \, \sum_{p} d_p(s,\nu) \cdot \int_{\Delta_p} \tau_p(s,\nu)(x) \, \phi \, = \, \sum_p \, \langle\, \Delta_p \otimes \tau_p(s,\nu), \,d_p(s,\nu) \, \phi \, \rangle. \]
That is, the cocycle now depends meromorphically on $s, \nu$. We set
\begin{align}
\sigma_{s_i} \bullet {\cal I}(s,\nu)   \, = \,  {\cal I}(s + e_i, \nu) &\, = \,  \sum_p \, \langle\, \Delta_p \otimes \tau_p(s+e_i,\nu), \,d_p(s+e_i,\nu) \, \phi \, \rangle \label{eq:s-shift1}\\ 
&\, = \, \sum_p \, \langle\, \Delta_p \otimes \tau_p(s,\nu), \, \, d_p(s+e_i,\nu) \,  f_i^{-1} \phi \, \rangle, \label{eq:s-shift2} \\
\sigma_{\nu_j} \bullet {\cal I}(s,\nu) \, = \, {\cal I}(s, \nu + e_j) &\, = \,  \sum_p \, \langle\, \Delta_p \otimes \tau_p(s,\nu+e_j), \,d_p(s,\nu+e_j) \, \phi \, \rangle \label{eq:nu-shift1}  \\
&\, = \, \sum_p \, \langle\, \Delta_p \otimes \tau_p(s,\nu), \, \, d_p(s+e_i,\nu) \,  x_j \, \phi \, \rangle. \label{eq:nu-shift2} 
\end{align}
Here $e_j$ is the $j$-th standard basis vector. The reader should check the passages from \eqref{eq:s-shift1} to \eqref{eq:s-shift2} and \eqref{eq:nu-shift1} to \eqref{eq:nu-shift2} carefully. The expressions \eqref{eq:s-shift2} and \eqref{eq:nu-shift2} show that the action of the shift operators can be viewed as an action on cohomology $H^n(X,\omega(s,\nu))$. For instance, 
\begin{equation} \label{eq:shiftcohom}
    \sigma_{s_i} \bullet [ \phi(s,\nu) ] \, = \,  [f_i^{-1} \phi(s+e_i, \nu) ], \quad \sigma_{\nu_j} \bullet [ \phi(s,\nu) ] \, = \, [x_j \phi(s,\nu+e_j)].
\end{equation}
In \cite{matsubara2023twisted}, this action is defined on a cohomology vector space with coefficients in $\mathbb{C}(s,\nu)$.
We will also use the inverses $\sigma_{s_i}^{-1}, \sigma_{\nu_j}^{-1}$ of these shift operators. Their action is straightforward to define. The variables $s_i$ and $\nu_j$ act on ${\cal I}(s,\nu)$ by multiplication: $s_i \bullet {\cal I}(s,\nu) = s_i{\cal I}(s,\nu)$, and $\nu_j \bullet {\cal I}(s,\nu) = \nu_j {\cal I}(s,\nu)$. Notice that the operators $s_i$ and $\sigma_{s_i}$ do not commute: 
\[ \sigma_{s_i} s_i \bullet {\cal I}(s,\nu) \, = \, (s_i+1){\cal I}(s+e_i,\nu) \, \neq \, s_i {\cal I}(s+e_i,\nu) \, = \, s_i \sigma_{s_i} \bullet {\cal I}(s,\nu). \]
Such commutator relations naturally lead to the following definition. 
\begin{definition}
    The \emph{ring of difference operators} $R = \mathbb{C}(s,\nu) \langle \sigma_{s_1}^{\pm 1}, \ldots, \sigma_{s_\ell}^{\pm 1}, \sigma_{\nu_1}^{\pm 1}, \ldots, \sigma_{\nu_n}^{\pm 1} \rangle$ is the $\mathbb{C}(s,\nu)$-vector space with basis $\sigma_{s}^a\sigma_\nu^b$, where $(a,b) \in \mathbb{Z}^{\ell + n}$. I.e., it consists of finite sums
    \[ \sum_{(a,b) \in \mathbb{Z}^{\ell+n}} g_{a,b}(s,\nu) \, \sigma_{s_1}^{a_1} \cdots \sigma_{s_\ell}^{a_\ell} \, \sigma_{\nu_1}^{b_1} \cdots \sigma_{\nu_n}^{b_n} \, = \, \sum_{(a,b) \in \mathbb{Z}^{\ell+n}} g_{a,b}(s,\nu) \, \sigma_s^a\, \sigma_\nu^b. \]
    The product is subject to the following relations. For any rational function $g(s,\nu) \in \mathbb{C}(s,\nu)$,
    \[ [\sigma_{s_i}^{\pm 1}, g(s,\nu)]  \, = \, (g(s\pm e_i,\nu) - g(s,\nu)) \sigma_{s_i}  \quad \text{and} \quad  [\sigma_{\nu_j}^{\pm 1}, g(s,\nu)] = (g(s,\nu \pm e_j) - g(s,\nu)) \sigma_{\nu_j}.\]
    Here $[A,B] = AB - BA$ denotes the commutator. 
\end{definition}

The ring $R$ acts on meromorphic functions in $s,\nu$ as explained above. The \emph{annihilator} ${\rm Ann}_R({\cal I})$ of a meromorphic function ${\cal I}(s,\nu)$ consists of all shift operators annihilating ${\cal I}$:
\begin{equation} \label{eq:annihilator}
{\rm Ann}_R({\cal I}(s,\nu))  \, = \, \{ S \in R \, : \, S \bullet {\cal I}(s,\nu) = 0 \}. 
\end{equation}
Clearly, if $S_1, S_2 \in {\rm Ann}_R({\cal I})$, then $S_1 + S_2 \in {\rm Ann}_R({\cal I})$ as well. Also, if $S_1 \in {\rm Ann}_R({\cal I})$, then $S_2S_1 \in {\rm Ann}_R({\cal I})$ for any $S_2 \in R$. In other words, ${\rm Ann}_R({\cal I})$ is a \emph{left ideal} of $R$. 

To describe the annihilator of our integral, we introduce the notation $f_i(\sigma_\nu) \in R$ for the difference operator obtained by replacing $x_j \to \sigma_{\nu_j}$ in $f_i(x)$. This is well-defined, since all $\sigma_{\nu_j}$ commute. For instance, for $f(x) = 1 - x$ from Example~\ref{ex:4.1}, we write $f(\sigma_\nu) = 1 - \sigma_\nu$. 

\begin{proposition} \label{prop:annR}
Let $J \subset R$ be the left ideal generated by the following $\ell + n$ operators:
\begin{align}\label{eq:J1}
1 - \sigma_{s_i} f_i(\sigma_\nu) \qquad &\text{for}\; i=1,\ldots,\ell,\\
\label{eq:J2} \sigma_{\nu_j}^{-1} \nu_j - \sum_{i=1}^{\ell} s_i \cdot \sigma_{s_i}  \frac{\partial f_i}{\partial x_j} (\sigma_\nu) \qquad&\text{for}\; j=1,\ldots,n.
\end{align} 
For any cycle $\Gamma$, the annihilator of the Euler integral ${\cal I}_\Gamma(s,\nu) = \langle \Gamma, \frac{{\rm d}x}{x} \rangle$ contains $J$.
\end{proposition}
\begin{proof}
We need to show that the operators \eqref{eq:J1} and \eqref{eq:J2} annihilate ${\cal I}_\Gamma(s,\nu)$. Let $f_i(x) = \sum_{\alpha} c_{i,\alpha} \cdot x^\alpha$, where $c_{i,\alpha} \in \mathbb{C}$ for $i=1,2,\ldots,\ell$. Hence, using \eqref{eq:nu-shift2}, we find
\begin{equation}
f_i(\sigma_\nu) \bullet {\cal I}_\Gamma(s,\nu) =  \left  \langle \Gamma, \,   \sum_{\alpha} c_{i,\alpha} \cdot x^{\alpha} \frac{{\rm d}x}{x} \right \rangle = \left \langle \Gamma, f_i(x) \frac{{\rm d} x}{x} \right \rangle.
\end{equation}
Together with \eqref{eq:s-shift2} we find
$\sigma_{s_i} f_i(\sigma_\nu) \bullet {\cal I}_\Gamma(s,\nu) = \langle \Gamma, \tfrac{{\rm d}x}{x}\rangle = {\cal I}_\Gamma(s,\nu)$, which shows that \eqref{eq:J1} annihilates ${\cal I}_\Gamma(s,\nu)$. Notice that, for this, we do not use the fact that $\Gamma$ is a cycle. To show that the operators in \eqref{eq:J2} annihilate ${\cal I}_\Gamma(s,\nu)$ as well, we compute
\[
\frac{\partial f_i}{\partial x_j}(\sigma_\nu) \bullet {\cal I}_\Gamma(s,\nu) \,= \, \left \langle \Gamma, \frac{\partial f_i}{\partial x_j} \frac{{\rm d}x}{x} \right \rangle \quad \text{and} \quad 
\sigma^{-1}_{\nu_j}  \nu_j \bullet {\cal I}_\Gamma(s,\nu) \, = \, \left \langle \Gamma, \frac{\nu_j - 1}{x_j} \frac{{\rm d}x}{x} \right \rangle .
\]
We combine these identities to get
\[
\left(\sigma_{\nu_j}^{-1}  \nu_j - \sum_{i=1}^{\ell} s_i \cdot \sigma_{s_i}   \frac{\partial f_i}{\partial x_j}(\sigma_\nu)  \right) \bullet {\cal I}_\Gamma(s,\nu) \, = \,  \left \langle \Gamma, \left( \frac{\nu_j-1}{x_j} - \sum_{i=1}^{\ell} s_i \frac{\frac{\partial f_i}{\partial x_j}(x) }{f_i(x)}\right) \frac{{\rm d}x}{x} \right \rangle.
\]
To show that this is zero for any cycle $\Gamma$, we need to show that the $n$-form on the right is exact (Lemma \ref{lem:modeverything}). That is, we need to find $\psi$ such that it equals $\nabla_\omega(\psi)$. The solution is
\[ \left( \frac{\nu_j-1}{x_j} - \sum_{i=1}^{\ell} s_i \frac{\frac{\partial f_i}{\partial x_j}(x) }{f_i(x)}\right) \frac{{\rm d}x}{x} \, = \, \nabla_{\omega} \left((-1)^{j-1} \frac{{\rm d}x_{\hat{j}}}{x_1 \cdots x_n} \right) ,\]
where ${\rm d}x_{\hat{j}}$ is the $(n-1)$-form ${\rm d} x_1 \wedge \cdots \wedge {\rm d}x_{j-1} \wedge {\rm d}x_{j+1} \wedge \cdots \wedge {\rm d}x_n$.
\end{proof}

\begin{example} \label{ex:contigbeta}
The operators $S_1$ and $\sigma_{\nu}^{-1} S_2$ from \eqref{eq:shiftbeta} are \eqref{eq:J1} and \eqref{eq:J2} for $n = \ell = 1$ and $f = 1-x$. We can use these operators to obtain the relation \eqref{eq:betarelation}. Observe that 
\[ s \cdot (1-\sigma_s + \sigma_s \sigma_\nu) - (\nu + s \sigma_\nu \sigma_s) \, = \, s - s \sigma_s - \nu \, \in \, J. \]
Since $J$ annihilates ${\cal I}_\Gamma = \langle \Gamma, {\rm d}x/x \rangle$  for all $\Gamma$ (Proposition \ref{prop:annR}), this means $\sigma_s \bullet [{\rm d}x/x] = (s-\nu)s^{-1} \bullet [{\rm d}x/x]$. Applying $\sigma_s^{-1}$ from the left (necessarily, because $J$ is a \emph{left} ideal), we get 
\[ \sigma_s^{-1} \left( \frac{s-\nu}{s}\right) \, = \, \left (\frac{s-1-\nu}{s-1} \right)\sigma_s^{-1} \, = \, 1 \, \, {\rm mod} \, \, J, \quad \text{i.e.,} \quad \sigma_s^{-1} \bullet \left [\frac{{\rm d}x}{x} \right] \, = \, \left ( \frac{1-s}{1+\nu-s} \right) \bullet \left [\frac{{\rm d}x}{x} \right] .\]
Similarly since $\sigma_s^{-1}(1-\sigma_s + \sigma_s\sigma_\nu) \in J$, we derive that 
\[ \sigma_\nu \bullet \left [\frac{{\rm d}x}{x} \right] \, = \, (1-\sigma_s^{-1}) \bullet \left [\frac{{\rm d}x}{x} \right]  \, = \, \left(\frac{\nu}{1 + \nu - s} \right) \bullet \left [\frac{{\rm d}x}{x} \right] .\]
Now, we use these identities to write $\sigma_s^a\sigma_\nu^b \bullet [{\rm d} x/x] = [x^b/(1-x)^a \cdot {\rm d} x/x]$ in terms of $[{\rm d} x/x]$:
\[\sigma_s^a \sigma_\nu^b  \bullet \left [\frac{{\rm d}x}{x} \right] \, = \, \sigma_s^a \sigma_\nu^{b-1} \left(\tfrac{\nu}{1 + \nu - s} \right ) \bullet \left [\frac{{\rm d}x}{x} \right] \, = \, \sigma_s^a \sigma_\nu^{b-2} \left(\tfrac{\nu+1}{2 + \nu - s} \right ) \left( \tfrac{\nu}{1 + \nu - s} \right) \bullet \left [\frac{{\rm d}x}{x} \right] \, = \, \ldots. \]
Here the second equality uses the commutation rules. Suppose $b, a\in \mathbb{Z}_{>0}$ are positive. After repeating this step $b$ times, we begin expanding $\sigma_s^a$:
\begin{align*}
    \sigma_s^a \sigma_\nu^b  \bullet \left [\frac{{\rm d}x}{x} \right]\, &= \, \sigma_s^a \frac{(\nu)_b}{(1+\nu-s)_b} \bullet \left [\frac{{\rm d}x}{x} \right] \, = \, \sigma_s^{a-1} \frac{(\nu)_b}{(\nu-s)(1+\nu-s)_{b-1}} \left( \tfrac{s-\nu}{s} \right ) \bullet \left [\frac{{\rm d}x}{x} \right] \\
    &= \, \sigma_s^{a-2} \frac{(\nu)_b}{(\nu - s - 1)(\nu-s)(1+\nu-s)_{b-2}} \left( \tfrac{s-\nu+1}{s+1} \right ) \left( \tfrac{s-\nu}{s} \right ) \bullet \left [\frac{{\rm d}x}{x} \right] \\
    &= \, \cdots \, = \, \frac{(\nu)_b}{(1+\nu-s)_{b-a}} \left (\frac{1}{(-s)(-s-1) \cdots (-s-a+1)} \right ) \bullet \left [\frac{{\rm d}x}{x} \right].
\end{align*}
The rational function we obtain is precisely that of \eqref{eq:betarelation}. Notice that the numerator factors $(s-\nu)$, $(s-\nu+1)$, \ldots cancel with the denominators $(\nu-s)$, $(\nu-s-1)$, \ldots, and the minus signs are absorbed in the denominator factors $s, s+1, \ldots$. The reader should check that if $a,b$ satisfy different sign conditions, we arrive at the same formula. 
\end{example}

Example \ref{ex:contigbeta} illustrates the concept of \emph{contiguity relations} for the Euler beta integral. In general, the action of shift operators can be captured by \emph{contiguity matrices}. Let $[\phi_1], \ldots, [\phi_\chi]$ be a basis for the twisted cohomology $H^n(X,\omega(s,\nu))$ (for generic $s,\nu$). There are $\chi = (-1)^n \cdot \chi(X)$ such basis elements by Corollary \ref{cor:dimchi}. We assume the $[\phi_k]$ have coefficients that are rational functions in $s,\nu$. There are $\chi \times \chi$-matrices $C_{s_i}$, $C_{\nu_j}$ such that 
\[ \sigma_{s_i} \bullet \begin{pmatrix}
    \langle \Gamma, \phi_1 \rangle \\ \vdots \\ \langle \Gamma, \phi_\chi \rangle
\end{pmatrix} \, = \, C_{s_i} \cdot \begin{pmatrix}
    \langle \Gamma, \phi_1 \rangle \\ \vdots \\ \langle \Gamma, \phi_\chi \rangle
\end{pmatrix}, \quad \quad \sigma_{\nu_j} \bullet \begin{pmatrix}
    \langle \Gamma, \phi_1 \rangle \\ \vdots \\ \langle \Gamma, \phi_\chi \rangle
\end{pmatrix} \, = \, C_{\nu_j} \cdot \begin{pmatrix}
    \langle \Gamma, \phi_1 \rangle \\ \vdots \\ \langle \Gamma, \phi_\chi \rangle
\end{pmatrix} . \]
Here $\sigma_{s_i}$ and $\sigma_{\nu_j}$ act entry-wise on vectors, and the contiguity matrices $C_{s_i}$, $C_{\nu_j}$ have entries in $\mathbb{C}(s,\nu)$. For more on these matrices and how to compute them, see \cite[Section 5]{matsubara2023twisted}.

\begin{remark}
In the context of Feynman integrals, special choices of contiguity relations are known as \emph{dimension-shift identities}, because they relate Feynman integrals evaluated in different space-time dimensions. We refer to \cite[Section 6.2]{weinzierl2022feynman} for more details.
\end{remark}

\begin{remark}
Shift relations \eqref{eq:shiftcohom} give a practical way of enlarging the Nilsson-Passare domain of convergence of Euler integrals in $(s,\nu)$, as explained in Section~\ref{sec:convergence}. For instance, the shift operator $\sigma_\nu^{-1}\nu - \tilde{s} \sigma_{\tilde{s}}$ annihilates the beta integral $\int_0^\infty y^\nu/(1+y)^{\tilde{s}} {\rm d}y/y$. Hence 
\begin{equation} \label{eq:extendNP} (\nu-1) \cdot \int_{\mathbb{R}_+} \frac{y^{\nu-1}}{(1+y)^{\tilde{s}}} \, \frac{{\rm d}y}{y} \, = \, \tilde{s} \cdot \int_{\mathbb{R}_+} \frac{y^\nu}{(1+y)^{\tilde{s}+1}} \, \frac{{\rm d}y}{y}.
\end{equation}
This is an equality of meromorphic functions after replacing $\mathbb{R}_+$ with its regularizer from Theorem* \ref{thm:backtosec1}. Suppose $\tilde{s} \in \mathbb{R}_+$. The convergence region from Theorem \ref{thm:convergence} for the left hand side is ${\rm Re}(\nu) \in (1,\tilde{s}+1)$, while that of the right hand side is ${\rm Re}(\nu) \in (0,\tilde{s}+1)$. Hence, \eqref{eq:extendNP} is used to evaluate the meromorphic continuation from Theorem \ref{thm:meromorphiccontinuation} for ${\rm Re}(\nu) \in (0,1)$.
\end{remark}

\subsection{Differential equations} \label{sec:differential}

As illustrated in Example \ref{ex:4.2}, Euler integrals are annihilated by differential operators in the coefficients of $f_i$. These are new parameters denoted by $z_{i,\alpha}$, i.e.,
\begin{equation} \label{eq:fiz} 
f_i \, = \, \sum_{\alpha \,  \in \, A_i} z_{i,\alpha} \cdot x^\alpha, \quad i = 1, \ldots, \ell 
\end{equation}
Here $A_i = {\rm supp}(f_i) \in \mathbb{Z}^n$ is the support of $f_i$, in the sense of Definition \ref{def:newton}. We fix complex parameters 
$s = (s_1, \ldots, s_\ell) \in \mathbb{C}^\ell$ and $\nu = (\nu_1, \ldots, \nu_n) \in \mathbb{C}^n$.
The variety $X$ defined in \eqref{eq:veryaffinevariety} is also dependent on $z$. 
For this reason, we write $X_z$ instead of $X$.
The Euler integral \eqref{eq:integralintro}, seen as a function of the coefficients $z=(z_{i,\alpha})_{i,\alpha}$, defines a holomorphic function on an open subset $U$ of coefficient space $\mathbb{C}^{A} = \mathbb{C}^{A_1} \times \cdots \times \mathbb{C}^{A_\ell}$. To define this function, we need to specify how the twisted integration cycle $\Gamma$ varies with $z$. For a fixed set of coefficients $z^*=(z^*_{i,\alpha})_{i,\alpha} \in \mathbb{C}^A$,
let $[\Gamma(z^*)]\in H_n(X_{z^*};-\omega(z^*))$ be a twisted cycle.
The open set $U$ is a \emph{sufficiently small} neighborhood $U$ of $z^*$. The choice $[\Gamma(z^*)]$ gives rise to a family of cycles $[\Gamma(z)]\in H_n(X_z;-\omega(z))$ defined for $z\in U$. For this, we fix the singular $n$-simplices, and vary the sections of ${\cal L}_{-\omega(z)}$ in the only sensible way.
That is, $\Gamma(z^*) = \sum_{p} d_p \cdot \Delta_p \otimes \tau_p(z^*)$, where $\tau_p(z^*) \in {\cal L}_{-\omega(z^*)}$ is a branch of $f(x;z^*)^{-s}x^\nu$ which depends holomorphically on $z$, and
\begin{equation} \label{eq:gammaofz}
\Gamma(z) \, = \, \sum_{p} d_p \cdot \Delta_p \otimes \tau_p(z) \, \in \, H^n(X_z,-\omega(z)). 
\end{equation}
Here $d_p \in \mathbb{C}$ are constants. This was illustrated for the beta integral in Example \ref{ex:4.2}. It is crucial that, with this construction, the twisted boundary $\partial_{\omega(z)}(\Gamma(z))$ is zero for all $z \in U$.

\begin{proposition}
    Let $U \ni z^*$ and $[\Gamma(z)]\in H_n(X_z;-\omega)$ be as above. The function 
    \begin{equation}\label{eq:I_Gamma}
    \mathcal{I}_\Gamma: U \longrightarrow \mathbb{C}, \quad  z\longmapsto \int_{\Gamma(z)}f(x;z)^{-s}x^\nu\frac{{\rm d}x}{x}\in\mathbb{C} 
    \end{equation}
    is holomorphic on $U$. 
\end{proposition}

\begin{proof}
    Since $\tau(z)$ is holomorphic in $z$, it suffices to observe that for $z \in U$,
    \[ {\cal I}_{\Gamma}(z) \, =\, \sum_p d_p \cdot \int_{\Delta_p} \tau(z)(x) \, \frac{{\rm d}x}{x}.\]
    The theorem follows from the definition of a holomorphic function and differentiation under the integral sign \cite[Chapter XVII, Theorem 8.2]{lang2013undergraduate}.
\end{proof}

Our goal is to derive a system of differential equations satisfied by ${\cal I}_{\Gamma}(z)$. This is an example of a class of such systems, called \emph{GKZ systems} (after Gelfand, Kapranov and Zelevinsky) or \emph{$A$-hypergeometric systems}. We introduce these in general, and then specialize to our integrals. 
The proof of the main theorem in this section (Theorem \ref{thm:maingkz}) uses the theory of \emph{$D$-modules}. Here the ring $D$ is the \emph{Weyl algebra}, which plays an analogous role as that of the ring $R$ of difference operators (see Section \ref{sec:difference}). To state Theorem \ref{thm:maingkz}, it is unnecessary to introduce $D$-modules. We refer the interested reader to \cite{coutinho1995primer} for a nice introduction.

Let $d$ be a positive integer and let $A\subset\mathbb{Z}^d$ be a finite subset.
To each $\alpha\in A$, we associate a complex variable $z_\alpha$ and a partial derivative operator $\partial_\alpha =\frac{\partial}{\partial z_\alpha}$.
For a function $f(z)$ of $z=(z_\alpha)_{\alpha\in A}$, its partial derivative $\frac{\partial f}{\partial z_\alpha}(z)$ with respect to $z_\alpha$ is denoted by $\partial_\alpha f(z)$.
The {\it toric ideal} $I_A\subset\mathbb{C}[\partial_\alpha,\, \alpha\in A]$ is an ideal generated by all binomials 
\[
\prod_{\alpha\in A}\partial_\alpha^{u_\alpha}-\prod_{\alpha\in A}\partial_\alpha^{v_\alpha},
\]
where $u=(u_\alpha)_{\alpha\in A},\, v=(v_\alpha)_{\alpha\in A}\in\mathbb{N}^A$ are such that $A \cdot (u-v) = 0$. That is,
\begin{equation}\label{eq:toric_constraint}
 \sum_{\alpha\in A}u_\alpha\alpha=\sum_{\alpha\in A}v_\alpha\alpha.
\end{equation}
Of course, there are infinitely many integer vectors $u-v$ satisfying $A \cdot (u-v)$, but finitely many suffice to generate the ideal $I_A$. Fix a vector $\beta\in\mathbb{C}^d$. The \emph{GKZ system} $H_A(\beta)$ associated to $A$ and $\beta$ is the following system of partial differential equations in $f(z)$: 
\begin{equation}\label{eq:GKZ}
H_A(\beta): \quad \sum_{\alpha\in A}z_\alpha\partial_\alpha\bullet f(z)\alpha- f(z)\beta \, = \, 0 \quad \text{and} \quad 
    P(\partial) \bullet f(z)\, = \, 0 \quad \text{ for } P(\partial)\in I_A.
\end{equation}
Note that the first equation of \eqref{eq:GKZ} is an identity of vectors in $\mathbb{C}^d$, and it is enough to check that $P(\partial)f = 0$ for a finite set of generators of $I_A$.
On an open subset $U\subset \mathbb{C}^A$, we define the space of solutions $\sol_{H_A(\beta)}(U)$ of $H_A(\beta)$ as the complex vector space 
\[\sol_{H_A(\beta)}(U) \, = \, \{ \,  f : U \rightarrow \mathbb{C} \text{ holomorphic } \, : \, f \text{ satisfies \eqref{eq:GKZ} } \}. \]
Our Euler integral \eqref{eq:I_Gamma} satisfies the GKZ system specified by the following parameters. Set $d = \ell + n$.
The \emph{Cayley configuration} $A \subset \mathbb{Z}^{d}$ of $A_1, \ldots, A_\ell$ is
    \begin{equation} \label{eq:cayley2}
        A \, = \, \{ (e_i,\alpha) \, :\, \alpha \in A_i, \,  i = 1, \ldots, \ell \} \quad \subset \mathbb{Z}^{d}.
    \end{equation} 
We have seen this in \eqref{eq:CayleyConfiguration}. The vector $\beta$ is $-(s,\nu) \in \mathbb{C}^d$.

\begin{proposition} \label{prop:gkz}
For $A, \beta$ as above, and $I_{\Gamma}(z)$ as in \eqref{eq:I_Gamma}, we have ${\cal I}_{\Gamma}(z) \in \sol_{H_A(\beta)}(U)$.
 \end{proposition}
 Before proving Proposition \ref{prop:gkz}, we encourage the reader to check that \eqref{eq:diffbeta} is the GKZ system for the beta integral. In that example, the toric ideal $I_A$ is 0. Notice that Proposition \ref{prop:gkz} is independent of the choice of cycle $\Gamma(z)$.
\begin{proof}[Proof of Proposition \ref{prop:gkz}]
    By the definition of the Cayley configuration \eqref{eq:cayley2}, the constraint \eqref{eq:toric_constraint} for $u=\left(u_{i,\alpha}\right)_{\substack{i=1,\dots,\ell\\ \alpha\in A_i}},v=(v_{i,\alpha})_{\substack{i=1,\dots,\ell\\ \alpha\in A_i}}$ takes the following form :
    \begin{equation}\label{eq:76}
        u_i =\sum_{\alpha\in A_i}u_{i,\alpha}=\sum_{\alpha\in A_i}v_{i,\alpha}\,\, (i=1,\dots,\ell)\,\,\,\,\text{and}\,\,\,\, \sum_{i=1}^\ell\sum_{\alpha\in A_i}u_{i,\alpha}\alpha \, = \, \sum_{i=1}^\ell\sum_{\alpha\in A_i}v_{i,\alpha}\alpha. 
    \end{equation}
    By our construction of $\Gamma(z)$ in \eqref{eq:gammaofz}, the integration contours $\Delta_p \subset X$ are independent of $z$, and we may apply the operators $\partial_{(e_i,\alpha)}$ to ${\cal I}_\Gamma(z)$ by differentiating under the integration sign. The sections $\tau_p(z): \Delta_p \rightarrow \mathbb{C}$ differentiate to the corresponding branches of  
    \begin{equation}\label{eq:77}
        \partial_{(e_i,\alpha)}\bullet f^{-s}=-s_if_1^{-s_1}\cdots f_i^{-s_i-1}\cdots f_\ell^{-s_\ell}x^{\alpha}.
    \end{equation} 
    The operators generating $I_A$ are $\prod_{i=1}^\ell \prod_{\alpha\in A_i}\partial_{(e_i,\alpha)}^{u_{i,\alpha}} - \prod_{i=1}^\ell \prod_{\alpha\in A_i}\partial_{(e_i,\alpha)}^{v_{i,\alpha}}$. We calculate
    \[
\prod_{i=1}^\ell
\prod_{\alpha\in A_i}\partial_{(e_i,\alpha)}^{u_{i,\alpha}}\bullet \, {\cal I}_\Gamma(z)
\, = \, \prod_{i=1}^\ell\prod_{j=0}^{u_i-1}(-s_i-j)\int_{\Gamma(z)}f_1^{-s_1-u_1}\cdots f_\ell^{-s_\ell-u_\ell}x^{\nu+\sum_{i=1}^\ell\sum_{\alpha\in A_i}u_{i,\alpha}\alpha}\frac{{\rm d}x}{x}
.
    \]
    Doing the same for $\prod_{i=1}^\ell
\prod_{\alpha\in A_i}\partial_{(e_i,\alpha)}^{v_{i,\alpha}}$ and applying \eqref{eq:76} we see that $I_A$ annihilates ${\cal I}_\Gamma(z)$. 

    It remains to verify that the other operators in the GKZ system annihilate ${\cal I}_\Gamma(z)$ as well: 
    \begin{equation} \label{eq:gkzagain}
    \sum_{\alpha\in A}z_\alpha\partial_\alpha\bullet {\cal I}_\Gamma(z)\alpha- {\cal I}_\Gamma(z)\beta \, = \, 0.
    \end{equation}
    The first $\ell$-entries come from the homogeneity relation
    \[
    -s_if_i^{-s_i} \, = \, -s_i\left(\sum_{\alpha\in A_i}z_{i,\alpha}x^\alpha\right)f_i^{-s_i-1} \, = \, \sum_{\alpha\in A_i}z_{i,\alpha}\partial_{(e_i,\alpha)}\bullet f_i^{-s_i}.
    \]
    For any $j=1,\dots,n$, the $(\ell+j)$-th entry of \eqref{eq:gkzagain} is derived by differentiating under the integral sign and observing that the result is the pairing of $\Gamma(z)$ with the exact $n$-form
    \[
    \nabla_\omega\left( (-1)^{j-1}\frac{{\rm d}x_1\wedge\cdots {\rm d}x_{j-1}\wedge {\rm d}x_{j+1}\cdots \wedge {\rm d}x_n}{x_1\cdots x_{j-1}x_{j+1}\cdots x_n}\right)
    =-\left(\sum_{i=1}^\ell s_if_i^{-1}\sum_{\alpha\in A_i}z_{i,\alpha}x^\alpha \alpha_j+\nu_j\right)\frac{{\rm d} x}{x}.
    \]
    Here, $\alpha_j$ is the $j$-th entry of $\alpha$. This generalizes what happened for $P_2$ in Example \ref{ex:4.2}.
\end{proof}

Proposition \ref{prop:gkz} implies the following homogeneity property for the function ${\cal I}_\Gamma(z)$. 

\begin{lemma} \label{lem:homogeneity}
    Let $A = \{ \alpha_1, \ldots, \alpha_N \} \subset \mathbb{Z}^d$, $\beta \in \mathbb{C}^d$ and let $U\subset \mathbb{C}^N$ be an open subset. If $f: U \rightarrow \mathbb{C}$ is holomorphic, then $(A\theta-\beta)\bullet f(z)=0$ if and only if, for all $z \in U$ and for all $u \in (\mathbb{C}^*)^d$ such that $(u^{\alpha_1} z_1, \ldots, u^{\alpha_N} z_N) \in U$, we have 
    \begin{equation}\label{eq:94}
        f(u^{\alpha_1} z_1, \ldots, u^{\alpha_N} z_N)\, = \, u^{\beta} \, f(z).
    \end{equation}
\end{lemma}

\begin{proof}
    Suppose $f(z)$ satisfies \eqref{eq:94}. We fix any $1\leq i\leq d$.
    Taking the derivative of \eqref{eq:94} with respect to $u_i$ and substituting $u=(1,\dots,1)$, we obtain the $i$-th entry of the identity $(A\theta-\beta)\bullet f(z)=0$. For the other direction, suppose $f(z)$ is annihilated by $A\theta-\beta$.
    To prove \eqref{eq:94}, it is enough to prove it for $u(i)=u_i \cdot e_i$ with $u_i\in\mathbb{C}^*$, and $e_i$ the standard basis vector.
    For any $z\in U$, the functions $\phi_1: u_i \mapsto f(u(i)^{\alpha_1}z_1,\dots,u(i)^{\alpha_N}z_N)$ and $\phi_2:u_i \mapsto u(i)^\beta f(z)$ are both annihilated by $u_i\frac{\partial}{\partial u_i}-\beta_i$.
    We also have $\phi_1(1) = \phi_2(1)$.
    Unique solvability of the initial value problem of an O.D.E. implies $\phi_1 = \phi_2$, so \eqref{eq:94} holds for $u = u(i)$.
\end{proof}

Often, in applications, one is interested in differential equations satisfied by ${\cal I}_\Gamma(z)$ after specializing the parameters $z$. This is the case for Feynman integrals, in which the coefficients $z_{i,\alpha}$ depend linearly on the \emph{Mandelstam invariants}. For the case of Example \ref{ex:0.2}, we now show how the equations in these new variables can be derived from the GKZ system.

\begin{example}
    We revisit the Euler integral \eqref{eq:Itriangle} of Example \ref{ex:triangle}. The differential equations satisfied by this integral as a function of $t_1, t_2, t_3$ are the running example of \cite{henn2023d}. In this case, they can be derived in an easy way from the GKZ differential equations for 
    \[{\cal I}_{\Gamma}(z) \, = \, \int_{\Gamma} \frac{x_1^{\nu_1} x_2^{\nu_2} x_3^{\nu_3}}{( z_1 \, x_1 + z_2 \, x_2 + z_3 \, x_3 + z_4  \, x_2 x_3 + z_5 \, x_1 x_3 + z_6 \,  x_1 x_2)^s} \, \frac{{\rm d}x_1{\rm d}x_2{\rm d}x_3}{x_1x_2x_3}. \]
    Here $\ell=1$, $n=3$, $\beta=-(s,\nu_1,\nu_2,\nu_3)$ and $A$ consists of the columns of
    \[
    \begin{pmatrix}
        1&1&1&1&1&1\\
        1&0&0&0&1&1\\
        0&1&0&1&0&1\\
        0&0&1&1&1&0
    \end{pmatrix}
    .
    \]
    We denote these columns by  $\alpha_1,\dots,\alpha_6 \in \mathbb{Z}^4$, and set $\partial_i=\partial_{\alpha_i}$ for brevity.
    The toric ideal $I_A$ is generated by $\Tilde{P}_1=\partial_{1}\partial_{4}-\partial_{3}\partial_{6}$ and $\Tilde{P}_2=\partial_{2}\partial_{5}-\partial_{3}\partial_{6}$.
    These operators can be found using any computer algebra software. For instance, in \texttt{Macaulay2} \cite{M2}, the commands are
    \begin{minted}{julia}
needsPackage "Quasidegrees"
A = matrix{{1,1,1,1,1,1},{1,0,0,0,1,1},{0,1,0,1,0,1},{0,0,1,1,1,0}}
D = QQ[d_1..d_6]
T = toricIdeal(A,D)
    \end{minted}
    The remaining operators constituting the GKZ system are
    \begin{equation} \label{eq:euler} 
    \textstyle \sum_{k=1}^6 \theta_k + s,  \quad   \theta_1 + \theta_5 + \theta_6+ \nu_1, \quad 
    \theta_2 + \theta_4 + \theta_6 + \nu_2, \quad 
    \theta_3 + \theta_4 + \theta_5 + \nu_3,
    \end{equation}
    where $\theta_k = z_k \partial_k$.
    Up to replacing $\mathbb{R}^3_+$ with $\Gamma$, the integral ${\cal I}_G$ in \eqref{eq:Itriangle} is ${\cal I}_{\Gamma}(1,1,1,-t_1,-t_2,-t_3)$. Notice that here one could take $\Gamma$ to be the \emph{regularization} of $\mathbb{R}^3_+$, see Theorem* \ref{thm:backtosec1}.
    To turn our operators in $z_1, \ldots, z_6$ into differential operators in $t_1,t_2,t_3$, we use the homogeneity condition from Lemma \ref{lem:homogeneity}:
    \begin{equation*}
{\cal I}_\Gamma(u^{\alpha_1}z_1,\dots,u^{\alpha_6}z_6)\, = \, u^{\beta} \, {\cal I}_\Gamma(z).
    \end{equation*}
    We adopted the usual notation $u^\beta=u_1^{-s}u_2^{-\nu_1}u_3^{-\nu_2}u_4^{-\nu_3}$.
    For $u=(1,z_1^{-1},z_2^{-1},z_3^{-1})$ this reads
    \begin{equation}\label{eq:homogeneity}
        {\cal I}_\Gamma\left(1,1,1,\frac{z_4}{z_2z_3},\frac{z_5}{z_1z_3},\frac{z_6}{z_1z_2}\right)
        \, = \,
    z_1^{\nu_1}z_2^{\nu_2}z_3^{\nu_3}{\cal I}_\Gamma(z).
    \end{equation}
    We differentiate \eqref{eq:homogeneity} with respect to $z_1$, and \emph{afterwards} we substitute $z = (\mathbbm{1},-t) = (1,1,1,-t_1,-t_2,-t_3)$. Using $(\partial_5{\cal I}_\Gamma)(\mathbbm{1},-t) = - \partial_{t_2}({\cal I}_\Gamma(\mathbbm{1},-t))$ and similarly for $\partial_6{\cal I}_\Gamma$, we~get
    \begin{align*}
    (\partial_1 {\cal I}_\Gamma)(\mathbbm{1},-t)&=-\left(\nu_1+t_2\partial_{t_2}+t_3\partial_{t_3}\right)\bullet {\cal I}_\Gamma(\mathbbm{1},-t)\\
    (\partial_2 {\cal I}_\Gamma)(\mathbbm{1},-t)&=-\left(\nu_2+t_1\partial_{t_1}+t_3\partial_{t_3}\right)\bullet {\cal I}_\Gamma(\mathbbm{1},-t)\\
    (\partial_3 {\cal I}_\Gamma)(\mathbbm{1},-t)&=-\left(\nu_3+t_1\partial_{t_1}+t_2\partial_{t_2}\right)\bullet {\cal I}_\Gamma(\mathbbm{1},-t).    
    \end{align*}
    We now eliminate $\partial_1{\cal I}_{\Gamma},\partial_2{\cal I}_{\Gamma},\partial_3{\cal I}_{\Gamma}$ in terms of $\partial_{t_1},\partial_{t_2},\partial_{t_3}$.
    E.g., the first equation of~\eqref{eq:GKZ} is
    \[
    \Tilde{P}_3\bullet {\cal I}_{\Gamma}(z) =\left( z_1\partial_1+\cdots+z_6\partial_6+s\right)\bullet {\cal I}_{\Gamma}(z)=0.
    \]
    This is equivalent to the following relation for ${\cal I}_\Gamma(\mathbbm{1},-t)= {\cal I}_\Gamma(1,1,1,-t_1,-t_2,-t_3)$:
    \begin{equation}\label{eq:97}
P_3\bullet {\cal I}_\Gamma(\mathbbm{1},-t)=0,\,\,\,P_3 =t_1\partial_{t_1}+t_2\partial_{t_2}+t_3\partial_{t_3}+\nu_1+\nu_2+\nu_3-s.
    \end{equation}
    Likewise, we obtain a relation 
    \begin{align}
        (\partial_4\partial_1{\cal I}_\Gamma)(\mathbbm{1},-t)& \, = \, -\partial_{t_1}\bullet(\partial_1{\cal I}_\Gamma)(\mathbbm{1},-t) \, = \, \partial_{t_1}\bullet\left(\nu_1+t_2\partial_{t_2}+t_3\partial_{t_3}\right)\bullet {\cal I}_\Gamma(\mathbbm{1},-t)\label{eq:98}\\
        & \, =-\partial_{t_1}\bullet\left(\nu_2+\nu_3-s+t_1\partial_{t_1}\right)\bullet {\cal I}_\Gamma(\mathbbm{1},-t)\label{eq:99}.
    \end{align}
    Here, we used \eqref{eq:97} when passing from \eqref{eq:98} to \eqref{eq:99}.
    Repeating this for the monomials in $\tilde{P}_1$, $\tilde{P}_2$, we find that ${\cal I}_{\Gamma}(1,1,1,-t_1,-t_2,-t_3)$ is annihilated by the following operators:
    \begin{align*}
P_1&=t_1\partial_{t_1}^2-t_3\partial_{t_3}^2+(1-s+\nu_2+\nu_3)\partial_{t_1}-(1-s+\nu_1+\nu_2)\partial_{t_3}\\
P_2&=t_2\partial_{t_2}^2-t_3\partial_{t_3}^2+(1-s+\nu_1+\nu_3)\partial_{t_2}-(1-s+\nu_1+\nu_2)\partial_{t_3}\\
P_3&=t_1\partial_{t_1}+t_2\partial_{t_2}+t_3\partial_{t_3}+\nu_1+\nu_2+\nu_3-s.
    \end{align*}
    These operators agree with the ones in \cite[Equation (2.6)]{henn2023d} after setting $s={\rm D}/2$.
\end{example}

The \emph{local solutions} of $H_A(\beta)$ at a point $z^*\in\mathbb{C}^A$ are given by the direct limit
\[
\sol_{H_A(\beta),z^*} \, = \, \lim_{\substack{\longrightarrow \\ z^*\in U}}\sol_{H_A(\beta)}(U).
\]
Elements of $\sol_{H_A(\beta),z^*}$ are represented by holomorphic solutions of $H_{A}(\beta)$, defined on a sufficiently small open neighborhood $U$ of $z^*$. Theorem \ref{thm:maingkz} describes it in terms of integrals. 
For a facet $Q$ of the polyhedral cone ${\rm pos}(A)$, we write $r_Q$ for the primitive ray generator of the dual ray $\{ y \in (\mathbb{R}^d)^\vee \, :\, y \cdot q \geq 0 \text{ for all } q \in Q \}$.
A complex vector $\beta\in\mathbb{C}^{n+\ell}$ is said to be {\it non-resonant} if $r_Q\cdot \beta\notin\mathbb{Z}$ for any facet $Q$ of ${\rm pos}(A)$.
The following is \cite[Theorem 2.10]{gelfand1990generalized}:
 \begin{theorem} \label{thm:maingkz}
     Let $\beta=-(s,\nu)\in\mathbb{C}^{\ell + n}$ be non-resonant.
     For any $z^*\in \mathbb{C}^A$, the map $H_n(X_{z^*},-\omega(z^*)) \rightarrow \sol_{H_A(\beta),z^*}$ given by $[\Gamma(z^*)]\mapsto{\cal I}_\Gamma(z) $ is a vector~space isomorphism.
 \end{theorem}

\begin{remark}
    Theorem \ref{thm:maingkz} implies that, for generic $s,\nu$ (in the sense of the Vanishing Theorem \ref{thm:3.3}), the dimension of the local solution space $\sol_{H_A(\beta),z^*}$ equals the signed Euler characteristic of $X_{z^*}$. For $z^*$ outside an algebraic hypersurface $\{E_A = 0\} \subset \mathbb{C}^A$, this number equals the normalized volume of the convex hull of $A$ \cite[Theorem 5.15]{Adolphson}. The polynomial $E_A$ is the \emph{principal $A$-determinant}, as introduced by Gelfand, Kapranov and Zelevinsky \cite[Chapter 10]{gelfand2008discriminants}. 
\end{remark}
 
While the function ${\cal I}_{\Gamma}$ is an integral over $\Gamma(z)$ against a {\it particular} cohomology class $[\frac{dx}{x}]$, the pairing $\langle\Gamma(z),\phi\rangle$ is well-defined for {\it any} twisted cocycle $[\phi]\in H^n(X_z,\omega)$.
Let $[\phi_1], \ldots, [\phi_\chi]$ be a basis for the twisted cohomology $H^n(X_z,\omega)$ for $z \in U$. Again, by Corollary \ref{cor:dimchi}, there are $\chi = (-1)^n \cdot \chi(X_{z^*})$ basis elements. 
We assume the $[\phi_k]$ have coefficients that are rational functions in $z$.
There exist $\chi \times \chi$-matrices $P_{\alpha}$ $(\alpha\in A)$ such that 
\begin{equation}\label{eq:Pfaffian}
\partial_{\alpha} \bullet \begin{pmatrix}
    \langle \Gamma(z), \phi_1 \rangle \\ \vdots \\ \langle \Gamma(z), \phi_\chi \rangle
\end{pmatrix} \, = \, P_\alpha \cdot \begin{pmatrix}
    \langle \Gamma(z), \phi_1 \rangle \\ \vdots \\ \langle \Gamma(z), \phi_\chi \rangle
\end{pmatrix}.
\end{equation}
Here $\partial_\alpha$ acts entry-wise on vectors. These expressions form the so-called {\it Pfaffian system}. The Pfaffian system can be derived from a system of differential operators, like a GKZ system. The general procedure is explained in \cite[Section 3]{chestnov2022macaulay}.

\begin{remark}
    Pfaffian systems lead to one of the most efficient ways of evaluating Feynman integrals \cite{Henn:2013pwa}. In practice, \eqref{eq:Pfaffian} can be solved by providing boundary conditions $\langle \Gamma(z^\ast), \phi_i \rangle$ for $i=1,\ldots,\chi$ at some $z=z^\ast$ and using path-ordered exponentiation of the matrices $C_\alpha$ to evaluate the Pfaffian system at other values of $z$. See \cite{Henn:2014qga} for a pedagogical introduction.
\end{remark}

\section{Open problems} \label{sec:5}

The previous sections provide an overview of the basics of Euler integrals. While this is a classical topic, the theory is currently still very much in development. We conclude with a list of open research problems, hoping that the reader will join this effort.  

\begin{enumerate}[leftmargin=*]
    \item \textbf{Evaluating integrals numerically.} When the integral \eqref{eq:Eulermellin} converges, it can be evaluated numerically using sector decomposition and Monte Carlo integration, see Remark \ref{rem:tropical}. When the real part of the exponents is large, it becomes increasingly important to concentrate the Monte Carlo samples in the close neighborhood of the critical point $a$ from Section \ref{sec:2.3}. This could lead to effective numerical algorithms for evaluating convergent Euler integrals, with applications in Bayesian statistics \cite{borinsky2023bayesian}. Likewise, in physics applications one often has to analytically continue in the parameters $(s,\nu)$ before numerical evaluation, which can be achieved using the results of Sections~\ref{sec:3.3} and \ref{sec:difference}.

    \item \textbf{Generic parameters.} Theorems \ref{thm:huh} and \ref{thm:3.3} make genericity assumptions on the exponents $s, \nu$. In Theorem \ref{thm:3.3} that means \emph{outside a closed algebraic subvariety}, in Theorem \ref{thm:3.3} it means \emph{outside a countable union of hyperplanes}. The former can be seen as a limit of the latter, by driving the parameter $\delta$ from Section \ref{sec:2} to zero \cite{matsubara2023twisted}. It is interesting to describe these hyperplanes explicitly, and investigate this problem in more detail. 

    \item \textbf{Non-generic parameters.}
    In physics applications, one often encounters Euler integrals with \emph{special} coefficients and parameters. In particular, these are \emph{not} generic in the sense of Theorems \ref{thm:huh} and \ref{thm:3.3}. It would be interesting to develop the analogous theory applicable to such cases, perhaps along the lines of \cite{Caron-Huot:2021xqj,matsumoto2019relative}.

    \item \textbf{Regularized integration cycles.} As mentioned in Section \ref{sec:3}, a rigorous proof of Theorem* \ref{thm:backtosec1} for general Euler integrals is currently still missing. 

    \item \textbf{Nice bases of cohomology.} There are several reasons for which it is favorable to use basis element for cohomology which are represented by \emph{${\rm dlog}$ forms} \cite{saito1980theory}. These are regular $n$-forms obtained as ${\rm dlog}$ of a rational function. Another notion of a nice basis is related to so-called \emph{canonical} differentials equations for Feynman integrals \cite{Henn:2013pwa}. In both cases, it would be interesting to find criteria for such bases to exist.

    \item \textbf{Beyond Euler integrals.} While our framework deals with Euler integral defined by \eqref{eq:integralintro}, there are other types of integrals that resemble it.
    The list includes {\it exponential integrals} \cite{fresan2020quadratic,majima2000quadratic}, {\it matrix hypergeometric integrals} \cite{hashiguchi2013holonomic}, and  integrals over ${\cal M}_{g,n}$ \cite{Felder:1995iv,Witten:2012bh}.
    Theorems stated in this article mostly remain unsolved for these integrals.

    \item \textbf{Intersection pairing.}
    The \emph{intersection pairing} is a canonically defined operation on a twisted cohomology, which can be used to reduce twisted cocycles to a basis \cite{cho1995intersection,matsumoto1998intersection,Mizera:2017rqa}. Efficient evaluation of intersection pairing remains an important computational challenge.

    \item \textbf{$\chi$-Stratification.} In Section \ref{sec:differential}, the very affine variety $X_z$ depends on the coefficients $z$ of the Laurent polynomials. We propose to study the loci in coefficient space on which the Euler characteristic $\chi(X_z)$ is constant. E.g., for which $z \in (\mathbb{C}^*)^A$ is $|\chi(X_z)|$ minimal?
\end{enumerate}

\section*{Acknowledgements}
This article served as accompanying notes for a series of four lectures given by the third named author at the Max Planck Institute for Mathematics in the Sciences in Leipzig. We thank J\"org Lehnert for organizing these lectures, and
Raluca Vlad for pointing out several typos in a previous version. 

This material is based upon work supported by the Sivian Fund and the U.S. Department of Energy, Office of Science, Office of High Energy Physics under Award Number DE-SC0009988.
It is also supported by JSPS KAKENHI Grant Number 22K13930 and partially supported by JST CREST Grant Number JP19209317.

\addcontentsline{toc}{section}{References}

\bibliographystyle{JHEPalph}
\bibliography{references.bib}

\subsection*{Authors' addresses:}

\noindent Saiei-Jaeyeong Matsubara-Heo, Kumamoto University \hfill \href{mailto:saiei@educ.kumamoto-u.ac.jp}{\tt saiei@educ.kumamoto-u.ac.jp}

\noindent Sebastian Mizera, IAS Princeton \hfill \href{mailto:smizera@ias.edu}{\tt smizera@ias.edu}

\noindent Simon Telen, MPI-MiS Leipzig \hfill \href{mailto:simon.telen@mis.mpg.de}{\tt simon.telen@mis.mpg.de}

\end{document}